\documentclass[a4paper,
draft,
11pt,openright]{article} 
\usepackage[italian, english]{babel} 
\usepackage[T1]{fontenc}
\usepackage[utf8]{inputenc}

\usepackage{url}
\usepackage{mathtools}
\usepackage[
final=true,                    
plainpages=false,              
pdfpagelabels=true,            
pdfencoding=auto,              
unicode=true,                  
hypertexnames=true,            
naturalnames=true              
]{hyperref}
\usepackage[autostyle, english = american]{csquotes}
\usepackage{appendix}
\usepackage{enumerate}
\usepackage{libertine}
\usepackage{bbold}
\usepackage{mathtools}
\usepackage{graphicx}
\usepackage{floatflt}
\usepackage{blindtext}
\usepackage{enumitem}
\usepackage{physics}
\usepackage{bm}
\usepackage{amsmath}
\usepackage{siunitx}
\usepackage{mathptmx}
\usepackage{amsthm}
\usepackage{subfig}
\usepackage{cancel}
\usepackage{xcolor}
\usepackage{listings}
\usepackage{amsmath}
\usepackage{amssymb}
\usepackage{framed}
\usepackage{wrapfig}
\usepackage{minibox}
\usepackage{float}
\usepackage{wrapfig}
\usepackage{longtable}
\usepackage{ulem}
\usepackage[strict]{changepage}
\usepackage{pgfplots}
\usepackage{tikz}
\usepackage{authblk}
\usepackage{urwchancal}
\usepackage[all,cmtip]{xy}
\usepackage{lineno}
\usetikzlibrary{matrix,calc}

\usetikzlibrary{positioning}
\usetikzlibrary{matrix}
\usepackage{pgfplots}
\usetikzlibrary{patterns}

\numberwithin{equation}{section}

\newcommand{\weylSH}{\mathcal{W}(\mathcal{S}(\mathbb{R}^{\nu}), h\sigma)}
\newcommand{\weylSO}{\mathcal{W}(\mathcal{S}(\mathbb{R}^{\nu}),0)}

\newcommand{\weylEO}{\mathcal{W}(E,0)}

\newtheoremstyle{TheoremStyle}
{\topsep}%
{\topsep}
{}
{}
{\sc}
{:}
{.5em}
{}

\makeatletter
\def\@endtheorem{\begin{flushright}$\diamond$\end{flushright}} 
\makeatother

\theoremstyle{TheoremStyle}
\newtheorem{theorem}{Theorem}[section]
\newtheorem{corollary}[theorem]{Corollary}

\newtheorem{proposition}[theorem]{Proposition}
\newtheorem{lemma}[theorem]{Lemma}
\newtheorem{definition}[theorem]{Definition}
\newtheorem{remark}[theorem]{Remark}
\newtheorem{example}[theorem]{Example} 
\newtheorem*{assumption}{Assumption}

\usepackage{xcolor}
\usepackage{listings}
\definecolor{codegreen}{rgb}{0,0.6,0}
\definecolor{codegray}{rgb}{0.5,0.5,0.5}
\definecolor{codepurple}{rgb}{0.58,0,0.82}
\definecolor{backcolour}{rgb}{0.95,0.95,0.92}

\lstdefinestyle{mystyle}{
	backgroundcolor=\color{backcolour},   
	commentstyle=\color{codegreen},
	keywordstyle=\color{magenta},
	numberstyle=\tiny\color{codegray},
	stringstyle=\color{codepurple},
	basicstyle=\ttfamily\footnotesize,
	breakatwhitespace=false,         
	breaklines=true,                 
	captionpos=b,                    
	keepspaces=true,                 
	numbers=left,                    
	numbersep=5pt,                  
	showspaces=false,                
	showstringspaces=false,
	showtabs=false,                  
	tabsize=2
}
\definecolor{someblu}{RGB}{84,120,165}

\usepackage[multiuser]{fixme}
\FXRegisterAuthor{lp}{anlp}{Lo}

\title{On Classical Aspects of Bose-Einstein Condensation}
\author[a]{\href{mailto:lorenzo.pettinari@unitn.it}{L. Pettinari}}
\affil[a]{Dipartimento di Matematica, Universit\`a di Trento and INFN-TIFPA and INdAM, Via Sommarive 14, I-38123 Povo, Italy}
\begin{document}
	
	\maketitle
	\begin{abstract}
	  Berezin and Weyl quantization are renown procedures for mapping, commutative Poisson algebras of observables to their non-commutative, quantum counterparts. The latter is famous for its use on Weyl algebras, while the former is more appropriate for continuous functions decaying at infinity. In this work, we define a variant of the Berezin quantization map, which acts on the classical Weyl algebra $\mathcal{W}(E,0)$ and constitutes a positive \textit{strict deformation quantization}. This construction provides a natural framework to compare classical and quantum thermal equilibrium states of a Bose gas through the computation of their semi-classical limit.
      To this end, we first introduce a purely algebraic notion of KMS states for the classical Weyl algebra and establish that, in finite volume, there exists a unique such state, which can be interpreted as the Fourier transform of a Gibbs measure on a Hilbert space. We then construct a new class of classical KMS states that realize representations of the canonical commutation relations with infinite local density. These states arise as the semi-classical high-density limit of the quantum equilibrium states originally studied by Araki and Woods \cite{Araki_Woods_63}. A key feature of our approach is that it preserves the macroscopic ground-state occupation of the Bose gas in the classical regime. Finally, we demonstrate that the infinite-volume classical states can be obtained as thermodynamic limits of finite-volume Gibbs states.
	\end{abstract}
	\tableofcontents
	
	\section{Introduction}\label{sec: introduction}
   A standard paradigm of quantum statistical mechanics in the framework of operator algebras is the use of KMS conditions for the description of thermal equilibrium states \cite{Bratteli_Robinson_97,Drago_Pettinari_VandeVen_2024, Kastler_65}. Within this setting, the physical observables are modeled by suitable $C^*$-algebras. Time evolution is introduced by means of a strongly continuous one-parameter group of $^*$-automorphism $t\rightarrow \tau_{t}$. \textit{States} of the system are described by linear, positive, normalized, functionals $\omega\colon \mathfrak{A}\to \mathbb{C}$. Once the dynamics and an inverse temperature $\beta$ are fixed, a state $\omega$ is said to satisfy the $(\tau,\beta)$-KMS condition for the $C^*$-dynamical system $(\mathfrak{A},\tau)$ if there exists a strongly dense $^*$-subalgebra $\mathfrak{A}_{\tau}$ of $\mathfrak{A}$, contained in the set of the $\tau$-analytic elements, for which the following identity is verified \begin{equation}\label{eq: KMS quantum}
   	\omega(\mathfrak{a}\tau_{i\beta h}(\mathfrak{b})) = \omega(\mathfrak{b}\mathfrak{a}), \quad \mathfrak{a},\mathfrak{b}\in \mathfrak{A}_{\tau}. 
   \end{equation} Moreover, the foregoing framework can be reformulated for the von Neumann algebras setting, using $\sigma$-weakly continuous groups and normal states.
   
   Although the quantum KMS condition is widely used by the algebraic community, its classical counterpart, introduced in \cite{Gallavotti_Verboven_75}, has received little attention over the years. More recently, the latter has been employed in different contexts, ranging from Poisson geometry \cite{Drago_Waldman_2024} to the study of classical, nonlinear Hamiltonian systems \cite{Ammari_Sohinger_2023} and for an infinite ensemble of particles \cite{Aizenman_Gallavotti_Goldstein_Lebowitz_1976, Aizenman_Goldstein_Gruber_Lebowitz_Martin_1977}.
   The classical KMS property has also been investigated in relation to the Dobrushin-Landford-Ruelle (DLR) condition for classical equilibrium \cite{Drago_VandeVen_2023}, while in \cite{Drago_Pettinari_VandeVen_2024} it has been justified using \textit{semi-classical techniques}.A semi-classical regime refers to a physical setting in which the behavior of a quantum system is governed by a parameter $h$—the semi-classical parameter—whose magnitude is small (or, in some contexts, large) relative to the other characteristic scales of the system. In such regimes, quantum effects remain relevant but are subdominant, allowing for a meaningful approximation of quantum dynamics by classical theories as for example in the WKB approximation of Schr\"odinger equations \cite{Landsman_2017,Maslov_81}. This transition is formalized through the study of semi-classical limits, where one analyzes how observables and states behave as $h \to 0$. In our approach, a family of quantum states $(\omega_h)_{h\in(0,+\infty)}$ is used to define classical ones by composition with a quantization map (see Def. \ref{def: SDQ} ) \begin{equation}\omega_{cl, h} :=\omega_h\circ Q_h.\end{equation} Under suitable hypothesis, the limit $h \to 0^+$ yields an $h$-independent classical state $\omega_{cl}$, which preserves physical properties of the original quantum states in the semi-classical regime, proving how classical behavior can emerge from a quantum theory. A particularly important aspect of this correspondence concerns the persistence of thermodynamic phases. It is widely expected that such phases should be preserved in the semi-classical limit, as quantum fluctuations become negligible compared to thermal ones. This intuition is supported by recent results in \cite{Drago_Pettinari_VandeVen_2024}, where it was shown that, at high (but finite) temperatures, classical and quantum spin systems exhibit the same unique thermal phase for a broad class of interaction potentials. 
 
The scope of this article is to construct a semi-classical description for a free Bose gas by defining a class of equilibrium classical states $\omega^{\alpha,0}_0$, describing an high density (weak$^*$) limit of a corresponding class of quantum non-unique equilibrium states $\omega^{\overline{\rho}(h),0}_h$.  It is our hope that these results could be extended to interacting settings, where the high density regime was initially meant for \cite{Bog_47} and produced many rigorous results \cite{GiulSeir2009,Lieb_Solovej_2001,Lieb_Solovej_2004}.

   We will model the observables of such system by the Weyl $C^*$-algebras $\mathcal{W}(E,h\sigma )$, where $E$ is a symplectic space, and $h\sigma$ denotes the corresponding symplectic form, with $h$ a suitable semi-classical parameter.
    Weyl algebras with non-degenerate symplectic spaces $(E,h\sigma)$ represent the first examples of  $C^*$-algebras modelling the \textit{canonical commutation relations} (CCR) of quantum mechanics. They have been widely explored in the literature \cite{Kastler_65,Landsman_1998,Landsman_2017,Loupias_Miracle-Sole_67,Manuceau_68,Petz_90} and are known to present many drawbacks when modeling realistic physical systems. Firstly, because even the simplest free dynamics is not continuous in this setting \cite{Bratteli_Robinson_97, Kastler_65,Petz_90} and secondly, because Weyl algebras are not left invariant by any interesting, interacting dynamics \cite{Fannes_Verbeure_(74)}. However, they provide a strong and simple framework for the discussion of phase transitions in non interacting Bose gases.  $E$ plays the role of \textit{test-functions space}, describing the effective wave-functions of a single particle. In the infinite-volume limit the correct choice for $E$ depends on the state we are considering \cite{Araki_Woods_63}. Specifically, there exist infinite-volume states $\omega^{\overline{\rho}(h),0}_h$ parametrized by the density $\overline{\rho}(h)$, satisfying the quantum KMS condition for all $\overline{\rho}(h)>\rho_c(\beta h)$ and whose two-points functions functions differ from those corresponding to $\overline{\rho}(h)\leq \rho_c(\beta h)$ by an additional term \begin{equation}\label{eq: condensate piece}
       \omega^{\overline{\rho}(h),0}_h (a_h^*(f)a_h(g)) = 2^\nu h(\overline{\rho}(h)-\rho_c(\beta h))\int d^\nu x d^\nu y\; \overline{g(x)}f(y) + \dots,
   \end{equation} which is proportional to the condensate fraction $\overline{\rho}(h)-\rho_c(\beta h)$, i.e. the density of particles in the ground state. We will address the quantity in Eq. \eqref{eq: condensate piece} by the name of \textit{condensate term}. Since the infinite-volume ground state correspond to the effective single-particle wavefunction $|1\rangle$, which is constant on all $\mathbb{R}^\nu$, to \textit{test} the distribution $\omega^{\overline{\rho}(h),0}_h$ we have to restrict to spaces $E$ containing functions decreasing sufficiently fast at large distances , e.g. $E = \mathcal{S}(\mathbb{R}^\nu)$.
   
   Weyl algebras with partially degenerate symplectic forms  can be used to model the presence of classical subsystems, \cite{Honegger_Rieckers_20023,Manuceau_Sirugue_Testard_Verbeure_73} and have found applications in the context of globally hyperbolic space-times \cite{Benini_Dappiaggi_Hack_2014,Dappiaggi_Drago_Longhi_2020}. Moreover, if the form is taken to be exactly $0$, one can describe classical commutative observable algebras \cite{Binz_Honegger_Rieckers_2004, Binz_Honegger_Roecker_2024_2}. In the present context, we consider the commutative algebra $\mathcal{W}(E,0)$ as a semi-classical limit of the quantum algebra $\mathcal{W}(E,h\sigma)$. Physically, this limit corresponds to the high density regime $\overline{\rho}(h)\to+\infty$, which will obtained by sending $h\to 0^+$. Hence, $h$ should be understood as being related to a macroscopic quantity: the mean density of the gas.

   Quantum and classical Weyl algebras can be connected within the setting of \textit{strict deformation quantization} (SDQ). This notion requires the introduction of a Poisson $^*$-subalgebra $(\mathcal{P},\{\cdot,\cdot\})$ of the classical observables, where the bracket $\{\cdot,\cdot\}$ plays the counterpart of commutators for quantum systems.
   The rigorous definition of a SDQ was given in \cite{Landsman_1998, Rieffel_94} as follows
  
   \begin{definition}\label{def: SDQ} (\textbf{Strict deformation quantization}) Let $I\subset\mathbb{R}$ be a subset of the real line containing $0$ as an accumulation point. A strict quantization $(\mathfrak{A}_{h},Q_{h})_{h\in I}$ of the Poisson algebra $(\mathcal{P},\{\cdot,\cdot\})$ consists for every $h\in I$, of a linear, $^*$-preserving map\begin{equation}
   		Q_{h}\colon \mathcal{P}\rightarrow \mathfrak{A}_{h},
   	\end{equation} where $\mathfrak{A}_{h}$ is a $C^*$-algebra with norm $\norm{\cdot}_{h}$, such that $Q_{0}$ is the identical embedding of $\mathcal{P}$ into $\mathfrak{A}_0$, and such that for all $\mathfrak{a},\mathfrak{b}\in \mathcal{P}$ the following conditions hold:\begin{enumerate}
   	\item[(i)] (\textbf{Rieffel's condition}) $I\ni h\rightarrow \norm{Q_{h}(\mathfrak{a})}_{h}\in [0,\infty)$ is continuous.
   	\item [(ii)](\textbf{von Neumann's condition}) $\lim_{h\rightarrow 0}$ $\norm{Q_{h}(\mathfrak{a})Q_{h}(\mathfrak{b}) - Q_{h}(\mathfrak{a}\mathfrak{b})}_{h} = 0$.
   	\item[(iii)] (\textbf{Dirac's condition}) $\lim_{h\rightarrow 0} \norm{\frac{i}{h}\comm{Q_{h}(\mathfrak{a})}{Q_{h}(\mathfrak{b} )} - Q_{h}(\{\mathfrak{a},\mathfrak{b}\})}_{h} =0 $.
   	\end{enumerate}
   	In addition, the \textbf{deformation} condition is satisfied if $Q_{h}$ is injective and $Q_{h}(\mathcal{P})$ is a $^*$-subalgebra of $\mathfrak{A}_{h}$.
   \end{definition} For Weyl algebras the non-negative real number $h$ appearing in the quantization maps $Q_{h}$ is the same semi-classical parameter appearing in the symplectic form $h\sigma$. 
   
   We discuss now the main results of this article:\begin{itemize}
   	\item[(I)] We define a quantization map $(Q_{h})_{h\in [0,+\infty)}$ directly at the algebraic level, relating classical and quantum Weyl algebras for a symplectic space $E$. The map $Q_h$ can be considered as an extension of Berezin quantization map to the infinite dimensional Weyl  setting. Indeed, if we restrict $Q_{h}$ to $\mathcal{W}(G,0)$, with $G$ a finite dimensional symplectic subspace of $E$, the latter map coincides with Berezin quantization in the Schr\"odinger representation. This property is fundamental in proving $Q_{h}$'s positivity and continuity in proposition \ref{prop: positivita e continuita della quantizzazione astratta}. Moreover, this map can be successfully defined on the whole Weyl algebra. In theorem \ref{th: la quantizzazione astratta è SDQ} we exploit results from \cite{Binz_Honegger_Rieckers_2004} to conclude that $(Q_{h})_{h\in[0,+\infty)}$ defines a SDQ and in proposition \ref{prop: iniettività della mappa} we establish the map's injectivity and lack of surjectivity.
   	\item[(II)] We elaborate a classical setting for studying dynamics and equilibrium conditions directly on the commutative algebra $\mathcal{W}(E,0)$. We define a \textit{weak KMS condition} in Def. \ref{def: classical KMS condition}. Furhtermore, we establish its equivalence with the classical KMS condition for measures (see \cite{Ammari_Sohinger_2023} and Def. \ref{def: KMS per misure}), in the case of a dynamics described by certain Hamiltonian operator $H$. This leads to a uniqueness result for weak KMS states $\omega_0^{L,\mu}$ on $\mathcal{W}(\mathbb{L}^2(\Lambda_L)),0)$.
    
   Exploiting the quantization map $Q_{h}$, \textit{cf}. (I), we construct a non-trivial class of \textit{infinite density states} for $\mathcal{W}(\mathcal{S}(\mathbb{R}^\nu),0)$ as the weak$^*$ limit point \begin{equation}\lim_{h\to 0^+}\omega^{\overline{\rho}(h),0}_h\circ Q_h = \omega^{\alpha,0}_0.\end{equation} The resulting states lead to a well defined representation of the CCR in the infinite-volume without requiring the existence of a local number operator $N_{\Lambda_L}$, nor the existence of a density operator $\lim_{L\to+\infty}N_{\Lambda_L}/|\Lambda_L|$ \cite{Araki_Woods_63}. Moreover these states  describe a macroscopic occupation of the ground-state by manifesting a condensate term like in Eq. \eqref{eq: condensate piece}. This condensate term is multiplied by a renormalized condensate density $\alpha$, obtained as \begin{equation}\lim_{h\to 0} h(\overline{\rho}(h)- \rho_c(\beta h)) =: \alpha. \end{equation}

   	\item [(III)] 
     We recover the states $\omega^{\alpha,0}_0$ from a thermodynamic limit $\Lambda_L\uparrow \mathbb{R}^\nu$ for the expectation values of Weyl elements on finite volume Gibbs states $\omega_0^{L,\mu}(W^0(f))$, in the spirit of the classical works of Cannon \cite{Cannon_73}, and Lewis and Pulè  \cite{LP_73}. This point of view highlights the main physical difference between the classical and quantum frameworks: for a quantum Bose gas the phase transition's order parameter can be chosen to be either the inverse temperature $\beta$ or the density $\overline{\rho}$; the classical theory is obtained as an high density limit $\overline{\rho( h)}\to +\infty$, so that the local density is always infinite and we have non-equivalent KMS states for all values of the temperature.
   \end{itemize}
     Another point of view on the semi-classical limit of a Bose-condensate can be found in \cite{Ammari_Sebastien_Nier_2019}, where a multi-scale approach is employed to deal at the same time with classical and quantum systems. There is a vast literature on semi-classical limits of boson gases. We cite here some recent results for finite temperature systems \cite{Deuchert_Seiringer_Yngvason_2019, Deuchert_Seiringer_2021, Lewin_Nam_Rougerie_2014, Nam_Triay_2023}. The zero temperature case is treated in \cite{Boccato_Brennecke_Cenatiempo_Schlein_2019} and in the review \cite{Lieb_Solovej_Seoronger_Yngvason_2005}.
   \bigskip
   
   The paper is organized as follows.
   In section \ref{sec: weyl algebra} we define the algebraic setting.
   Specifically, in \ref{sec: general construction} we recall the construction of Weyl algebras for a  symplectic form $h \sigma$, $h\in [0,+\infty)$ while in \ref{sec: commutative case} we focus on commutative algebras taking $h = 0$.  Poisson brackets and infinitesimal generators are introduced for the classical Weyl algebra in Sec. \ref{sec: additional structures}. In \ref{sec: classical KMS condition} we define an appropriate notion of derivation and we introduce the weak KMS condition for classical systems.
   In \ref{subsec: definition of stricti deformation quantization} we build the strict deformation quantization $Q_{h}$ and we describe its properties in Sec. \ref{sec: further properties}. The construction of classical states via the quantization map is discussed in Sec. \ref{sec: dequantization}.
We summarize the construction of equilibrium states for quantum Bose gases in \ref{sec: condensazione quantistica}, while in \ref{sec: sistemi a volume finito} we characterize classical finite-volume equilibrium states.
   In \ref{sec: limiti classici degli stati} we connect quantum and classical KMS states via the abstract quantization $Q_{h}$ and in \ref{sec: limite termodinamico classico} we analyze the termodynamic limit for classical states.
   Appendices \ref{sec: appendice A}-\ref{sec: appendice B} contain the proofs of ancillatory results while appendix \ref{sec: appendice C} discusses further properties of weak KMS states.

   \paragraph{Acknowledgements.}
I am grateful to N. Drago and C.J.F. van de Ven for their supervision and crucial proofreading/correction of the paper, and to V. Moretti and C. De Rosa for many helpful discussions on this project. I am much indebted with  M. Falconi for the clarifying discussions on Sobolev Hilbert spaces and quantization maps. 
I would like to thank K. Fredenhagen for formulating the question that motivated this paper.

 I acknowledge the support of the GNFM group of INdAM.

\paragraph{Data availability statement.}
Data sharing is not applicable to this article as no new data were created or analysed in this study.

\paragraph{Conflict of interest statement.}
The author certifies that he has no affiliations with or involvement in any
organization or entity with any financial interest or non-financial interest in
the subject matter discussed in this manuscript.
   
	\section{Weyl $C^*$-algebras}\label{sec: weyl algebra}
	In the following sections we recall the construction of  $C^*$-algebras of canonical commutation relations, also known as \textit{Weyl} $C^*$-algebras. The generators of the latter algebra are labeled  by elements of a symplectic, real vector space $E$, equipped with a \textit{symplectic form} $h\sigma$. For concreteness, we will take $E$ to descend from a separable complex Hilbert space $(\mathcal{H},\langle\cdot,\cdot\rangle)$ with $\sigma(\cdot,\;\cdot):= \Im{\langle\cdot,\cdot\rangle}$ and $h\in [0,+\infty)$. For $h> 0$, the form is non degenerate and standard constructions follow \cite{Bratteli_Robinson_97, Kastler_65, Manuceau_68, Petz_90}. However, for $h=0$ we will benefit from the more general setting introduced in \cite{Manuceau_Sirugue_Testard_Verbeure_73} and explored further in \cite{Binz_Honegger_Rieckers_2004,Binz_Honegger_Roecker_2024_2}, where the form is allowed to be degenerate. In particular, we are interested in the commutative Weyl $C^*$-algebra $\mathcal{W}(E,0)$. This will be taken as the \textit{observables algebra} for the classical theory analyzed in section \ref{sec: limite classico del gas libero}.

   After reviewing some standard results in Sec. \ref{sec: general construction}, \ref{sec: commutative case} and \ref{sec: additional structures}, we introduce the weak KMS condition and weak derivations in Sec. \ref{sec: classical KMS condition}.

\subsection{General construction}\label{sec: general construction}
The goal of this section is to build a  Weyl $C^*$-algebra $\mathcal{W}(E,h \sigma)$, where $h$ is a semi-classical parameter ranging from $0$ (fully degenerate case) to any finite positive value. We consider possibly degenerate symplectic forms, motivated by the need to build a quantization map which relates in a uniformed setting the commutative and non-commutative Weyl $C^*$-algebras \cite{Binz_Honegger_Rieckers_2004, Honegger_Rieckers_2005, Landsman_1998, Landsman_2017}. Proofs of the results presented in this section can be found in \cite{Binz_Honegger_Rieckers_2004} and references cited therein.

\paragraph{The $^*$-algebra $\Delta(E,h\sigma)$:} We start with the abstract unital free algebra $\mathfrak{W}_{0}$ generated by a set of elements labelled by vectors in $E$,  $\{W^{h}(f),\; f\in E\}\cup \{I\}$.
   Quotienting by the relations \begin{align}\label{eq: relazioni generiche per Weyl}
	& W^{h}(f)W^{h}(g) = e^{-\frac{i}{2}h \sigma(f,g)}W^{h}(f+g),\\
	& (W^h)^{*}(f)= W^h(-f),\\
	&W^h(0) = I,
	\end{align} we obtain a $^*$-algebra generated by the unitary elements $\{W^{h}(f),\; f\in E\}$ satisfying relations (2.2)-(2.4) in \eqref{eq: relazioni generiche per Weyl}. This $^*$-algebra is usually denoted as $\Delta(E,h \sigma)$. Thanks to Eq. \eqref{eq: relazioni generiche per Weyl}, the algebra can be made explicit as \begin{equation}
	\Delta(E,h \sigma) = \textsc{LH}\{W^{h}(f),\; f\in E\},
	\end{equation}that is, the \textit{finite linear combinations} of \textit{Weyl elements} $\{W^{h}(f),\; f\in E\}$ give the full $^*$-algebra.

	To obtain a  concrete realization of the $^*$-algebra, one can consider the space $\mathcal{F}(E,\mathbb{C})$ of all complex functions on $E$. We select from this space the \textit{delta-functions} \cite[Sec. 2.2]{Manuceau_68}: \begin{equation}\label{eq: concrete realization}
		\delta_{f}[g] = \begin{cases}
			1\quad f = g\\
			0\quad f\neq g
		\end{cases},\quad f,g\in E.
	\end{equation}  If we endow the linear hull $\textsc{LH}\{\delta_{f},\; f\in E\}$ with the product $\delta_{f}\cdot \delta_{g} = \exp{-ih \sigma(f,g)/2}\delta_{f+g}$ and the $^*$-operation $\delta_{f}^{*} = \delta_{-f}$, we see that the $^*$-algebra generated by $\textsc{LH}\{\delta_{f},\; f\in E\}$ is unital, with $\delta_{0}$ as unit element, and is $^*$-isomorphic to $\Delta(E,h \sigma)$ by the linear $^*$-isomorphism\begin{equation}
	    \sum_{i=1}^n \alpha_i W^h(f_i)\to \sum_{i=1}^n \alpha_i \delta_{f_i}.
	\end{equation} From this representation it is clear that the elements $\{W^{h}(f),\; f\in E\}$ are linearly independent. Indeed, suppose that \begin{equation}
	    \sum_{i=1}^n \alpha_i \delta_{f_i} = 0,
	\end{equation} for some complex coefficients $\{\alpha_i\}_{i\in{1,\dots,n}}$ and distinct vectors $\{ f_{i}\}_{i\in\{1,\dots,n\}}\subset E$. Then, by Def. \eqref{eq: concrete realization} we have that \begin{equation}
	    \sum_{i=1}^n \alpha_i \delta_{f_i}[f_j] = \alpha_j = 0,
	\end{equation} for all $j\in{1,\dots, n}$.
	
 \paragraph{Completion to a $C^*$-algebra:} To complete the $^*$-algebra $\Delta(E,h \sigma)$ to an appropriate $C^*$-algebra $\mathcal{W}(E,h \sigma) = \overline{\Delta(E,h \sigma)}$
 we need a suitable $C^*$-norm. To settle the notation, we write the definition of a state on Weyl $^*$-algebra\begin{definition}
     We define $\mathcal{S}(\Delta(E,h\sigma))$ to be the state space of $\Delta(E,h\sigma)$, that is, the set of linear functionals $\omega\colon \Delta(E,h\sigma)\to \mathbb{C}$, satisfying \begin{itemize}
         \item [(i)] $\omega(A^* A)\geq 0$, for all $A\in\Delta(E,0)$;
         \item [(ii)] $\omega(I) = 1$.
     \end{itemize}
 \end{definition}

{Now, we collect a number of well known properties required for the characterization of  $\mathcal{W}(E,h\sigma)$, referring the reader to the literature for their proofs.
\begin{enumerate}
    \item [(a)] It follows that the positive map 
  \begin{equation}\label{eq: unicita della C norma}
		\Delta(E,h \sigma)\ni A \longrightarrow \norm{A}_{h\sigma}:= \sup\bigg{\{}\omega(A^*A),\;\omega\in S(\Delta(E,h \sigma))\bigg{\}}\in [0,+\infty),
	\end{equation} is a $C^*$-norm \cite[Lem. 3.1]{Manuceau_Sirugue_Testard_Verbeure_73}.
    \item [(b)] $\norm{\cdot}_{h \sigma}$ is the biggest possible $C^*$-norm on $\Delta(E,h \sigma)$  \cite[Cor. 3.8]{Manuceau_Sirugue_Testard_Verbeure_73}, i.e. if $\norm{\cdot}$ is another $C^*$-norm, then, for all $A\in\Delta(E,h \sigma)$, $\norm{A}\leq \norm{A}_{h \sigma}$.
    \item[(c)] As a consequence of (a) and (b), one gets that $\norm{\cdot}_{h \sigma}$ is the unique $C^*$-norm making all representations $(\pi,\mathcal{H}_\pi)$ of $\Delta(E,h \sigma)$ $\norm{\cdot}_{h \sigma}$-continuous, i.e. $\norm{\pi(A)}\leq \norm{A}_{h \sigma}$ for all $A\in\Delta(E,h \sigma)$.
    \item[(d)] By (c), every representation of the $^*$-algebra extends $\norm{\cdot}_{h \sigma}$-continuously to a representation of $\mathcal{W}(E,h \sigma)$.
    \item[(e)]  If $h>0$, $\norm{\cdot}_{h \sigma}$ is the unique $C^*$-norm and $\mathcal{W}(E,h \sigma)$ is simple \cite[Cor. 4.23, 4.24]{Manuceau_Sirugue_Testard_Verbeure_73}.
As a corollary, when $h>0$, every representation of the Weyl $C^*$-algebra is faithful.
For a similar uniqueness result when the form is degenerate, see \cite[Th. 3-7]{Binz_Honegger_Rieckers_2004}.  
\item[(f)] Through the explicit expression of the norm in \eqref{eq: unicita della C norma} and the properties of states in $\mathcal{S}(\Delta(E,h \sigma))$, one has that for every symplectic subspace $E_{0}\subset E$ with  symplectic form $h \sigma_{0} = h \sigma|_{E_{0}\cross E_{0}}$,
   the Weyl $C^*$-algebra $\mathcal{W}(E_{0},h \sigma_{0})$ obtained by completing $\Delta(E_{0},h \sigma_{0})$ with the maximal norm $\norm{\cdot}_{h \sigma_{0}}$, is a $C^*$-subalgebra of $\mathcal{W}(E,h \sigma)$ and the respective norms satisfy $\norm{\cdot}_{h \sigma_{0}} = \norm{\cdot}_{h \sigma|_{E_{0}}}$ \cite[Sec. 3.3]{Manuceau_68}.
    \end{enumerate}
 From now on we will denote the norms as $\norm{\cdot}_{h}:= \norm{\cdot}_{h\sigma}$.}
	
	\begin{example}\label{ex: schrodinger representation}
	For later convenience, we briefly recall here a concrete example of Weyl $C^*$-algebra and discuss its representations. We take $E = \mathbb{L}^2(\mathbb{R}^\nu)$,  $h\sigma(\cdot,\cdot) = h\Im{\langle\cdot,\cdot\rangle }$ with $h>0$. We choose a symplectic basis $\{e_{k},ie_{k}\}_{k\in\mathbb{N}}$, where  $\{e_{k}\}_{k\in\mathbb{N}}$ is an orthonormal basis of $\mathbb{L}^2(\mathbb{R}^\nu)$. By re-naming $ie_{k}=: f_{k}$, this basis satisfies the usual relations\begin{equation}
			\sigma(f_{j},f_{k}) = \sigma(e_{j},e_{k}) = 0,\quad \sigma(e_{j},f_{k}) = \delta_{jk}.
		\end{equation}

			Now, let $E_{0}\subset \mathbb{L}^2(\mathbb{R}^\nu)$ be the finite dimensional complex subspace of $\mathbb{L}^2(\mathbb{R}^\nu)$ spanned by $\{e_{k},f_{k}\}_{k\in\{1,\dots, n\}}$ for some $n\in\mathbb{N}$. Since $\sigma$ is the imaginary part of the standard scalar product of $\mathbb{L}^2(\mathbb{R}^\nu)$, $\sigma|_{E_{0}\cross E_{0}}$ is non degenerate and we can regard $E_{0}$ as a symplectic real vector space of dimension $\dim_{\mathbb{R}}(E_{0}) = 2\dim_{\mathbb{C}}(E_{0}) = 2n$. Given any $f,g\in E_{0}$, we have
   \begin{multline}
			\sigma(f,g) = \sum_{k,j=1}^{n}\sigma(\lambda_{k}e_{k}+i\mu_{k}e_{k},\lambda'_{j}e_{j} + i\mu_{j}'e_{j})\\ = \sum_{k=1}^{n}(\lambda_{k}\mu_{k}' - \lambda'_{k}\mu_{k}) = \sum_{k=1}^{n}\sigma_{\mathbb{R}^{n}}(\lambda_{k}\bm{e}_{k}+\mu_{k}\bm{f}_{k},\lambda_{k}'\bm{e}_{k}+\mu'_{k}\bm{f}_{k}),
		\end{multline} where in the last line we have called $\sigma_{\mathbb{R}^{2n}}$ the standard symplectic form of $\mathbb{R}^{2n}$ and defined the canonical symplectic basis $\{\bm{e}_{1},\bm{f}_{1},\dots, \bm{e}_{n},\bm{f}_{n}\}\subset \mathbb{R}^{2n}$.
		
		It is a standard result \cite{Bratteli_Robinson_97,Folland_1989,Petz_90}, that the unique (up to $^*$-isomorphism) irreducible regular representation of $\mathcal{W}(E_{0},h\sigma)$ is the Schr\"odinger representation on $\mathbb{L}^{2}(\mathbb{R}^{n})$. This can be explicitly written extending by linearity and continuity the prescription on the Weyl  elements
		\begin{equation}\label{eq: sch representation}
			\pi_{S}(W^{h}(f)) = \pi_{S}(W^{h}(\sum_{k=1}^{n} \lambda_{k}e_{k} +\mu_{k}f_{k} )) := \exp{i\overline{\sum_{k=1}^{n}\lambda_{k}\hat{X}_{k} +\mu_{k}\hat{P}_{k}  }},
		\end{equation} where $(\hat{X}_{k},\hat{P}_{k})_{k\in\{1,\dots,n\}}$ are the standard position and momentum operator components, defined by the closure of the operators $(\hat{X}_{k}f)(x):= x_{k}f(x),\; (P_{k}f) = -ih (\partial f/ \partial x_{k})(x)$ initially defined on $\mathcal{S}(\mathbb{R}^{n})$. Indeed, we can verify that  \begin{align}
			\pi_{S}(W^{h}(f))\pi_{S}(W^{h}(g)) =&\exp{i\overline{\sum_{k=1}^{n}\lambda_{k}\hat{X}_{k} +\mu_{k}\hat{P}_{k}  }}\exp{i\overline{\sum_{j=1}^{n}\lambda_{j}'\hat{X}_{j} +\mu_{j}'\hat{P}_{j}  }} \nonumber \\
			=& \exp{-\frac{1}{2}\overline{\sum_{j,k=1}^{n}\comm{\lambda_{k}\hat{X}_{k}+\mu_{k}\hat{P}_{k}}{\lambda'_{j}\hat{X}_{j}+ \mu'_{j}\hat{P}_{j}}}}\pi_{S}(W^{h}(f+g))\nonumber\\
			=&e^{-\frac{ih}{2}\sigma(f,g)}\pi_{S}(W^{h}(f+g)).
		\end{align} Since $h\sigma$ is non-degenerate, the Schr\"odinger representation is faithful.
	\end{example}
	\subsection{Commutative case: $h = 0$}\label{sec: commutative case}
	 In this section we analyze the commutative Weyl  $C^*$-algebra $\mathcal{W}(E,0)$\cite{Binz_Honegger_Rieckers_2004}. {We introduce the following notation for the the topological dual of $E$:\begin{itemize}
	     \item [(i)] $E'_{\tau}$ denotes the topological dual with respect to a given convex topology $\tau$;
         \item[(ii)] in particular, the dual with respect to the norm topology is denoted by $E'_{\norm{\cdot}}$.
	 \end{itemize}} Some of the results of this section can be generalized for arbitrary symplectic spaces $E$.
     
     To effectively make computations with $\mathcal{W}(E,0)$, it is useful to identify the latter  $C^*$-algebra with some suitable space of functions.
	\paragraph{Almost periodic functions:}
	We introduce the  $C^*$-algebra of \textit{almost periodic functions} \cite{Corduneanu_1989}. For every $f\in E$, we consider the $\sigma(E'_{\tau},E)$ bounded, continuous mapping $\xi(f): E'_{\tau}\rightarrow \mathbb{C}$, defined as \begin{equation}
		\xi(f)(F) := \exp{iF(f)}, \quad F\in E'_{\tau}.
	\end{equation} These functionals satisfy the commutative Weyl relations with respect to the usual pointwise function product and $^*$-operation\begin{equation}
	\xi(f)\xi(g) = \xi(f+g),\quad \xi(f)^{*} = \xi(-f),\quad \xi(0) = 1.
	\end{equation} The completion of the set $\text{LH}\{\xi(f),\; f\in E\}$ with respect to the sup norm $\norm{\cdot}_{\infty}$ forms the  $C^*$-algebra of almost periodic functions, denoted by $\text{AP}(E,E'_{\tau})$ \cite{Corduneanu_1989}

\begin{remark}
    In \cite[Th. 4-3]{Binz_Honegger_Rieckers_2004}) the authors proved that the fully degenerate Weyl $C^*$-algebra $\mathcal{W}(E,0)$ can be characterized in terms of Almost periodic functions, i.e.
\begin{equation}
			\mathcal{W}(E,0) \simeq \textsc{AP}(E,E'_{\tau})\,,
		\end{equation}
		with the identification of the classical Weyl element $W^{0}(f)\leftrightarrow \xi(f)$, $f\in E$. It follows that the $C^*$-algebra $\textsc{AP}(E,E'_{\tau})$ does not depend on the chosen locally convex topology $\tau$. Indeed, the  $C^*$-algebra $\textsc{AP}(E,E'_{\tau})$ only sees the test function space $E$. Indeed, even if $E'_{\tau}\simeq F'_{\overline{\tau}}$, but $E\neq F$, then $AP(E,E'_{\tau})\neq AP(F,F'_{\overline{\tau}})$. From now on we will remove the dependence on $\tau$ from the notation and refer to the almost periodic functions space as $\textsc{AP}(E):= \textsc{AP}(E,E'_{\tau})$. 
        \end{remark}
	
		{ Since $^*$-isomorphisms  between  $C^*$-algebras preserve the  $C^*$-norm, if $\beta\colon\textsc{AP}(E)\to\mathcal{W}(E,0)$ is the $^*$-isomorphism between the almost periodic functions and the Weyl  $C^*$-algebra, it follows that for any faithful representation $(\pi,\mathcal{H}_\pi)$
	of $\textsc{AP}(E)$,  $(\pi\circ \beta,\mathcal{H}_\pi)$ is also a faithful representation of $\mathcal{W}(E,0)$. Thus, the norm coincides  \begin{equation}\norm{\beta(c)}_{\infty} = \norm{\pi(\beta(c))} = \norm{c}_{0}. \end{equation} From now on, will identify the Weyl $C^*$-algebra with $\textsc{AP}(E)$ without writing the isomorphism map explicitly.}

\begin{remark}\label{ex: struttura di weyl classica per L2}
	{In our case, $E$ is also a real Hilbert space with the same norm of its complexification $(\mathcal{H},\langle\cdot,\cdot\rangle)$ and scalar product $\langle\cdot,\cdot\rangle_{E}:= \Re{\langle\cdot,\cdot\rangle}$. Every linear, real,  continuous functional $\Psi_{\mathbb{R}}\in E'_{\norm{\cdot}}$ can be seen as the real part of a linear, complex, continuous functional on $\mathcal{H}$ with the same norm \cite[Prop. 1.9.3]{Megginson_98}: \begin{equation}\Psi(f)\colon= \Psi_{\mathbb{R}}(f) -i\Psi_{\mathbb{R}}(if),\quad \text{for all }f\in E.\end{equation} By Riesz representation theorem we can write for every linear, real, continuous functional \begin{equation}\label{eq: Riesz}
			\Psi_{\mathbb{R}}(f) = \Re{\Psi(f)} = \Re{\langle g_{\Psi}, f\rangle},\quad \text{for all }f\in E.
		\end{equation} Equation \eqref{eq: Riesz} tells us that every functional of this kind can be expressed in term of the real scalar product of $E$. Since Hilbert spaces are reflexive with respect to the norm topology, we have $\mathcal{H} \equiv E \simeq E'_{\norm{\cdot}}$, so that, the almost periodic functions (and accordingly, also the Weyl elements) can be seen as maps acting on $E$ as \begin{equation}
        W^{0}(f)[g] = \exp{i\Re{\langle f, g\rangle}},\quad \text{for all }f,\; g \in E.
        \end{equation}}

	\end{remark}
	
	\begin{remark}\label{remark: duality for subspaces}
		{Consider a subspace $E_{0}\subset E$. We can equip $E_0$ with the norm topology, which coincides with the subspace topology induced from $E$.} By standard results \cite[Th. 1.10.26]{Megginson_98}, $(E_{0})'_{\norm{\cdot}}$ is a Banach space. Moreover, \begin{equation}(E_{0})'_{\norm{\cdot}}\simeq E'_{\norm{\cdot}}/\{\Psi\in E_{\norm{\cdot}}',\; \Psi|_{E_{0}} = 0\},\end{equation} as Banach spaces and in particular  $E'_{\norm{\cdot}}\simeq (E_{0})'_{\norm{\cdot}}$, if $E_{0}$ is dense in $E$.
        \end{remark}

    {Armed with the previous remarks we consider two relevant examples of classical Weyl $C^*$-algebra: $\mathcal{W}(\mathbb{L}^2(\mathbb{R}^\nu),0)$ and $\mathcal{W}(\mathcal{S}(\mathbb{R}^\nu),0)$.}

        {\begin{example} Consider $E = \mathbb{L}^2(\mathbb{R}^\nu)$, with $E'_{\norm{\cdot}}\simeq \mathbb{L}^2(\mathbb{R}^\nu)$, and the symplectic subspace $E_{0} = \mathcal{S}(\mathbb{R}^{\nu})$. The Weyl $C^*$-algebra $\mathcal{W}(E_{0}, 0)$ can be identified with the $C^*$-algebra of almost periodic functions $\textsc{AP}(E_{0})$ by taking the sup-norm closure of \begin{equation}\label{eq: star algebra for schwartz}
			\Delta(E_{0},0) = \text{LH}\{ W^{0}(f),\; f\in E_{0},\;\text{such that}\; W^{0}(f)[g] = e^{i\Re{\langle g,f\rangle}},\;  g\in E'_{\norm{\cdot}}\}.
		\end{equation}
		Note that in \eqref{eq: star algebra for schwartz} we have implicitly identified $E'_{\norm{\cdot}}$ with $(E_{0})'_{\norm{\cdot}}$, so that the duality relation between $E_{0}$ and $(E_{0})'_{\norm{\cdot}}$ can be written in the same way as the one between $E$ and $E'_{\norm{\cdot}}$.
		By writing the $C^*$-algebra in this way, it is clear that if $\{g_{n}\}_{n\in\mathbb{N}}\subset E_{0}$ is a sequence in the Schwartz functions space, then $g_{n}\xrightarrow{n\rightarrow\infty}g\in E_{0}$ in the $\sigma(E_{0},(E_{0})'_{\norm{\cdot}})$ topology if and only if $W^{0}(g_{n})\xrightarrow{n\rightarrow\infty} W^{0}(g)$ pointwise.
	\end{example}}

 \subsection{Poisson brackets and infinitesimal generators on $\mathcal{W}(E,0)$}\label{sec: additional structures}
	We define the necessary objects for describing the classical dynamics on $\mathcal{W}(E,0)$. 

   \paragraph{Poisson brackets:} Following \cite{Binz_Honegger_Rieckers_2004}, we equip $\mathcal{W}(E,0)$ with a Poisson structure.
	 We take the dense $^*$-algebra $\Delta(E,0)\subset\mathcal{W}(E,0)$ as the domain of the Poisson bracket. These are defined on Weyl elements as \begin{equation}\label{eq: Poisson parenthesis}
		\Delta(E,0)\cross \Delta(E,0) \ni (W^{0}(f), W^{0}(g))\rightarrow \{W^{0}(f),W^{0}(g)\}:= \sigma(g,f) W^{0}(f+g)\in \Delta(E,0),\end{equation} and extended on $\Delta(E,0)$ by linearity.
  This definition can be justified by considering the standard Poisson bracket in $\mathbb{R}^{2n}$: $\{f,g\} = \sum_{i=1}^{n}\partial_{q_{i}}f\partial_{p_{i}}g - \partial_{p_{i}}f\partial_{q_{i}}g$,  applied to the Weyl functions $W^{0}(\lambda,\mu)(q,p) = \exp{i(\lambda\cdot q+ \mu\cdot p)} $. Indeed, by an explicit computation \begin{multline}
		\{W^{0}(\lambda,\mu), W^{0}(\lambda',\mu')\} = (\mu\cdot\lambda' -\lambda\cdot\mu)W^{0}(\lambda+\lambda',\mu +\mu')\\ = \sigma_{\mathbb{R}^{2n}}((\lambda',\mu'),(\lambda,\mu))W^{0}(\lambda+\lambda',\mu+\mu').
	\end{multline}
	The bracket in \eqref{eq: Poisson parenthesis},  together with the dense $^*$-algebra $\Delta(E,0)$ give rise to a Poisson commutative algebra $(\mathcal{W}(E,0),\{\cdot,\cdot\})$. 
	
	\paragraph{Infinitesimal generators:} Another fundamental family of objects is the set of classical \textit{infinitesimal generators}.
	We can define \begin{equation}
		\Phi_{0}(f) : E'_{\norm{\cdot}}\rightarrow \mathbb{R}, \quad \Phi_{0}(f)[g] := \Re{\langle f,g\rangle} \quad f\in E,
	\end{equation} where we have identified the elements of $E'_{\norm{\cdot}}$ and $E$.
 We will refer to $\Phi_{0}(\cdot)$ as \textit{field functions}. The Weyl elements take the form $W^{0}(f) = \exp{i\Phi_{0}(f)}$.  When working with the Weyl $C^*$-algebra $W(E,h\sigma)$, one can extend the action of analytic states on \textit{field operators} $\Phi_{h}(\cdot)$. In the same way,  we extend classical states $\omega_{0}: \mathcal{W}(E,0)\rightarrow \mathbb{C}$ to compute expectation values of pointwise products of functions like \begin{equation}
		\omega_{0}(\Phi_{0}(f_{1})\dots \Phi_{0}(f_{n})).
	\end{equation}
To this avail it suffices to consider \textit{analytic} classical states, which are normalized linear functionals such that the associated GNS representation $(\mathcal{H}_{0},\pi_{0},\Omega_{0})$ is regular and $\Omega_{0}$ is an analytic vector for $\{\Phi_{\omega_{0}}(f),\; f\in E\}$, where $\Phi_{\omega_{0}}(f)$ is the strong (weak) generator of $\pi_{0}(W^{0}(f))$. A particular class of analytic states are the quasi-free states, which can be defined on Weyl elements as (see \cite[Chap. 3]{Petz_90}) \begin{equation}
		\omega_{s}(W^{0}(f)) := \exp{-\frac{1}{4}s(f,f)},
	\end{equation} where $s(\cdot,\cdot): E\cross E\rightarrow \mathbb{R}$ is a real, bi-linear symmetric form on $E$. The definition of analytic states implies that $\mathbb{R}\ni t \rightarrow \omega_{0}(W^{0}(tf))$ is infinitely differentiable and that the state can be extended on field operators as \begin{multline}\label{eq: estensione stato classico}
		\omega_{0}(\Phi_{0}(f_{1})\dots \Phi_{0}(f_{n}))= \langle\Omega_{0}, \Phi(f_{1})\dots \pi_{\omega_{0}}\Phi(f_{n})\Omega_{0}\rangle \\= (-i)^{n}\frac{\dd}{\dd t_{1}}\dots\frac{\dd}{\dd t_{n}}\omega_{0}(W^{0}(t_{1}f_{1}+\dots + t_{n}f_{n}))\Big|_{t_1=\ldots = t_n = 0}.
	\end{multline}
		\begin{remark}\label{remark: regularity of states}
		We can also define $C^{2k}$ states as normalized, positive, linear functionals whose GNS representation is regular with associated cyclic vector satisfying $\Omega_{0}\in D(\Phi_{\omega_{0}}(f)^{k}),$ for all $ f\in E$. This condition is equivalent to saying that $\Omega_{0}$ is of class $C^{k}$ (see the discussion before example 5.2.18 in \cite{Bratteli_Robinson_97}). For these states, it is possible to compute in a sensible way $\omega_{0}(\Phi_{0}(f_{1})\cdot\Phi_{0}(f_{m}))$ for $m\leq k$. This observation is relevant for the definition the weak KMS condition for a commutative Weyl $C^*$-algebra in the next section.
        \end{remark}
        \subsection{KMS conditions for Weyl $C^*$-algebras }
\label{sec: classical KMS condition}
We introduce appropriate definitions of dynamics and equilibrium conditions on $\mathcal{W}(E,0)$. We start with Rmk. \ref{remark: dinamica classica} by introducing the classical dynamics as \textit{semi-classical limit} of the quantum $^*$-automorphisms.
 {Following this, we construct \textit{weak derivations} in Def. \ref{def: weak derivation} and characterize their properties in two Propositions,   \ref{prop: caratterizzazione derivazioni continue I} and \ref{prop: chiusura pointwise I}, whose proofs are deferred to Appendix \ref{sec: appendice A}}. Lastly, inspired by the classical KMS condition introduced in \cite{Gallavotti_Verboven_75}, we define the weak KMS-condition for the states \ref{def: classical KMS condition}.

	 We introduce the dynamics on $\mathcal{W}(E,0)$ by means of the one-parameter group of $^*$-automorphism 
 \begin{equation}
		\mathbb{R}\ni t \longrightarrow \tau_{0,t}(W^{0}(f)) = W^{0}(e^{iHt}f)\in \mathcal{W}(E,0),\quad f\in E,
	\end{equation}where $H$ is some self-adjoint operator. {\begin{remark}(\textbf{Interpretation of the dynamics}) \label{remark: dinamica classica}  The motivation for the introduction of the latter dynamics is semi-classical in nature. $E$ is a space of test functions for the effective wavefunctions of the single particle. The dynamics on the quantum Weyl algebra is introduced via an automorphism, whose action is implemented on the space $E$. In the $h\to 0^+$ limit, the interpretation of the latter space does not change, i.e. the theory is still fundamentally quantum and only the resulting macroscopic setting is classical. In particular, we expect that the classical limit of the dynamics $\tau_{h,t}(W^{h}(f))= W^{h}(e^{iHt}f)$ should be $W^{0}(e^{iHt}f) = \tau_{0,t}(W^{0}(f))$.
	\end{remark}}

    {The introduction of the classical KMS condition in\cite{Gallavotti_Verboven_75} is motivated by a formal $h\to0^+$ limit of the quantum condition. We repeat this argument in the specific setting of Weyl algebras. The quantum KMS condition \eqref{eq: KMS quantum} can be equivalently written as
\begin{equation}\label{eq: road to classical KMS}
  \omega_h\!\left(
    W^h(f)\,\frac{1}{ih}\big(\tau_{i\beta h}(W^h(g)) - W^h(g)\big)
  \right)
  \;=\;
  \omega_h\!\left(
    \frac{1}{ih}\,[\,W^h(g),\,W^h(f)\,]
  \right).
\end{equation}
In the semiclassical limit $h\to 0^+$, the right-hand side of \eqref{eq: road to classical KMS} should converge to
\[
  \omega_0\!\big(\{W^0(f),W^0(g)\}\big),
\]
where $\{\cdot,\cdot\}$ denotes the Poisson bracket, while the left-hand side should converge to
\[
  \beta\,\omega_0\!\big(W^0(f)\,\delta_0(W^0(g))\big),
\]
where $\delta_0$ is the derivation generated by the $^*$-automorphism group $\tau_0$.
}
    Thus, in order to define classical KMS condition for a Weyl algebra, we must introduce the correct derivations.
	To get some intuition, we consider the small time limit of the $^*$-automorphism $\tau_{0}$. Taking $f\in D(H)$ and identifying $\mathcal{W}(E,0)$ with the $C^*$-algebra of the almost periodic functions we have: \begin{equation}
		\lim_{t\rightarrow 0}\frac{1}{t}\{W^{0}(e^{iHt}f )[g] - W^{0}(f)[g]\} = i\Re{\langle g, iHf\rangle} W^{0}(f)[g],\quad \forall g\in E.
	\end{equation}
This suggest to set
 \begin{equation}\label{eq: tentative definition for derivations} 
		\delta_{0}(W^{0}(f)) := i\Phi_{0}(iHf) W^{0}(f), \quad f\in D(H).
	\end{equation}
 This prescription satisfies the two characterizing properties of derivations 
	\begin{align}
		&\delta_{0}(W^{0}(f)W^{0}(g)) = i\Phi_{0}(iH(f+g))W^{0}(f+g) = \delta_{0}(W^{0}(f))W^{0}(g) + W^{0}(f)\delta_{0}(W^{0}(g))\nonumber \\
		&\delta_{0}(W^{0}(f)^{*}) = -i\Phi_{0}(iH f)W^{0}(-f) = \delta_{0}(W^{0}(f))^{*},
	\end{align} for $f,g\in D(H)$.
    However, $\delta_0$ is not a derivation as defined in \cite[Def 3.2.21]{Bratteli_Robinson_87} since $\delta_{0}(W^{0}(f))\notin \textsc{AP}(E)$. { Indeed, while $E\ni g \to \Phi_0(f)[g]$ is a continuous function of $g$ in the norm topology of $E$, it is not bounded since  $\lim_{\lambda\to+\infty}|\Phi_{0}(f)[\lambda f]| = +\infty$. Thus, $\Phi_0(f)$ cannot belong to $\text{AP}(E)$, because the latter is a subset of $C_b(E'_{\norm{\cdot}})$, the space of bounded, continuous functions on $E_{\norm{\cdot}}'\simeq E$}.
   Nevertheless, we observe that when dealing with an analytic state $\omega_{0}$, it makes sense to compute \begin{equation}\label{eq: computation of weak derivations}
		\omega_{0}(\delta_{0}(W^{0}(f))) = i\omega_{0}(\Phi_{0}(iHf)W^{0}(f)).
	\end{equation}  As we will see, this is basically all that matters for defining a weaker version of the standard classical KMS condition.	Motivated by the above considerations we introduce the following notion of weak derivation.
		\begin{definition}\label{def: weak derivation} (\textbf{Weak derivation}) Let $E$ be a normed, symplectic space and 
		let $C( E'_{\norm{\cdot}})$ be the space of $\sigma(E'_{\norm{\cdot}}, E)$-continuous mapping from $E'_{\norm{\cdot}}$ to $\mathbb{C}$. A weak derivation is a linear operator from a {$^*$-subalgebra $D(\delta_{0})\subset \Delta(E,0)$} to the space of continuous functions on $E'_{\norm{\cdot}}$  which satisfy the following properties for all $A,B\in D(\delta_{0})$\begin{itemize}
		    \item [(i)] $ \delta_{0}(A)^{*} = \delta_{0}(A^{*})$
            \item [(ii)] $\delta_{0}(AB) = \delta_{0}(A)B +A\delta_{0}(B)$
		\end{itemize}
		and for which the subspace $E(\delta_{0}):=\{f \in E\,|\, W^{0}(f)\in D(\delta_{0})\}$ is dense $E$.
		We say that the derivation is \textit{continuous} if \begin{equation}
			\mathbb{C}\ni \lambda \rightarrow \delta_{0}(W^{0}(\lambda f)) \in C(E'_{\norm{\cdot}})
		\end{equation} is pointwise continuous for all $f \in E$. 
    Furthermore, we say that a derivation is \textit{linear} if
    \begin{equation} W^{0}(-f)\delta_{0}(W^{0}(f))\,,\end{equation}
    is a $\mathbb{R}$-linear functional on $E'_{\norm{\cdot}}$ for all $f\in E(\delta_{0})$.  
	\end{definition}
The following proposition asserts that a continuous, linear weak derivation can be explicitly represented in terms of the field functions $\Phi_{0}(\cdot)$ and a linear operator on $E$ with domain given by $E(\delta_{0})$. We will refer to this operator as the \textit{associated operator} of $\delta_{0}$.
	{\begin{proposition}\label{prop: caratterizzazione derivazioni continue I}
		Let $E$ be a normed, symplectic space. Then, $\delta_{0}: D(\delta_{0})\rightarrow C(E'_{\norm{\cdot}})$  is a continuous and linear weak derivation if and only if for all $W^{0}(f) \in D(\delta_{0})$ \begin{equation}
			\delta_{0}(W^{0}(f)) = i\Phi_{0}(L_{0} f)W^{0}(f),
		\end{equation} where $L_{0}: D(L_{0})\subset E \rightarrow E$ is a linear operator on the normed space $E$ with $D(L_{0}) = E(\delta_{0})$ a dense domain in $E$. 
	\end{proposition}}
	
	Now, we introduce a weaker notion of closability and of closed derivation that will help us to understand when the domain of definition of $\delta_{0}$, $D(\delta_{0})\subset \Delta(E,0)$ is in some sense, maximal.
	\begin{definition} \label{def: pointwise closure}(\textbf{Pointwise closure})
		Let $E$ be a normed, symplectic space. We say that a weak derivation $\delta_{0}$ is \textit{pointwise closed} if for every $\{f_{n}\}_{n\in\mathbb{N}}\subset E(\delta_{0})$ such that $f_{n}\rightarrow f\in E$ in the $\sigma(E,E_{\norm{\cdot}}')$-topology and $\delta_{0}(W^{0}(f_{n}))\rightarrow \Psi \in C(E'_{\norm{\cdot}})$ pointwise, it follows that $W^{0}(f)\in D(\delta_{0})$ ($f\in E(\delta_{0})$) and $\Psi = \delta_{0}(W^{0}(f))$.
	\end{definition} The following proposition shows that, for a linear and continuous weak derivation $\delta_0$, closability of $\delta_0$ is equivalent to the closability of the associated operator $L_0$.
	
		\begin{proposition}\label{prop: chiusura pointwise I} Given $\delta_{0}$ a continuous, linear weak derivation, it follows that $\delta_{0}$ is pointwise closed if and only if the associated operator $L_{0}$ is closed. Moreover, it follows that $\delta_{0}$ is closable if and only if $L_{0}$ is closable.
	\end{proposition}

 \begin{remark}
     {Putting together Propositions \ref{prop: caratterizzazione derivazioni continue I} and \ref{prop: chiusura pointwise I} we see that Eq. \eqref{eq: tentative definition for derivations} defines a weak, continuous and closed derivation with $H = H^*$ as associated operator. Thus, we can associate a derivation for every one-parameter group $\tau_{0,t}(W^0(f)) = W^0(e^{iHt}f)$. However, self-adjointness of $H$ is not a necessary requirement to have a well defined closed and continuous weak derivation.}
 \end{remark}
 
		Eq. \eqref{eq: computation of weak derivations} allows us to define a classical KMS condition for the Weyl $C^*$-algebra $\mathcal{W}(E,0)$:
	
		\begin{definition}\label{def: classical KMS condition} (\textbf{Weak classical KMS condition})
		Let $E$ be a normed symplectic space, $\beta\in \mathbb{R}$, and $\delta_0$ a pointwise closed, continuous and linear weak derivation. Then, if $\omega_{0}\in \mathcal{S}(\mathcal{W}(E,0))$ is a $C^2$ state as in Rmk. \ref{remark: regularity of states}, we say that $\omega_0$ satisfies the weak $(\delta_{0},\beta)$-KMS property if \begin{equation}\label{eq: KMS twisting}
			\omega_{0}(\{a,b\}) = \beta\omega_{0}(b\delta_{0}(a)),\quad \forall\; a,b\in D(\delta_{0}),
		\end{equation}
        where the Poisson bracket has been defined in \eqref{eq: Poisson parenthesis}.
	\end{definition}
		We observe that, having defined the Poisson bracket only on $\Delta(E ,0)$, it is not relevant for definition  \ref{def: classical KMS condition} to extend further $\delta_{0}$ outside of $\Delta(E,0)$.
        Moreover, note that in \ref{def: classical KMS condition} we cannot require less than $C^2$-regularity on the classical state $\omega_{0}$, since we need to compute the expectation value of $\Phi_0(f),\;\text{for } f\in E$; {see the discussion on regularity in Rmk. \ref{remark: regularity of states}.}
The foregoing definition is suitable for Weyl $C^*$-algebras since it does not involve continuity properties of time evolution. One can also give a definition more in line with the approach of $W^*$-dynamical systems and verify that it is satisfied by weak KMS states, see appendix \ref{sec: appendice C}.
    
	\section{Berezin quantization for Weyl algebras}\label{sec: stricti deformation quantization} This section is devoted to the construction of a  quantization map on the full classical Weyl $C^*$-algebra, designed to possess good continuity and structural properties. The definition given in Eq. \eqref{eq: quantizzazione astratta sugli elementi di weyl} is motivated by the Berezin–Toeplitz quantization studied in \cite{Coburn_1992}, as well as by analogous constructions for the Resolvent algebra in \cite{Nuland_2019}. We remind the reader that in our setting $(E,\sigma)$ is a symplectic vector space descending from a complex Hilbert space $(\mathcal{H},\langle\cdot,\cdot\rangle)$ with $\sigma(\cdot,\cdot) = \Im{\langle\cdot,\cdot\rangle}$ and $\norm{f} = \langle f, f\rangle^{1/2},\;\text{for } f\in E$.

{In Proposition~\ref{prop: positivita e continuita della quantizzazione astratta}, we show that the resulting map inherits several key features from the Berezin quantization map, which serves as the main tool throughout our proof. For clarity, we summarize in the following lemma the main properties of the Berezin map, as established in \cite{Berezin}, \cite[Ch. II, Sec.2.3]{Landsman_1998}, and \cite[Thm 2.8]{Moretti_VanDeVen_2021}.
\begin{lemma}\label{lemma: berezin properties}
    The family of maps  \begin{align*} &Q_h^B \colon \mathbb{L}^\infty\left(\mathbb{R}^{2\ell},\frac{d^\ell q d^\ell p}{(2\pi h)^n}\right)\to \mathcal{B}(\mathbb{L}^2(\mathbb{R}^\ell)), \quad h>0\\
    &Q_0^B = I_{\mathbb{L}^\infty},\quad h=0,
    \end{align*} is defined for $h>0$ by the weak integral \begin{equation}
        Q_h^B(f):= w-\int_{\mathbb{R}^{2\ell}}\frac{d^\ell q\; d^\ell p}{(2\pi h)^\ell}f(q,p) |\psi^{q,p}_h\rangle\langle \psi_h^{q,p}|,
    \end{equation}where $\psi^{q,p}_h\in \mathbb{L}^{2}(\mathbb{R}^{\ell}, \dd^{\ell}x)$ is the \textit{coherent states} defined as  $\psi^{q,p}_h(x):= ( h\pi)^{-\ell/4}e^{-i\frac{q\cdot p}{2 h}}e^{i\frac{p\cdot x}{ h}}e^{-\frac{(q-x)^{2}}{2 h}}$. $Q_h^B$ enjoys the following properties. \begin{enumerate}
    \item[(1)] $Q_h^B$ restricted to $C_c^\infty(\mathbb{R}^{2\ell})$ defines a strict deformation quantization as in Def. \ref{def: SDQ} (\textbf{SDQ});
          \item[(2)] $\norm{Q_h^B(f)}\leq \norm{f}_\infty$ for any $f\in \mathbb{L}^\infty\left(\mathbb{R}^{2\ell},\frac{d^\ell q d^\ell p}{(2\pi h)^n}\right) $ (\textbf{norm continuity});
        \item[(3)] $f\in \mathbb{L}^\infty\left(\mathbb{R}^{2\ell},\frac{d^\ell q d^\ell p}{(2\pi h)^n}\right)$ and $ f\geq 0$, except for a set of zero Lebesgue measure, implies $Q^B_h(f)\geq 0$ (\textbf{positivity}).
    \end{enumerate}
\end{lemma}}

\subsection{Construction of the quantization map}\label{subsec: definition of stricti deformation quantization}

	In this subsection we construct a positive and continuous quantization map, satisfying the SDQ requirements \ref{def: SDQ}. We start by writing down its action directly on  Weyl elements {\begin{definition}(\textbf{Abstract quantization map})\label{def: abstract quantization} We define a net of linear maps \begin{equation}(Q_h\colon \Delta(E,0) \to \Delta(E,h\sigma))_{h\in [0,+\infty)}\end{equation} by linearly extending the following action on Weyl elements\begin{align}\label{eq: quantizzazione astratta sugli elementi di weyl}
	   & Q_{h}(W^{0}(f)) := e^{-\frac{h}{4} \norm{f}^{2}} W^{h}(f),\quad h\in [0,+\infty),\; f\in E
       .
	\end{align}  
	\end{definition}}

 To discuss the positivity of the quantization map, we will need the following simple lemma

  \begin{lemma}\label{lemma: positività indotta da rappresentazioni}
 				 	Given a unital $ C^*$-algebra $\mathfrak{A}$ and a faithful representation $(\mathcal{H},\pi)$, it follows that $A\in\mathfrak{A}$ is a positive element if and only if $\pi(A)\in\mathcal{B}(\mathcal{H})$ is positive.
 			\end{lemma}
 			\begin{proof}
 					$(\Rightarrow)$ $A\in \mathfrak{A}$ is positive if and only if $A = CC$ for some $C\in\mathfrak{A}$ self-adjoint operator (see \cite[Th. 2.10]{Bratteli_Robinson_87}) and $\pi(A) = \pi(CC) = \pi(C)\pi(C)$, which is positive in $\mathcal{B}(\mathcal{H})$. 
 				$(\Leftarrow)$ If $A\in\mathfrak{A}$ was not positive but $\pi(A)$ was, then, there would exists a $\lambda\notin  \mathbb{R}_{+}$ such that $A-\lambda I$ is not invertibile in $\mathfrak{A}$, whereas $\pi(A)-\lambda I = \pi(A - \lambda I)$ is invertible in $\mathcal{B}(\mathcal{H})$. We know \cite[Prop. 2.3.1]{Bratteli_Robinson_87} that $\pi(\mathfrak{A})$ is a $C^*$-subalgebra of $\mathcal{B}(\mathcal{H})$ and that \cite[Prop. 2.2.7]{Bratteli_Robinson_87}
  $\text{Sp}_{\mathcal{B}(\mathcal{H})}(\pi(A)) = \text{Sp}_{\pi(\mathfrak{A})}(\pi(A))$ 
     for every $A\in\mathfrak{A}$. Then, $\lambda\notin \text{Sp}_{\mathcal{B}(\mathcal{H})}(A) = \text{Sp}_{\pi(\mathfrak{A})}(A)$ and so, $\pi(A) -\lambda I$ is invertible in $\pi(\mathfrak{A})$. But  since $\pi$ is faithful, if $\pi(B) = \pi(A -\lambda I)^{-1}$, it follows that $B = (A-\lambda I)^{-1}\in \mathfrak{A}$ and this is absurd. 
 			\end{proof}

	To extend $Q_{h}$, we estimate the norm of generic elements $Q_{h}(c)$ for $c\in \Delta(E,0)$. We have the following
	\begin{proposition}\label{prop: positivita e continuita della quantizzazione astratta} 
				The abstract quantization map defined on \eqref{eq: quantizzazione astratta sugli elementi di weyl} satisfies the following. \begin{itemize}
				    \item [(i)] $\norm{Q_h(c)}_h \leq \norm{c}_0$ for every $c\in\Delta(E,0)$;
                    \item [(ii)] for every $c\in\Delta(E,0)$, $c\geq 0$ implies that $Q_h(c)\geq 0$ as elements of the respective $C^*$-algebras $\mathcal{W}(E,0)$ and $\mathcal{W}(E,h\sigma)$.
				\end{itemize}
			As a consequence, $Q_h$ on $\Delta(E,0)$ can be extended to a positive map\begin{equation}
				Q_{h}: \weylEO\rightarrow \mathcal{W}(E,h\sigma).
			\end{equation}
 	\end{proposition}
 	\begin{proof}
 
 		{Given an arbitrary $c\in \Delta(E,0)$, this can be written  as a finite linear combination of classical Weyl elements
 		\begin{equation}\label{eq: general element c}
 			c= \sum_{k=1}^{n}z_{k}W^{0}(f_{k}),
 		\end{equation}
         with $(f_k)_{k=1}^n\subset E$, $(z_k)_{k=1}^n\subset \mathbb{C}$.   The action of $Q_{h}$ on $c$ is read by linearity from equation \eqref{eq: quantizzazione astratta sugli elementi di weyl} as\begin{equation}\label{eq: general element c}
 			Q_{h}(c) = \sum_{k=1}^{n}z_{k}e^{-h\frac{\norm{f_{k}} }{4}} W^{h}(f_{k}).
 		\end{equation}
        Now, we construct a real vector space of dimension $2\ell$,  $E_{\ell} := \mathbb{C}\text{-span}\{f_{1},\dots, f_{n}\} $, where $\ell$ is the dimension of the complex span just defined. Then,we consider the Weyl $C^*$-algebra $\mathcal{W}(E_{\ell}, 0)$ and we identify a symplectic basis for $E_{\ell}$:  $\{g_{1},ig_{1},\dots, g_{\ell},ig_{\ell}\}$.  Let us focus on the generic term  $W^{h}(f_{k})$ in Eq. \eqref{eq: general element c} for a fixed $k$. We can expand $f_k$ as \begin{equation}\label{eq: espansione in base}
            f_k = \sum_{j=1}^\ell \left(\lambda_j g_j +\mu_j ig_j\right),
        \end{equation} for some real coefficients $(lambda_j,\mu_j)_{j=1}^\ell$. Now, we can exploit the basis expansion \eqref{eq: espansione in base} together with the chain of isomorphisms \begin{equation} \mathcal{W}(E_\ell,0)\simeq \text{AP}(E_\ell) \simeq \text{AP}(\mathbb{R}^{2\ell}),\end{equation}  to identify the classical Weyl element $W^0(f_k)$ with the almost periodic function \begin{equation}
          \mathbb{R}^\ell \ni (q,p)\to  W^0(f_k)(q,p) = \exp{ i\sum_{j=1}^\ell \lambda_j q_j + \mu_j p_j}\in\mathbb{C}.
        \end{equation}
 Similarly, the quantized Weyl element $W^h(f_h)$ can be written as a unitary operator on $\mathbb{L}^{2}(\mathbb{R}^{\ell})$ by means of the Schr\"{o}dinger representation \ref{ex: schrodinger representation} \begin{equation}
 			\pi_{S}(W^{ h}(f_{k}))  = \exp{i\overline{\sum_{j=1}^{\ell} \lambda_{j}\hat{X}_{j} + \mu_{j}\hat{P}_{j}}}.
 		\end{equation}
  Now, we verify that $\pi_S(W^h(f_k))$ can be obtained from the Berezin quantization \ref{lemma: berezin properties} of $W^0(f_k)\in \text{AP}(\mathbb{R}^{2\ell})\subset \mathbb{L}^\infty(\mathbb{R}^2\ell, \frac{d^\ell q d^\ell p}{(2\pi h)^\ell})$, so to obtain an identification between $\pi_S\circ Q_h$ and $Q_h^B$.}
        To do so, it is sufficient to evaluate the sesquilinear form 
 		\begin{equation} \mathbb{L}^{2}(\mathbb{R}^{\ell})\cross \mathbb{L}^{2}(\mathbb{R}^{\ell}) \ni (\varphi,\psi)\rightarrow \langle \varphi,Q_{ h}^{B}(W^0(f_k))\psi\rangle,
 		\end{equation}for $\psi,\varphi \in\mathcal{S}(\mathbb{R}^{\ell})$.
  By direct inspection we have:
   \begin{align}\label{eq: computation Berezin}
 			\langle \varphi, Q^{B}_{ h}(W^0(f_k))\psi\rangle &= ( h\pi)^{-\ell/2}\int \frac{\dd^{\ell}q\dd^{\ell}p}{(2\pi h)^{\ell}} e^{i(\lambda\cdot q +\mu\cdot p)}\int \dd^{\ell} x e^{-i\frac{p\cdot x}{h}}e^{-\frac{(q-x)^{2}}{2 h}}\psi(x)\int \dd^{\ell} y e^{i\frac{p\cdot y}{ h}}e^{-\frac{(q-y)^{2}}{2 h}}\overline{\varphi(y)}\nonumber\\
 			&= (\pi h)^{-\ell/2}\int \frac{\dd p^{\ell}}{(2\pi h)^{\ell}} \int \frac{\dd^{\ell} x}{(2\pi h)^{\ell/2}} e^{-ip\cdot(x -y -\mu h)/ h} \int \dd^{\ell}y \dd^{\ell}q e^{\frac{-(q-x)^{2}}{2 h}} e^{\frac{-(q-y)^{2}}{2 h}}e^{i\lambda \cdot q }\psi(x)\overline{\varphi(y)} \nonumber \\
 			&=(\pi h)^{-\ell/2}\int \dd^{\ell} y\dd^{\ell}q e^{-\frac{(q-y-\mu h)^{2}}{2 h}}e^{-\frac{-(q-y)^{2}}{2 h}}e^{i\lambda\cdot q}\psi(y+ h\mu)\overline{\varphi(y)} \nonumber\\
 			&= e^{-\frac{ h}{4}(\lambda^{2} + \mu^{2})}\int \dd^{\ell}y \overline{\varphi(y)}e^{i\lambda\cdot y}e^{i h \frac{\lambda\cdot\mu}{2}}\psi(y + h \mu).
 		\end{align}
 		The equality in Eq. \eqref{eq: computation Berezin} can be read as follows
 		\begin{align}
 			\langle\varphi, Q_{ h}^{B}(W^0(f_k)) \psi\rangle & = e^{-\frac{ h}{4}(\lambda^{2} + \mu^{2})}\langle \varphi, \pi_{S}(W^{ h}(f_{k}))\psi\rangle \nonumber \\
 			& = \langle\varphi, \pi_{S}(Q_{h}(W^{0}(f_{k})))\psi \rangle,\quad\quad \text{for all } \varphi,\psi\in \mathcal{S}(\mathbb{R}^{\ell}).
 		\end{align}
 		By density of $\mathcal{S}(\mathbb{R}^{n})$ we have $Q^{B}_{ h}(W^0(f_k)) = \pi_{S}(Q_{ h}(W^{0}(f_k)))$. Then, by linearity of $Q_h, Q^B_h$ and $\pi_S$, we conclude that $\pi_{S}(Q_{h}(c)) = Q^{B}_{ h}(c)$. {Thanks to the faithfulness of Schr\"{o}dinger representation and to the equality $\norm{\cdot}_{0} = \norm{\cdot}_{\infty}$ between the norms of $\mathcal{W}(E_\ell,0)$ and $C_b(\mathbb{R}^{2\ell})$, we can prove the norm continuity (i) as 
   \begin{equation}\label{eq: norm estimate with sch}
 			\norm{Q_{h}(c)}_{ h} = \norm{\pi_{S}(Q_{ h}(c))} = \norm{Q^{B}_{ h}(c)} \leq \norm{c}_{\infty} = \norm{c}_{0},
 		\end{equation} where the last inequality exploit the norm continuity property, lemma \ref{lemma: berezin properties} (2).  Since $c$ was arbitrary, estimate \eqref{eq: norm estimate with sch} is valid on all $\Delta(E,0)$.}
   
   Thanks to the norm-continuity of the abstract quantization we can extend this map to the whole $\mathcal{W}(E,0)$ by a standard density argument.
   Indeed, if $(c_{n})_{n\in\mathbb{N}}\subset \Delta(E,0)$ is a Cauchy sequence, converging to some element $c\in \mathcal{W}(E,0)$, then the sequence $\{Q_{ h}(c_{n})\}_{n\in\mathbb{N}}\subset \Delta(E,h\sigma)$ satisfies \begin{equation}
 			\norm{Q_{h}(c_{n})- Q_{h}(c_{m})} _{ h}\leq \norm{c_{n}-c_{m}}_{0}< \epsilon,
 		\end{equation} for $n,m$ sufficiently large. So, the sequence $	(Q_{h}(c_{n}))_{n\in\mathbb{N}}$ is Cauchy in $\mathcal{W}(E,  h\sigma)$. If we define $Q_{ h}(c) := \lim_{n\rightarrow\infty}Q_{h}(c_{n})$, we obtain a well-defined quantization satisfying \begin{equation}
 			\norm{Q_{ h}(c)}_{ h} = \lim_{n\rightarrow\infty}\norm{Q_{ h}(c_{n})}_{ h}\leq \lim_{n\rightarrow\infty}\norm{c_{n}}_{0} = \norm{c}_{0}
 		\end{equation}

 		{Positivity of $Q_h$ follows from an application of lemma \ref{lemma: berezin properties} (3) and lemma \ref{lemma: positività indotta da rappresentazioni}: if $c\in \Delta(E,0)\subset AP(E)$ is a positive function we have $\pi_{S}(Q_{ h}(c))=Q^{B}(c)\geq 0$, which implies $Q_{ h}(c)\geq 0$.}
 	\end{proof}

 	In \cite[Th. 5.6]{Binz_Honegger_Roecker_2024_2} the authors proved that Weyl quantization $(Q^{W}_{ h})_{ h\in [0,+\infty)}$, together with the Poisson bracket \eqref{eq: Poisson parenthesis}, constitutes a strict deformation quantization of $(\Delta(E,0),\{\cdot,\cdot\},(\mathcal{W}(E,h\sigma))_{h\in [0,+\infty)})$. Now, the function
 	\begin{equation}\label{eq: quotient factor}
 		w\colon\mathbb{R}\cross E \ni ( h,f)\rightarrow w( h,f):= \exp{- h\frac{\norm{f}^{2}}{4}} \in \mathbb{C},
 	\end{equation}
 	satisfies $w( h,f)\in(0,1]$, $w( h,f) = w( h, -f)$, $w(0,f) = w( h, 0) = 1$ for every $f\in E$, $ h\in[0,+\infty)$, and moreover $\mathbb{R}\ni h\rightarrow w( h, f)\in \mathbb{C}$ is continuous for fixed $f$ and {locally bounded}. These properties make the function in \eqref{eq: quotient factor} a \textit{quantization factor}, as defined in \cite[Def. 4.1]{Honegger_Rieckers_2005}. {Then, it follows from \cite[Th. 4.4, Cor. 4.7]{Honegger_Rieckers_2005} that}
  \begin{theorem} \label{th: la quantizzazione astratta è SDQ}The data $(\mathcal{W}(E, h\sigma), Q_{ h})_{ h\in [0,+\infty)}$ together with the Poisson sub-algebra $(\Delta(E,0),\{\cdot,\cdot\})$ define a strict deformation quantization as in definition \ref{def: SDQ}.
 	\end{theorem}
 	In particular, by Theorem \ref{th: la quantizzazione astratta è SDQ}, the quantization map $Q_h$ is injective on $\Delta(E,0)$. 

    \subsection{Further properties of the quantization map}\label{sec: further properties}
 	In this subsection, we prove the injectivity and non-surjectivity on the full $C^*$-algebra $\mathcal{W}(E,0)$ of the quantization map defined in Def. \ref{def: abstract quantization}.
 To this avail, we pick the \textit{canonical central state} $\omega_{c}^{h}$ \cite[Prop. 2.17]{Manuceau_Sirugue_Testard_Verbeure_73}.
  The latter is defined by linear extension of the functional
  \begin{equation}\label{eq: canonical central state}
 		\omega_{c}^{h}(W^{h}(f)) = \begin{cases}
 			1,\quad f = 0\\
 			0,\quad f\neq 0
 		\end{cases},\quad f\in E,\; h\in [0,\infty).
 	\end{equation} Since $\omega_{c}^{h}$ is faithful on $\Delta(E,h\sigma)$ the following positive map \begin{equation}
 	\norm{C}_{h,2}:= \sqrt{\omega_{c}^{h}(C^* C)},\quad C\in \mathcal{W}(E,h\sigma),
 	\end{equation} is a norm \cite[Lem. 3.1]{Manuceau_Sirugue_Testard_Verbeure_73}. By completing the $^*$-algebra $\Delta(E,h\sigma)$ with the foregoing norm we obtain a Banach $^*$-algebra $\overline{\Delta(E,h\sigma)}^{2} \supset \mathcal{W}(E,h\sigma)$, with $\norm{\cdot}_{2}\leq \norm{\cdot}$. Moreover, we have the following\begin{lemma}\label{lemma: espansione degli elementi in Weyl}
 	Every element $A\in \overline{\Delta(E,h\sigma)}^{2}$ can be written as \begin{equation}
 		\sum_{f\in E} \mu(f)W^{h}(f) = \sum_{f\in E}\omega_{c}^{h}(W^{h}(-f)A) W^{h}(f),
 	\end{equation} for some $\{\mu(f),\; f\in E\}\subset \mathbb{C}$ where the sum converges in norm $\norm{\cdot}_{h,2}$.
 	 	\end{lemma}
 	 	\begin{proof}
 	 		We verify that $\Delta(E,h\sigma)\cross \Delta(E,h\sigma)\ni (A,B)\rightarrow \omega_{c}^{h}(A^{*}B)$ induces a scalar product on $\overline{\Delta(E,h\sigma)}^{2}$ and that the set $\{W^{h}(f),\; f\in E\}$ forms an orthogonal basis. The first step is clear by the properties of states and the faithfulness of $\omega^{h}_{c}$, while for the second step we have orthogonality by definition of the state \eqref{eq: canonical central state} and by the density of $\Delta(E,h\sigma)$ in $\overline{\Delta(E,h\sigma)}^{2}$.
     In particular, $\overline{\Delta(E,h\sigma)}^{2}$ is an Hilbert space with scalar product $\langle A|B\rangle:=\omega_h^c(A^*B)$ and the lemma follows from standard results of Hilbert space theory \cite[Chap. 2]{Kadison_Ringrose_91}.
 	 	\end{proof}
    Now, we can extend the abstract quantization map in a continuous way to a linear map  \begin{equation}Q_{h}\colon\overline{\Delta(E,0)}^{2}\rightarrow\overline{\Delta(E,h\sigma)}^{2}\,.\end{equation}
    Indeed, for every element $c\in \Delta(E,0)$ we have an expansion in terms of a finite number of Weyl elements $c= \sum_{k=1}^{n}z_{k}W^{0}(f_{k})$ and the following estimate for the norms\begin{equation}\label{eq: continuita con norma 2}
 	 	\norm{Q_{h}(c)}_{h,2}^{2} = \sum_{k=1}^{n}|z_{k}|^{2} e^{-h\frac{\norm{f_{k}}^{2}}{2}} \leq  \sum_{k=1}^{n}|z_{k}|^{2} = \norm{c}_{0,2}^{2}.
 	 	\end{equation} Thanks to  this extension we are able to prove the following
 	
 	 \begin{proposition}\label{prop: iniettività della mappa}
 		Let $Q_{h}\colon \mathcal{W}(E,0)\rightarrow \mathcal{W}(E,h\sigma)$ be the abstract quantization map defined in \eqref{eq: quantizzazione astratta sugli elementi di weyl}. This map is injective on the full $C^*$ algebra $\mathcal{W}(E,0)$, but it is not surjective.
 	\end{proposition}\begin{proof}
  Injectivity of the map comes from its injectivity when acting on $\overline{\Delta(E,0)}^{2}$. Indeed, given $c\in \overline{\Delta(E,0)}^{2}$, by lemma  \ref{lemma: espansione degli elementi in Weyl} and the continuity expressed in equation \eqref{eq: continuita con norma 2} we have \begin{equation}
  	\norm{Q_{h}(c)}_{h,2}^{2} = \sum_{f\in E} |\omega_{c}^{0}(W^{0}(-f)c)|^{2}e^{-h\frac{\norm{f}^{2}}{2}},
  \end{equation} which is equal to $0$ iff $\omega_{c}^{0}(W^{0}(-f)c) = 0$ for every $f\in E$, that is, iff $c =0$. Thanks to the inclusion $\mathcal{W}(E,0)\subset \overline{\Delta(E,0)}^{2}$, we have injectivity also for $Q_{h}|_{\mathcal{W}(E,0)}$. To discuss surjectivity we argue by contradiction. If the map $Q_{h}: \mathcal{W}(E,0)\rightarrow \mathcal{W}(E,h\sigma)$ was onto, then it would be an invertible, continuous linear map from a Banach space onto another Banach space.
  By \textit{Banach inversion theorem} this implies \cite[Th. 1.8.5]{Kadison_Ringrose_91} that also the inverse map $(Q_{h})^{-1} : \mathcal{W}(E,h\sigma)\rightarrow \mathcal{W}(E,0)$ is continuous, hence bounded. In other words, there exists some $Q>0$, such that $\norm{(Q_{h})^{-1}(A)}_{0}\leq Q\norm{A}_{h}$, for every $A\in \mathcal{W}(E,h\sigma)$. Now, it suffices to take an arbitrary $n\in\mathbb{N}$ and $ f\in E$, $f\neq 0$ to have \begin{equation}
  \norm{(Q_{h})^{-1}(W^{h}(nf))}_{0} = e^{h\frac{n^{2}\norm{f}^{2}}{4}} \leq Q\norm{W^{h}(nf)}_{h} = Q,
  \end{equation} which is clearly absurd if we pick $n\in\mathbb{N}$ large enough.
 	\end{proof}
 	
 	\begin{remark}\label{remark: surgettività di Berezin}
 	    We can exhibit  explicitly an element in $\mathcal{W}(E,h\sigma)$ without a counter-image in $\mathcal{W}(E,0)$:\begin{equation}\label{eq: esempio di elemento non raggiungibile}
 		C = \sum_{n=1}^{\infty}\frac{1}{n^{2}}W^{h}(nf),
 	\end{equation} where $f\in E$ is different from zero, but otherwise arbitrary. By direct inspection, the series in equation \eqref{eq: esempio di elemento non raggiungibile} converges with respect to the norm $\norm{\cdot}_{h}$. If there existed an element $c_{h}\in \mathcal{W}(E,0)$ such that $Q_{h}(c_{h}) = C $, then, its expansion would be 
  \begin{equation}
 	c_{h} = \sum_{n=1}^{\infty}e^{h\frac{n^{2}}{4}\norm{f}^{2}}\frac{1}{n^{2}}W^{0}(nf).
 	\end{equation} However, the latter series is divergent with respect to the norm $\norm{\cdot}_{0,2}$, which proves that $c_h\notin\mathcal{W}(E,0)$.
 \end{remark}

{\subsection{De-quantization procedure}\label{sec: dequantization} In this short subsection we present one of the major applications for $Q_h$. The quantization map sends classical observable into quantum ones, i.e. $Q_h\colon \mathcal{W}(E,0)\to \mathcal{W}(E,h\sigma)$. Hence, we can define its pullback action on the algebraic dual of the  Weyl $C^*$-algebras \begin{equation}
    Q_h^* \colon \mathcal{W}(E,h\sigma)^* \to \mathcal{W}(E,0)^*,\quad Q_h^*(\eta_h):= \eta_h\circ Q_h\quad \text{for }\eta_h\in \mathcal{W}(E,h\sigma)^*  .  
\end{equation} Now, as $Q_h$ is norm-continuous, positive and normalized as $Q_h(I_{\mathcal{W}(E,0)}) = I_{\mathcal{W}(E,h\sigma)}$, its pullback action maps quantum states to classical states \begin{equation}
    Q_h^* \colon \mathcal{S}(\mathcal{W}(E,h\sigma)) \to \mathcal{S}(\mathcal{W}(E,0)).
\end{equation}Hence, given a quantum state $\omega_h\in \mathcal{S}(\mathcal{W}(E,h\sigma))$, we can always define the classical state $\omega_h\circ Q_h$. The only inconvenience in the previous definition is the residual dependence on the semi-classical parameter $h$. This dependence can be \textit{removed} by taking the limit $h\to 0^+$. More precisely, for any net of quantum states $(\omega_{h})_{h\in(0,+\infty)}$, with $\omega_h\in \mathcal{S}(\mathcal{W}(E,h\sigma))$, we obtain a net of classical states as $(\omega_{h}\circ Q_h)_{h\in(0,+\infty)}$. Then, since $\mathcal{S}(\mathcal{W}(E,0))$ is $*$-weakly compact, we can always extract a subsequence $(\omega_{h_n}\circ Q_{h_n})_{n\in\mathbb{N}}$, with  $\lim_{n\to +\infty} h_n = 0$, converging $*$-weakly to some classical state $\omega_0\in \mathcal{S}(\mathcal{W}(E,0))$.}

{Other approaches to the \textit{de-quantization} of states, can be found in \cite{Ammari_Sohinger_2023,Falconi_2018, Falconi_Fratini_2024, Folland_1989} and references therein.}
	
	\section{Classical limit of a free Bose Gas}\label{sec: limite classico del gas libero}

    {This section is devoted to the study of classical equilibrium states for a Bose gas. The classical character of these states arises as a consequence of the infinite-density limits that we shall consider. More precisely, our results can be viewed as an extreme instance of the mean-field approach for a weakly interacting Bose gas, in which the interaction potential is set to zero and the local density diverges. There is a vast literature addressing the high-density regime of interacting Bose gases. 
The seminal work of Bogoliubov \cite{Bog_47}, which marked the starting point for the study of interactions, 
was originally intended to explain the behavior of liquid helium in this regime. 
More recent results were obtained by Lieb and Solovej \cite{Lieb_Solovej_2001,Lieb_Solovej_2004}, 
who rigorously validated Foldy’s and Dyson’s formulas for the ground-state energy of a high-density Bose gas with Coulomb interactions. 
Building on these, Giuliani and Seiringer \cite{GiulSeir2009} derived, in the mean-field and high-density setting at zero temperature, 
the celebrated Lee–Huang–Yang formula for the ground-state energy. 
Another approach was developed in \cite{Lewin_Nam_Rougerie_2014}, 
where the authors studied generic mean-field systems in the large particle-number limit 
and derived the Hartree approximation, which describes the dynamics of the gas via an effective nonlinear differential equation. }

{In the opposite limit, one finds the \textit{dilute gas}, or weak-coupling regime, 
characterized by the condition $\rho a^3 \ll 1$, where $\rho$ is the density and $a$ the scattering length of the interaction potential. 
A particular scaling consistent with this condition is the Gross-Pitaevskii  regime, 
where the interaction potential is rescaled with the number of particles $N$ as $N^2 V(N(x-y))$, 
corresponding to short-range repulsive forces and very rarefied gases. 
In this context, Lieb and Seiringer proved the occurrence of Bose-Einstein condensation in the zero-temperature limit \cite{Lieb_Seiringer_2002}. 
Recently, this result has been generalized and extended in several directions 
\cite{Boccato_Brennecke_Cenatiempo_Schlein_2018,Boccato_Brennecke_Cenatiempo_Schlein_2019,COSS_25,Deuchert_Seiringer_Yngvason_2019,Nam_Triay_2023}. }

{Another way to satisfy the dilute gas approximation is by considering interaction potentials of the form $\frac{1}{N} V(x-y)$. 
In \cite{GS_2013}, the authors computed the excitation spectrum of a weakly interacting Bose gas to leading order in the number of particles, 
a result which was later refined in \cite{Brietzke_Solovej_2020,DerNa_2013}.}

{The novelty of our analysis lies in the use of algebraic techniques to construct examples of non-trivial finite-temperature equilibrium states  within the infinite-density regime, where neither a local number operator nor a local density can be defined. These states are obtained, in a one-to-one correspondence, as semi-classical limits of the equilibrium states introduced in the seminal work of Araki and Woods\cite{Araki_Woods_63}. In particular, we will identify a class of classical infinite density states $\omega^{\alpha,0}_0$ describing a macroscopic occupation of the ground-state.}
	
{Before proceeding with the semi-classical analysis in \ref{sec: sistemi a volume finito},\ref{sec: limiti classici degli stati}, \ref{sec: limite termodinamico classico}, we introduce the standard algebraic formalism of quantum Bose gases in the following section.} 
 
	\subsection{Bose-Einstein condensation in quantum systems}\label{sec: condensazione quantistica} {In this subsection we discuss the algebraic description of a quantum free Bose gas. One of the first works to deal with this formulation is the seminal paper by Araki and Woods \cite{Araki_Woods_63}. In this, the authors constructed inequivalent representation of the CCR describing an infinite number of particles in the thermodynamic limit, i.e. representation of $\mathcal{W}(E,h\sigma)$, for specific symplectic spaces $E$, for which the number operator can be defined only locally; these are called \textit{strange representations}. Later, this work was expanded by Cannon, Lewis and Pulé \cite{Cannon_73, LP_73} (see also \cite{Berg_Lewis_86} and for a summary of the main results \cite[Ch. 5.2.5]{Bratteli_Robinson_97}) which derived the representations introduced in \cite{Araki_Woods_63} as specific GNS representation relative to infinite-volume states obtained by means of a thermodynamic limit of the finite volume Gibbs states.}
	
	 We fix a single-particle Hilbert space $(\mathcal{H},\langle\cdot,\cdot\rangle)$ and the symplectic space $(E,h\sigma)$, $h>0$, descending from $\mathcal{H}$. Then, $\mathcal{W}(E,h\sigma)$ can be represented on the \textit{Fock-Cook} Hilbert space \begin{equation}
		\mathbb{F}(\mathcal{H}): = \bigoplus_{n=0}^{\infty}\mathcal{H}_{n},
	\end{equation} where $\mathcal{H}_{n}$ is the symmetrized tensor product of $\mathcal{H}$, taken $n$-times. {This representation can be constructed by the mapping \begin{equation}
	    W^h(f) \to W^h_F(f) = e^{i\Phi_h(f)}, \quad \Phi_h(f) = \frac{a_h(f)+ a^*_h(f)}{\sqrt{2}},
	\end{equation} where $a_h(f),\; a_h^*(f)$ are the creation and annihilation operators relative to the cyclic vacuum state $\Omega_h$ \cite[Sec. 2]{Araki_Woods_63}:\begin{equation}
	    \langle \Omega_h, W^h_F(f) \Omega_h\rangle = e^{-\frac{h}{4}\norm{f}^2}.
	\end{equation}Moreover, the Fock-Cook representation is unitary equivalent to the GNS representation of the state \begin{equation}
	    \omega_F(\cdot):= \langle \Omega_F,\cdot\Omega_F\rangle.
	\end{equation}
    The operators $a_h(f),\;a^*_h(f)$ are unbounded, closed and satisfy $a_h(f)^* = a^*_h(f)$. These are completely specified by the CCR \begin{equation}
	    [a_h(f), a_h^*(g)] = h\langle f,g\rangle, \quad [a_h(f),a_h(g)] = [a_h^*(f),a_h^*(g)]  =0,\quad f,g\in E,
	\end{equation} and by their action on the vacuum state \begin{equation}
	    a_h(f)\Omega_h = 0,\quad \frac{1}{\sqrt{n! h^n}}a^*_h(f_n)\dots a^*_h(f_1)\Omega_h = P_+(f_n\otimes\dots \otimes f_1), \quad f,f_1,\dots, f_n\in E,
	\end{equation} where $P_+$ is the projection onto $\mathcal{H}_n$. Independently from $E$, there is always a well defined number operator on $\mathbb{F}(\mathcal{H})$ given by \begin{equation}
	    N_F = \frac{1}{h}\sum_{n=1}^{+\infty}a^*_h(f_n)a_h(f_n), 
	\end{equation} where $(f_n)_{n>0}$ is an arbitrary orthonormal basis of $E$.}
    
    {Now, we discuss states and representations which describes a non-trivial particle distribution in momentum space $\rho({\bm p})$. We will address this situation by using the \textit{gran canonical formalism}; details concerning the canonical formalism can be found in \cite{Araki_Woods_63}[Sec. 4] and \cite{Cannon_73}.}

   { We start with the finite volume case. Take $E_L:= \mathbb{L}^2(\Lambda_L)$, where $\Lambda_L:= [-L,L]^\nu\subset \mathbb{R}^\nu$, $L>0$} \footnote{{It is possible to generalize all the results to more general net of finite volumes, e.g. rectangular boxes with suitable conditions on the length of the edges, see \cite{Berg_Lewis_86} and \cite{Bratteli_Robinson_97}[Th. 5.2.32] }} and $E:= \mathbb{L}^2(\mathbb{R}^\nu)$. On $E_L$, we can define an Hamiltonian $H_L$ given by some self-adjoint extension on $\mathbb{L}^2(\Lambda_L)$ of the operator $-\Delta/2|_{C_c^\infty(\Lambda_L)}$. $\Delta$ has many self-adjoint extensions; for clarity we will employ the one obtained by \textit{Dirichlet boundary conditions}\cite[Th. 5.2.30]{Bratteli_Robinson_97}. The spectrum of $H_L$ is characterized by the following eigenvectors and eigenvalues, labeled by $\underline{n}\in \mathbb{N}^{*,\nu}$\footnote{With the notation $\mathbb{N}^{*,\nu}$ we mean the Cartesian product of $\mathbb{N}^{*}$ $\nu$-times, where $\mathbb{N}^{*}$ are the natural numbers, zero excluded.
 }
	\begin{equation}\label{eq: autovettori e energie del laplaciano}
		\psi_{\underline{n}}(x) = \sqrt{\frac{|\underline{n}|2^{\nu}}{|\Lambda_{L}|}}\prod_{i=1}^{\nu}\sin(\pi  \frac{n_{i}(x_{i}- L)}{2L}),\quad E_{\underline{n}}(L) = \frac{1}{2}\sum_{i=1}^{\nu}\frac{\pi^{2}n_{i}^{2}}{(2L)^{2}}.
	\end{equation} The ground state eigenfunction and energy will be denoted by $\psi_{0}$ and $E_{0}(L)$.
	 The resulting Hamiltonian on the Fock-Cook space $\mathbb{F}(E_L)$ is $\dd\Gamma(K_{\mu,L}) = \dd\Gamma(H_{L}) -\mu N_{L}$, where $N_{L}$ is the self-adjoint number operator on $E_L$, $\mu\in\mathbb{R}$ is the \textit{chemical potential}, while \begin{equation}
		\dd\Gamma (H_{L}) = \bigoplus_{n=0}^{\infty} H_{L}\otimes I \otimes \dots \otimes I + \dots + I\otimes \dots \otimes I\otimes H_{L},
	\end{equation} is the second quantization of the Hamiltonian $H_{L}$. Then, if $\mu < E_{0}(L)$, the operator $ e^{-\beta h\dd\Gamma(K_{\mu,L})}$ is trace class, and we can introduce the \textit{Gibbs-von Neumann states}  on the Weyl $C^*$-algebra\begin{equation}\label{eq: Gibbs state}
	\omega^{L,\mu}_{h}(W^h(f)) = \frac{\Tr{e^{-\beta h\dd\Gamma(K_{\mu,L})} W^h_F(f)}}{\Tr{e^{-\beta h \dd\Gamma(K_{\mu,L}) }}}, \quad W^h(f)\in \mathcal{W}(E_{L},h\sigma).
	\end{equation} 
    {We have inserted a factor of $h$ in front $\beta$ at the exponent in \eqref{eq: Gibbs state}; as we will see, this choice connects the semi-classical parameter to the mean density of the gas.
	From the explicit expression of $\omega_h^{L,\mu}$, it is clear that these states define a normal representation which is normal to the Fock-Cook one. In particular, it is possible to define a number operator $N_{\omega_h^{L,\mu}}$.} These states can be extended to arbitrary products of the field operators $\{\Phi_{h}(f),\;f\in E_{L}\}$. An easy computation      \cite[Appendix 1]{Araki_Woods_63} shows that two-point functions are expressed in terms of bi-linear forms on $E_{L}$:
    \begin{align}
		\omega^{L,\mu}_{h}(\Phi_{h}(f)\Phi_{h}(g))
        & = \frac{h}{2}\Re{\langle f ,(I +e^{-\beta h K_{\mu,L}})(I - e^{-\beta h K_{\mu,L}})^{-1}g\rangle} +\frac{ih}{2}\sigma(f,g)\,,
        \\
        \omega^{L,\mu}_{h}(W^{h}(f))
        &= \exp{-\frac{h}{4}\langle f, (I +e^{-\beta h K_{\mu,L}})(I - e^{-\beta h K_{\mu,L}})^{-1}f\rangle}\,,
    \end{align}
    for every $f,g\in E_{L}$. With these values it is possible to calculate the associated density of the system as
    \begin{multline}\label{eq: densità quantistica su scatola}
        \rho_{h}^{\mu,L}
        :=\frac{1}{|\Lambda_{L}|}\omega^{L,\mu}_{h}(N_{L})
        = \frac{1}{h|\Lambda_{L}|}\sum_{\underline{n}\in\mathbb{N}^{\nu}}\omega^{L,\mu}_{h}(a_h^*(\Psi_{\underline{n}})a_h(\Psi_{\underline{n}}))
        \\
        = \frac{1}{|\Lambda_{L}|}\sum_{\underline{n}\in\mathbb{N}^{\nu}}z_{h} e^{-\beta h E_{\underline{n}(L)}}(1- z_{h}e^{-\beta h E_{\underline{n}}(L)})^{-1}\,,
    \end{multline}
 where we have defined $z_{h}:= e^{\beta h \mu}$.

{Now, we introduce infinite-volume CCR representations. Existence of the number operator, reducibility and existence of a cyclic vector were discussed by Araki and Woods \cite[Sec. 4 \& 5]{Araki_Woods_63}, while Cannon \cite{Cannon_73} derived the latter representations from a thermodynamic limit of the states in Eq. \eqref{eq: Gibbs state}. We will follow this latter point of view.}

{We consider the Weyl algebra $\mathcal{W}(E_L,h\sigma)$ as a sub $C^*$-algebra of the inductive limit \begin{equation}
    \mathcal{W}(E_L, h\sigma)\subset \overline{\bigcup_{L'>0} \mathcal{W}(E_{L'},h\sigma)} \subset \mathcal{W}(E,h\sigma)
\end{equation} where the inclusion is obtain ed by sending $ E_L\ni f\to f \in E$ by extending $f$ equal to $0$ outside of $\Lambda_L$. Similarly, the Hamiltonian $H_L$ can be extended to a self-adjoint operator parametrized by $L$ on $E$  \cite{Bratteli_Robinson_87}[Ex. 3.1.29]. }

	\paragraph{Chemical potential $\mu< 0$ fixed, variable density $\rho_{h}^{\mu,L}$:} If we fix a chemical potential $\mu<0$ for all $L>0$, we will always have $K_{\mu,L} = H_{L}-\mu I \geq -\mu I>0 $.  Moreover, for every compactly supported $f$ \begin{equation}\label{eq: quantum Gibbs con potenziale chimico fissato}
		\lim_{L\to+\infty}\omega^{L,\mu}_{h}(W^{h}(f)) = \exp{-\frac{h}{4}\langle f, (I +z_{h}e^{-\beta h H})(I- z_{h}e^{-\beta h H})^{-1}\rangle f} =: \omega^{\mu}_{h}(W^{h}(f)),
	\end{equation} where $H = -\Delta/2 $ is the unique self-adjoint extension $-\Delta/2|_{C_c^\infty(\mathbb{R}^\nu)}$. The functional in \eqref{eq: quantum Gibbs con potenziale chimico fissato} can be extended to a state on the full Weyl $C^*$-algebra $\mathcal{W}(E,h\sigma)$. It was proved by Cannon \cite{Cannon_73} that this state coincides with the thermodynamic limit of the canonical state. {Furthermore, it is an equilibrium state for the dynamics $\tau^{\mu}_{h,t}(W^{h}(f)) = W^{h}(e^{i(H- I\mu)t}f)$ in the sense that,  moving to the $\omega^{\mu}_{h}$-GNS representation, it satisfies a KMS condition for the $W^*$-dynamical system $(\pi_{\omega^{\mu}_h}(\mathcal{W}(E,h\sigma)'', \tau_{\omega^\mu_h})$ \cite[Ex. 5.3.2]{Bratteli_Robinson_97}\footnote{Clearly, $\omega^\mu_h$ is a quasi-free state, so that we can define its value on field operators $\Phi_{h}^\mu(f),\; f\in E$.}. }

{The state $\omega^\mu_h$ defines a reducible representation of the CCR, which is inequivalent to the Fock-Cook one. In particular, it is not possible to define a global number operator. However, one can still construct a local number operator for  bounded regions of $\mathbb{R}^\nu$ \begin{equation}
        N_{\Lambda_L} := \frac{1}{h}\sum_{m\geq 0}a_{\mu,h}^*(f_{L,m})a_{\mu,h}(f_{L,m}),
    \end{equation} where $(f_{L,m})_{m\geq 0}$ is any orthonormal basis for $E_L$, extended to $E$.
	Then, we obtain a meaningful notion of local density as \begin{equation}\label{eq: densità nel continuo potenziale chimico fissato}
		\lim_{L\to+\infty}\omega^\mu_h\left(\frac{N_{\Lambda_L}}{|\Lambda_L|}\right) = \lim_{L\to+\infty}\rho^{\mu,L}_{h} = \int \frac{\dd^\nu p}{(2\pi)^{\nu}}z_{h}e^{-\beta h p^{2}/2}(1 - z_{h}e^{-\beta h p^{2}/2}) =: \rho^{\mu}_h.
	\end{equation}}
    Note that, if the spatial dimension $\nu$ is greater or equal than $3$, the density in \eqref{eq: densità nel continuo potenziale chimico fissato} is bounded for $z_{h}\in [0,1]$ and it has as a finite upper bound the value \begin{equation}\label{eq: critical density}\rho^{0}_{h} = \int \frac{\dd ^{\nu}p}{(2\pi)^{\nu}}e^{-\beta h p^{2}/2}(1- e^{-\beta h p^2/2}) =: \rho_{c}(\beta h),\end{equation} which is usually referred to as \textit{critical density}. This upper bound is physically absurd, since we should be able to increase the density of the gas by adding further particles. The resolution of this pickle leads to the understanding of Bose-Einstein condensation and it is briefly analyzed in the next paragraph.
    If $\nu =1,2$ the density diverges for $\mu \rightarrow 0$. This is usually interpreted as an absence of condensation, but in modern formulations \cite{Buchholz_2022}  it has been remarked how there might be signs of condensation also for these dimensions.

\paragraph{Fixed density $\overline{\rho}$, variable chemical potential $\mu_{L}$:}
 Here we take $\nu \geq 3$. {Physically, we wish to keep the gas density $\overline{\rho}$ fixed while performing the thermodynamic limit. Then, the chemical potential has to depend on the volume $\mu\to \mu_L$, and it is uniquely specified by the condition \cite[Eq. 1.20]{Cannon_73} \begin{equation}\label{eq: condition on the chemical pot}
     \overline{\rho} = \omega^{L,\mu}_h\left(\frac{N_{\Lambda_L}}{|\Lambda_L|}\right) = \rho^{\mu,L}_h,\quad \text{for all } L>0.
 \end{equation}
 We can distinguish two cases. In the first one , $\overline{\rho}<\rho_{c}(\beta h)$. Then, it is possible to show \cite[Th. 2']{Cannon_73} that $\lim_{L\to+\infty }\mu_L = \overline{\mu}<0$ and the resulting state is given by
 \begin{equation}\label{eq: llimit state se densita sotto quella critica}
			\lim_{L\to+\infty}\omega^{L,\mu}_{h}(W^{h}(f)) = \exp{-\frac{h}{4}\langle f, (I +\overline{z}_{h}e^{-\beta h H})(I- \overline{z}_{h}e^{-\beta h H})^{-1}\rangle f} = \omega^{\overline{\mu}}_{h}(W^{h}(f)),
	\end{equation} where $\omega^{\overline{\mu}}_h$ has the same properties of the state $\omega^\mu_h$ of the previous paragraph.}

	{If instead $\overline{\rho} \geq \rho_{c}(\beta h)$, then $\lim_{L\to+\infty}\mu_L =  0$. In this case, it is possible to show \cite[Th. 2']{Cannon_73} that the finite volume states converge on Weyl elements to\begin{multline}\label{eq: limit state se densita sopra quella critica}
	\lim_{L\to+\infty}\omega^{L,\mu}_{h}(W^{h}(f))\\ = \exp{-\frac{h}{4}\bigg{[}\langle f, (I +e^{-\beta h H})(I- e^{-\beta h H})^{-1}\rangle f + 2^{\nu}(\overline{\rho} -\rho_{c}(\beta h)\Big|\int_{\mathbb{R}^{\nu}}\dd^{\nu }xf(x)\Big|^{2})\bigg{]}} =: \omega^{\overline{\rho},0}_{h}(W^{h}(f)).
	\end{multline} While for $f\in E_{L}$, the exponent in \eqref{eq: limit state se densita sopra quella critica} is always finite, the operator $(I+e^{-\beta h H})(I- e^{-\beta h H})^{-1}$ is unbounded with domain equal to $D(\Delta^{-1})$ and the integral $\int_{\mathbb{R}^{\nu}}\dd^{\nu}xf(x)$ might not be well defined for $f\in E$. Indeed, $\omega^{\overline{\rho},0}_{h}$ can be extended to a state on $\weylSH$ or on $\overline{\bigcup_{\Lambda_{L}\subset\mathbb{R}^{\nu}}\mathcal{W}(E_{L},h\sigma)}$, but not to the whole $\mathcal{W}(E,h\sigma)$. These states are labeled by a continuous parameter $\overline{\rho}\in [\rho_c(\beta h),+\infty)$. They define inequivalent GNS representations for different values of $\overline{\rho}$ and they satisfy the $W^*$-KMS condition with respect to the same temperature and represented dynamics. In other words, we have non-uniqueness of the equilibrium states, which physically represents the occurrence of a phase transition. As before, it is not possible to define a global number operator for $\omega^{\overline{\rho},0}_h$.}
    
	{The condensate term \begin{equation}
	    \exp{-\frac{h}{4}2^\nu(\overline{\rho}-\rho_{c}(\beta h))\left| \int_{\mathbb{R}^\nu}f(x)\right|^2}
	\end{equation} corresponds to a macroscopic occupation of the ground state, i.e.\begin{equation}
		\lim_{L\to+\infty} \frac{1}{|\Lambda_{L}|} e^{-\beta h(E_{0}(L)- \mu_{L})}(1-e^{-\beta h(E_{0}(L)- \mu_{L})}) = \overline{\rho} - \rho_{c}(\beta h).
	\end{equation} This is the only energy level for which fluctuations persist in the thermodynamic limit. Physically, this reflects the fact that a macroscopic number of particles occupy the ground-state energy level, while the remaining particles form together a thermal background.}

{\begin{remark}(\textbf{Classical nature of the condensate})\label{rmk: classical nature of the condensate}
    The condensate interpretation is further supported by an argument due to Araki and Woods \cite[Sec.3 \& 4]{Araki_Woods_63}. In the case of a macroscopic occupation of the ground-state, they succeed in constructing, within suitable infinite-volume representations of the CCR, annihilation and creation operators associated with particles in the ground-state sector \begin{equation}
\lim_{L\to+\infty}\frac{1}{|\Lambda_L|}a_{\omega^{\overline{\rho},0}_h}(\chi_{\Lambda_L})=: a,\quad\lim_{L\to+\infty}\frac{1}{|\Lambda_L|}a_{\omega^{\overline{\rho},0}_h}^*(\chi_{\Lambda_L})=: a^*,
    \end{equation}
    where $\chi_{\Lambda_L}$ is the characteristic function of $\Lambda_L$. A striking property of these operators is that they commute with the represented Weyl algebra. Araki and Woods interpret this commutativity as reflecting the \textit{classical character} of the condensate, which arises from the infinite population of particles in the zero-momentum mode. The classical nature of the condensate term will be exploited in Sec. \ref{sec: limiti classici degli stati} to construct classical infinite-volume states having a macroscopic occupation of the ground state sector. This will be done by studying the Araki-Woods representation in the limit of infinite local density.\end{remark}}

	\subsection{Finite volume Classical Systems}\label{sec: sistemi a volume finito}
 
In this section we find classical analogous of the finite volume states $\omega^{L,\mu}_h$ introduced in the previous subsection. These belong to a family of weak KMS states relatively to a single particle Hamiltonian $H$.  {The main results of this section are Theorems \ref{th: KMS cilindrico se e solo se KMS come misura},\ref{th: weak KMS se e solo se KMS cilindrico},  and Corollary \ref{corollary: unicità dei KMS su volume finito}. The proofs of some auxiliary results, Lemma \ref{lemma: estensione a funzioni limitate} and Prop. \ref{prop: caratterizzazioni equivalenti della condizione KMS cilindrica}, are collected in Appendix \ref{sec: appendice B}.}

Since the classical dynamics is obtained from the quantum unitary evolution we are tempted to perform a formal limit  $h\rightarrow 0^{+}$ of Eq. \ref{eq: Gibbs state} and formulate the ansatz \begin{equation}\label{eq: ansatz classico per stati KMS}
		\omega_{0}^{\mu,L}(W^{0}(f)) = \exp{ -\frac{1}{2\beta}\langle f, (H_{L}- \mu I)^{-1} f\rangle},\quad f\in E_{L}.
	\end{equation} {This prescription can be linearly extended to a state on $\Delta(E_{L},0)$. Then, by continuity of of the sesquilinear, positive map appearing at the exponent, we can extend this definition to $\mathcal{W}(E_L,0)$} \cite[Th 3.4]{Petz_90}. In particular, they are quasi-free states, therefore, we can compute their values on arbitrary products of  generators $\{\Phi_{0}(f),\; f\in E_{L}\}$.

States of the form \eqref{eq: ansatz classico per stati KMS} resemble the characteristic functionals of Gibbs measure for infinite dimensional Hilbert spaces \cite{Ammari_Sohinger_2023}. Indeed, this analogy has already been  explored in other works \cite{Falconi_2018, Falconi_Fratini_2024}. Moreover, we know that for a finite number of classical particles, the unique equilibrium state is obtained by employing Gibbs distribution, which takes the form $\dd\rho(q,p) = Z^{-1}e^{-\beta h(q,p)}\dd q\dd p$.
 We would like to establish a similar result in the present setting. 

 To proceed forward we need some additional mathematical tools. In the following, we will discuss Sobolev spaces, infinite dimensional measures and cylindrical measures; more details can be found in \cite{Ammari_Sohinger_2023,Gel'fand_Vilenkin_,Huang_Yan_2000,Skorohod_1974}.
	
	\paragraph{Abstract setting: } We consider a  real symplectic space $E$, coming from a separable, complex Hilbert space $(\mathcal{H},\langle\cdot,\cdot\rangle)$. We take a self-adjoint operator $H > 0$, with compact resolvent, such that there exists $s>0$ for which $H^{-s}$ is trace class.
    
    We endow $E$ with a \textit{Rigged Hilbert space} structure: $\Phi \subset E \subset \Phi'$, where $\Phi$ is a dense subset in $E$ with respect to the Hilbert space norm, equipped with a topological vector space structure, while $\Phi'$ is the topological dual of $\Phi$. A relevant example is given by $\mathcal{S}(\mathbb{R}^{\nu})\subset \mathbb{L}^2(\mathbb{R}^\nu) \subset \mathcal{S}(\mathbb{R}^{\nu})'$. For the current discussion, we consider the \textit{Sobolev Hilbert spaces} $\mathcal{H}^{s}\subset E \subset \mathcal{H}^{-s}$.  We recall the definition of these spaces

 \begin{definition} (\textbf{Sobolev Hilbert spaces})\label{def: Sobolev Hilbert spaces}
		The Sobolev Hilbert space $\mathcal{H}^s$  is defined as the vector subspace $D(H^{s/2})\subset E$, with scalar product $\langle f, g\rangle_{s}:= \Re{\langle H^{s/2}f, H^{s/2} g\rangle},\; f,g\in D(H^{s/2}) $ and topology induced by the norm $\norm{f}_{s}:= \norm{H^{s/2}f}$. The  space $\mathcal{H}^{-s}$ is defined as the completion of $E$ with respect to the norm $\norm{f}_{-s}:= \norm{H^{-s}f}$, induced by the scalar product $\langle f, g\rangle_{-s}:= \Re{\langle H^{-s/2}f, H^{-s/2}g\rangle}$.
	\end{definition} 
 
        We summarize the basic properties of these spaces: (The proofs of the following assertions can be found in \cite[Lem. 4.5]{Huang_Yan_2000})
	\begin{enumerate}
		\item[(a)] $\mathcal{H}^{s}$ is a separable Hilbert space with respect to $\langle\cdot,\cdot\rangle_{s}$;
		\item[(b)] $H^{s}$ can be extended to an isometry from $\mathcal{H}^{s}$ onto $\mathcal{H}^{-s}$, equivalently $H^{-s}$ can be extended to an isometry from $\mathcal{H}^{-s}$ onto $\mathcal{H}^{s}$;
		\item[(c)] $H^{-s}$ is a positive, symmetric and trace class operator in $\mathcal{H}^{-s}$ with $\Tr H^{-s}= \Tr_{-s}{H^{-s}}$ ;
		\item[(d)] $\mathcal{H}^{s}$ e $\mathcal{H}^{-s}$ are mutually adjoint with canonical duality given by \begin{equation}
			\langle f, u\rangle := \langle H^{s}f, u\rangle_{-s},\quad f\in \mathcal{H}^{s},\; u\in \mathcal{H}^{-s}.
		\end{equation}
        
	\end{enumerate}

	We introduce  classes of functions which can be manipulated easily as in finite dimensional spaces. These are called \textit{smooth cylindrical functions}, where the word \textit{smooth} suggests the possibility of being differentiated, while \textit{cylindrical} says that they only depends on finite dimensional subspaces of $E$.
	\begin{definition}(\textbf{Smooth Cylindrical functions})\label{def: smooth cylindrical functions}
		We fix an orthonormal basis of $(E, \Re{\langle\cdot,\cdot\rangle})$ composed of $H$'s normalized eigenvectors $\{e_{k}, ie_{k}\}_{k\in\mathbb{N}} =: \{e_{k}, f_{k}\}_{k\in\mathbb{N}}\subset \mathcal{H}^{s}$. These form an orthogonal basis also for $\mathcal{H}^{s}$ and $\mathcal{H}^{-s}$. A smooth cylindrical function is a functional $F\colon \mathcal{H}^{-s}\rightarrow \mathbb{R}$ such that there exists $n\in\mathbb{N}_{0}$. a map $\varphi\in C^{\infty}_{c}(\mathbb{R}^{2n}),\; \mathcal{S}(\mathbb{R}^{2n}),\; \text{or } C^{\infty}_{b}(\mathbb{R}^{2n})$ and  \begin{equation}
		F(u) = \varphi(\pi_{n}(u)) =\varphi(\langle e_{1}, u\rangle,\ldots, \langle e_{n}, u\rangle,
		\langle f_{1}, u\rangle,\ldots,\langle f_{n}, u\rangle),\quad  u \in \mathcal{H}^{-s},
	\end{equation} where $\pi_{n}: \mathcal{H}^{-s}\ni u\rightarrow (\langle e_{j},u\rangle\dots \langle f_{n},u\rangle)\in\mathbb{R}^{2n} $.
We denote the spaces of smooth cylindrical functions respectively $C^{\infty}_{c,cyl}(\mathcal{H}^{-s})$, $\mathcal{S}_{cyl}(\mathcal{H}^{-s})$, $C^{\infty}_{b,cyl}(\mathcal{H}^{-s})$.
	\end{definition} Definition \ref{def: smooth cylindrical functions}  can be generalized for the rigged Hilbert space setting $\Phi \subset E \subset \Phi' $ \cite{Ammari_Sohinger_2023}. In view of definition \ref{def: smooth cylindrical functions} we call $\{L_{n}\}_{n\in\mathbb{N}}$, the finite dimensional subspaces spanned by $\{e_{k},f_{k}\}_{k\in\{1,\dots, n\}}$.
	
These functionals can be derived, and in particular, the gradient of a smooth cylindrical function is a well defined object
\begin{equation}\label{eq: gradient of smooth cf}
	\nabla F(u) = \sum_{j=1}^n(\partial^{1}_{j}\varphi)(\pi_{n}(u))e_{j} + (\partial^{2}_{j}\varphi)(\pi_{n}(u))f_{j}\in \mathcal{H}^{s},\end{equation}where $\partial^{1}_{j}\; (\partial^{2}_{j})$ is the directional derivative along the $j\; (j+n)$-direction.

 With \eqref{eq: gradient of smooth cf} we can also define a Poisson bracket for smooth cylindrical functions by simply taking \begin{equation}\label{def: poisson cylindrical}
	\{F, G\}(u):= \sigma(\nabla F(u), \nabla G (u)).
\end{equation}

We recall that smooth cylindrical functions can be \textit{integrated} with respect to cylindrical measures \cite{Falconi_2018,Skorohod_1974}. We call $\mathfrak{B}(\mathcal{H}^{-s})$ the Borel sigma algebra of $\mathcal{H}^{-s}$ and $\mathfrak{P}(\mathcal{H}^{-s}),\; \mathfrak{P}_{cyl}(E)$ the space of \textit{Borel probability measures} and \textit{normalized cylindrical measures} respectively. 

Within this setting, we can adapt \cite[Def. 2.3]{Ammari_Sohinger_2023} (see also definition  \ref{def: KMS per misure})  to construct a notion of KMS states for cylindrical measures.
	\begin{definition} \label{def: KMS cilindrici} (\textbf{Cylindrical KMS state})
 Let $X \colon \mathcal{H}^{-s} \rightarrow \mathcal{H}^{-s}$ be the Borel vector field $X(u):= iHu$ and take $\beta >0$. We say that $\mu_{*}\in \mathfrak{P}_{cyl}(E)$ is a $(X,\beta)$-\textit{cylinder KMS state} if , the function $\langle \varphi , X(\cdot)\rangle$ is $\mu_{L_{n}}$ integrable for all $\varphi\in  L_{n}$, for all $n\in\mathbb{N}$, and if for arbitrary $F,\; G\in C^{\infty}_{c,cyl}(\mathcal{H}^{-s})$  the equality \begin{equation}\label{eq: cylindrical KMS}
\int \{F, G\}(u)\dd\mu_{*}(u) = \beta \int \langle \nabla F(u), X(u)\rangle G(u)\dd \mu_{*}(u),
\end{equation} is satisfied
with the Poisson bracket defined in \eqref{def: poisson cylindrical}.
\end{definition}
\begin{remark}
A few remarks on Equation \eqref{eq: cylindrical KMS} are in order:
\begin{enumerate}
\item[(a)] The bracket in \ref{def: poisson cylindrical} leaves invariant the space of smooth cylindrical functions.
\item[(b)] The vector field $X(u) = iHu$ has to be interpreted in a distributional sense, i.e. for all $f\in \mathcal{H}^{-s}$, $\langle f, X(u)\rangle = \langle -iHf, u\rangle$. In particular, $\mathcal{H}^{-s}\ni u \rightarrow \langle\nabla F(u), X(u)\rangle\in\mathbb{R}$ is still a smooth cylindrical function.
\end{enumerate} 
\end{remark} 
Cylindrical measures are completely specified by finite dimensional integrals over the subspaces $\{L_{n}\}_{n\in\mathbb{N}}$; thus, the cylindrical KMS condition in \ref{def: KMS cilindrici} shares many properties of the one for Borel measures given in \cite[Def. 2.3]{Ammari_Sohinger_2023}: \begin{definition}
		\label{def: KMS per misure}(\textbf{KMS state}) Let $X: \Phi' \rightarrow \Phi'$ be a borel vector field  and $\beta >0$. We say that $\mu\in \mathfrak{P}(\Phi')$ is a $(\beta,X)$-\textit{KMS state} if and only if for all $F,\; G\in C^{\infty}_{c,cyl}(\Phi)$ the function $\langle \varphi , X(\cdot)\rangle$ is $\mu$-integrable for all $\varphi\in \Phi$ and if the equality \begin{equation}\label{eq: measure KMS condition}
			\int_{\Phi'} \{F, G\}(u)\dd\mu(u) = \beta \int_{\Phi'} \langle \nabla F(u), X(u)\rangle G(u)\dd \mu(u),
		\end{equation} is satisfied.
	\end{definition}
 First of all, one has the following lemma
 \begin{lemma} \label{lemma: estensione a funzioni limitate}If $\mu_{*}\in\mathfrak{P}_{cyl}(E)$ is a $(X,\beta)$-cylindrical KMS, then, identity \eqref{eq: cylindrical KMS} is true for all $F,G\in C^{\infty}_{b,cyl}(\mathcal{H}^{-s})$.
	
\end{lemma}

	Next, there is an equivalent characterization of cylindrical KMS-states similar to the one stated in \cite[Th. 2.7]{Ammari_Sohinger_2023} for Borel probability measures
\begin{proposition}\label{prop: caratterizzazioni equivalenti della condizione KMS cilindrica}
	Let $\mu_{*}\in \mathfrak{P}_{cyl}(E)$, $X(u) = iHu$ and $\beta>0$. Then the following properties are equivalent: \begin{enumerate}
		\item[(i)] $\mu_{*}$ is a $(\beta,X)$-cylindrical KMS state.
		\item[(ii)] For all $\varphi_{1},\varphi_{2}\in \mathbb{R}$-span$\{(e_{n},f_{n})_{n\in\mathbb{N}}\}$, the function $\langle\varphi_{1}, X(\cdot)\rangle$ is $\mu_{L_{n}}$ integrable, for $n$ sufficiently large, and we have the identity \begin{equation}\label{eq: second KMS cylinders}
			\Re{\langle i\varphi_{1},\varphi_{2}\rangle}\int e^{i\langle \varphi_{2}, u\rangle}\dd\mu_{*} +i\beta\int \langle \varphi_{1}, X(u)\rangle e^{i\langle \varphi_{2},u\rangle}\dd\mu_{*}=0
		\end{equation}
        
	\end{enumerate}
\end{proposition}

With the equivalent formulation in \eqref{eq: second KMS cylinders}, we are able to show the uniqueness of cylindrical KMS-states when the Borel field is induced by a Hamiltonians $H$ with the previously mentioned properties. Cylindrical measures can be characterized with their characteristic functionals \begin{equation}
    \hat{\mu}_{*}(f) = \int e^{i\langle f,u\rangle}\dd\mu_{*}(u),
\end{equation} and we will show that the only possibility for KMS ones is given by\begin{equation}\label{eq: fourier di gaussiana}
	\theta(f) = \exp{-\frac{1}{2\beta}\langle f, H^{-1} f\rangle},\quad f\in E.
\end{equation} The cylindrical measure induced by $\theta$ can be \textit{radonified} on $\mathcal{H}^{-s}$ \cite{Ammari_Sohinger_2023,Bogachev_1998,Huang_Yan_2000, Skorohod_1974} to obtain a Borel probability measure $\mu\in \mathfrak{P}(\mathcal{H}^{-s})$ with Fourier transform $\hat{\mu}(f) = \theta(f)$, for every $f\in \mathcal{H}^{s}$.
By  \cite[Th. 4.8]{Ammari_Sohinger_2023}, the Borel probability measure associated with the latter Fourier transform is the unique  one satisfying the KMS condition for measures \ref{def: KMS per misure}.

	\begin{theorem}\label{th: KMS cilindrico se e solo se KMS come misura}
	Let $\mu_{*}\in\mathfrak{P}_{cyl}(E),$ $X(u) = iHu$ and $\beta >0$. Then, $\mu_{*}$ is a $(X,\beta)$-cylindrical KMS state iff its characteristic functional takes the form \eqref{eq: fourier di gaussiana}. In particular, there is only one such state.
	\end{theorem}
	\begin{proof}
		$(\Rightarrow)$ We assume that $\mu_{*}\in \mathfrak{P}_{cyl}(E)$ satisfies the cylindrical KMS condition. Then, for every $F,\; G\in C_{c,cyl}^{\infty}(\mathcal{H}^{-s})$ such that $F(u) = \varphi(\pi_{n}(u)),\; G(u) = \psi(\pi_{n}(u))$, we have \begin{equation}
		\int_{L_{n}}\{F,G\}(u)\dd\mu_{L_{n}}(u) = \beta \int_{L_{n}} \Re{\langle \nabla F(u),  X_{0}(u) \rangle}\dd\mu_{L_{n}}(u).
	\end{equation}
 This implies that for all finite dimensional subspaces $L_{n}$, the measures $\mu_{L_{n}}\in \mathfrak{P}(L_{n})$ satisfy the $(\beta,X)$-KMS condition. By \cite[Th. 4.2]{Ammari_Sohinger_2023}, we conclude that \begin{equation}\label{eq: weak distribution}
		\dd\mu_{L_{n}}(f) = \frac{e^{-\frac{\beta}{2}\langle f, H f\rangle }\dd\mathcal{L}_{2n}(f)}{\int_{L_{n}}e^{-\frac{\beta}{2}\langle f, H f\rangle }\dd\mathcal{L}_{2n}(f)},
	\end{equation} Where $\mathcal{L}_{2n}$ is the Lebesgue measure of $L_{n}$.  A cylindrical measure on an Hilbert space is fully specified by its weak distribution $\{\mu_{L_{n}}\}_{n\in\mathbb{N}}$ \cite{Skorohod_1974} and by a simple calculation {\begin{equation}
	\hat{\mu}_*(\sum_{k=1}^{n}\alpha_{k}e_{k} +\mu_{k}f_{k}) = \int_{\mathbb{R}^{2n}} e^{i\sum_{k=1}^{n}\alpha_{k}q_{k} + \mu_{k}p_{k}}e^{-\frac{\beta}{2}\sum_{k=1}^{n}\lambda_{k}(q_{k}^{2}+p_{k}^{2})}\dd\mathcal{L}_{2n}(q,p)\prod_{k=1}^{n}\frac{\lambda_{k}\beta}{\pi}  
 =\theta(\sum_{k=1}^{n}\alpha_{k}e_{k} +\mu_{k}f_{k}).
	\end{equation} }This implies that for arbitrary $n\in\mathbb{N}$, the functional $\theta$ restricted to $L_{n}$ induces the Borel probability measure $\mu_{L_{n}}$ and so, it is equal to the characteristic functional of $\mu_{*}$.
 
	$(\Leftarrow)$ if $\mu_{*}$ has as characteristic functional \eqref{eq: fourier di gaussiana}, then its weak distribution is given by \eqref{eq: weak distribution}, which satisfies the KMS condition for measures.
	\end{proof} 
    
    With the previous hypothesis on $H$, the measure which extends the unique cylindrical $(X,\beta)$-KMS $\mu_{*}$ is usually referred to as Gibbs measure \cite{Ammari_Sohinger_2023}. The resemblance with the usual Gibbs distribution $\dd\rho(q,p)$ is clear from equation \eqref{eq: weak distribution} where we identify the classical Hamiltonian $h(q,p) \leftrightarrow h(\alpha, \mu):= \langle \sum_{j=1}^{n}\alpha_{k}e_{k} +\mu_{k}f_{k}, H\sum_{k=1}^{n}\alpha_{k}e_{k} +\mu_{k}f_{k} ) \rangle$.
	
	{\paragraph{States on Weyl $C^*$-algebras:}  The values of the  characteristic functional $\theta(\cdot)$  in \eqref{eq: fourier di gaussiana} coincide with the expectation values of the quasi-free state \begin{equation}\label{eq: quasi free first version}\omega_{0}(W^0(f)):= \exp{-\frac{1}{2\beta}\langle f, H^{-1} f\rangle},\quad f\in E.\end{equation} Moreover, we can immediately verify that $\omega_0$ satisfies a weak KMS condition 
    \begin{proposition}\label{prop: verifica KMS weak}
   Let $\delta_0$ be the pointwise closed, linear, continuous weak derivation defined by \begin{equation}\delta_0(W^0(f)) := i\Phi_0(iHf)W^0(f),\end{equation} with $H>0$ self-adjoint operator, such that there exists an $s>0$ for which\begin{equation}\Tr{H^{-s}}<+\infty.\end{equation} Then the classical quasi-free state $\omega_{0}$ \eqref{eq: quasi free first version} is a $(\delta_{0},\beta)$-weak KMS state.
\end{proposition}
\begin{proof}
    The functional $\omega_{0}$ is quasi-free and in particular of class $C^2$, so, we can test whether it is a $(\delta_{0},\beta)$-weak KMS state. We need to show that the identity \begin{equation}	\label{eq: KMS wanna be}
		\sigma(g,f)\omega_{0}(W^{0}(f+g)) = i\beta\omega_{0}(\Phi_{0}(iHf)W^{0}(f+g)),
	\end{equation} is satisfied for $f\in D(H) = E(\delta_{0})$; then the conclusion follows by linearity. We compute for arbitrary $h,k\in E$
 \begin{align}\label{eq: state on fields}
		\omega_{0}(\Phi_{0}(h)W^{0}(k)) &= -i\frac{\dd}{\dd t}\omega_{0}(W^{0}(k+th))|_{t = 0}\nonumber\\
				&= \frac{i}{\beta}\Re{\langle k, H^{-1} h\rangle}\omega_{0}(W^{0}(k)).
	\end{align} By substituting $h\rightarrow iHf$ and $k\rightarrow f+g$, for $f,g\in D(H)$, we obtain \begin{align}\label{eq: verifica della condizione KMS su volume finito}
		i\beta\omega_{0}(\Phi_{0}(iHf)W^{0}(f+g)) &= -\Re{\langle f+g, if\rangle}\omega_{0}(W^{0}(f+g)) \nonumber\\ & = \sigma(g,f)\omega_{0}(W^{0}(f+g)).
	\end{align}
\end{proof}}
    
    We can now formulate the main theorem of this section \begin{theorem}\label{th: weak KMS se e solo se KMS cilindrico}
	Let us consider the dynamics induced on the Weyl $C^*$-algebra $\mathcal{W}(E,0)$ by the derivation $\delta_0$ as in Prop. \ref{prop: verifica KMS weak}. If  $\omega_{0}\in \mathcal{S}(\mathcal{W}(E,0))$ is a $(\delta_{0},\beta)$-weak KMS state, then there exists a Borel probability measure $\mu_{0}\in\mathfrak{P}(\mathcal{H}^{-s})$  satisfying the KMS condition \ref{def: KMS per misure}  with Fourier transform given by \begin{equation}
		\hat{\mu}_0(f)= \omega_{0}(W^{0}(f)),\quad f\in \mathcal{H}^{s}.
	\end{equation} Viceversa, if $\mu_{0}$ is a $(X,\beta)$-KMS state, the algebraic state obtained by extending \begin{equation}
		\omega_{0}(W^{0}(f)): = \hat{\mu}_0(f),\quad f\in \mathcal{H}^{s}
	\end{equation} is a $(\delta_{0},\beta)$-weak KMS state.
	\end{theorem}
	\begin{proof}
	$(\Rightarrow)$ Let $\omega_{0}$ be a $(\delta_{0},\beta)$- weak KMS state and define $\theta(f):= \omega_{0}(W^{0}(f)),\; f\in E$. Then, since $\omega_{0}$ is a $C^{2}$ state on $\mathcal{W}(E,0)$, it is regular and we know that $\theta: E\rightarrow \mathbb{C}$ is a positive definite function with $\theta(0) = 1$, which in addition is continuous on any finite dimensional subspace $E_{n}\subset E$. Bochner theorem (see \cite[Th 1.8 \& Th. 1.9]{Falconi_2018} and references therein) ensures that these are sufficient conditions for $\theta$ to completely specify a cylindrical measure $\mu_{0,*}$ on $E$. The same cylindrical measure  is also characterized by its weak distribution $\{\mu_{0,L_{n}}\in \mathfrak{P}(L_{n})\}$. The measures $\mu_{0,L_{n}}$ have Fourier transforms given by\begin{equation}\label{eq: finite KMS measures}
		\hat{\mu}_{0,L_{n}}(f) = \omega_{0}(W^{0}(f)),\quad f\in L_{n}.
	\end{equation}
	In addition, since the state $\omega_{0}$ is $C^{2}$, it follows that for every fixed $n\in\mathbb{N}$ the characteristic functional $\mathbb{R}^{2n}\ni (\lambda,\mu) \rightarrow \theta(\sum_{i=1}^{n}\lambda_{i}e_{i} + \mu_{i}f_{i})\in \mathbb{C}$ is twice differentiable.
	This implies \cite[Th. 2.3.2]{Lukacs_70} that the measures $\{\mu_{L_{n}}\}_{n\in\mathbb{N}}$ have finite first moments, i.e. $\langle\varphi, \cdot\rangle$ is $\mu_{L_{n}}$-integrable for every $\varphi\in L_{n}$ and for every $n\in\mathbb{N}$.
	Now, we show that $\mu_{0,*}$ is a cylindrical KMS state. Let us take $\varphi_{1},\varphi_{2}\in L_n$. Then, we can compute
 {\begin{align}
 \int \langle \varphi_{1}, X(u)\rangle e^{i\langle \varphi_{2}, u\rangle}\dd \mu_{0,*}(u)
  &= -i\frac{\dd}{\dd t} \omega_{0}(W^{0}(\varphi_{2}-itH\varphi_{1}))|_{t=0}\nonumber\\ &= -\omega_{0}(\Phi_{0}(iH\varphi_{1})W^{0}(\varphi_{2}))\nonumber\\
		&=	\frac{i}{\beta}\sigma(\varphi_{1},\varphi_{2})\omega_{0}(W^{0}(\varphi_{2})) = \frac{i}{\beta}\langle i\varphi_{1},\varphi_{2}\rangle_{\mathbb{R}}\int_{L_{n}} e^{i\langle \varphi_{2},u\rangle} \dd\mu_{0,*}(u).
	\end{align}}
	
	According to \ref{prop: caratterizzazioni equivalenti della condizione KMS cilindrica}, the previous computation proves that $\mu_{0,*}$ is a $(X,\beta)$-cylindrical KMS. However, we know from theorem \ref{th: KMS cilindrico se e solo se KMS come misura} that the only possible KMS cylindrical state is the Gaussian one, with characteristic functional\begin{equation}
		\theta(f) = e^{-\frac{1}{2\beta}\langle f, H^{-1} f\rangle},\quad f\in E.
	\end{equation} This defines uniquely a Borel probability measure $\mu_{0}\in \mathfrak{P}(\mathcal{H}^{-s})$, which is the unique \cite[Th. 4.8]{Ammari_Sohinger_2023} $(X,\beta)$-KMS state and has Fourier transform given by \begin{equation}\label{eq: gaussian characteristic}
		\hat{\mu}_0(f) = e^{-\frac{1}{2\beta}\langle f, H^{-1}f\rangle},\quad f\in \mathcal{H}^{s}.
	\end{equation}
	
	$(\Leftarrow)$ Conversely, assume that $\mu_{0}\in\mathfrak{P}(\mathcal{H}^{-s})$ is a $(\beta,X_{0})$-KMS state.
	Then, by uniqueness, its characteristic functional is given by \eqref{eq: gaussian characteristic}.
	Since $H^{-1}$ is bounded, $\hat{\mu}_0$ is norm-continuous and defined on a dense subspace of $E$, so, it can be extended uniquely in a continuous way to all of $E$. By defining\begin{equation}\label{eq: KMS algebraic from masure}
		\omega_{0}(W^{0}(f)) := e^{-\frac{1}{2\beta}\langle f, H^{-1} f\rangle},\quad f\in E,
	\end{equation} and extending the action of $\omega_{0}$ by linearity and continuity, we obtain a quasi-free state on $\mathcal{W}(E,0)$. We have already verified that this state satisfies the $(\delta_{0},\beta)$-weak KMS condition with equation \eqref{eq: verifica della condizione KMS su volume finito}.
	\end{proof} 
	
	As an important corollary we obtain
	\begin{corollary}(\textbf{Uniqueness of finite volume weak KMS states})\label{corollary: unicità dei KMS su volume finito}
			Let $\delta_{0}$ be the pointwise closed, linear, continuous weak derivation defined by $\delta_{0}(W^{0}(f)):= i\Phi_{0}(iHf)W^{0}(f)$ with $H$ strictly positive, self-adjoint, with compact resolvent such that we can find an $s\geq 0$ for which \begin{equation}
			\Tr{H^{-s}}< +\infty.
		\end{equation} Then, there exists a unique $(\delta_{0},\beta)$-weak KMS state defined by \begin{equation}
			\omega_{0}(W^{0}(f)) = e^{-\frac{1}{2\beta}\langle f, H^{-1} f\rangle}, \quad f\in E.
		\end{equation}
		In particular the state $\omega^{L,\mu}_{0}$ defined in \eqref{eq: ansatz classico per stati KMS} on $\mathcal{W}(E_{L},0)$ is the unique $(\delta_{0}^{\mu,L},\beta)$-weak KMS state.
	\end{corollary}
	
	Thanks to the identification between Fourier transforms of Gibbs measures and weak KMS-states, we have a nice interpretation of $\omega_{0}(W^{0}(f))$ as the mean value of the (extended) function $W^{0}(f) \colon \mathcal{H}^{s}\rightarrow \mathbb{C},\; f\in \mathcal{H}^{s}$ in the Gibbs state\begin{equation}
	\omega_{0}(W^{0}(f)) = \int_{\mathcal{H}^{-s}}W^{0}(f)[u] \dd \mu_{0} (u). 
	\end{equation}
	This motivates the definition as \textit{Gibbs states}, for the finite volume weak KMS states $\omega^{L,\mu}_{0}$.

	\subsection{Classical Limits}\label{sec: limiti classici degli stati} 
    {In this section, we construct classical representations of the CCR by means of  infinite-density equilibrium states. Our approach relies on the semi-classical formalism developed in Sec.\ref{sec: dequantization}. Specifically, we consider the Araki-Woods and Cannon equilibrium states, $\omega^{L,\mu}_h,\; \omega^{\mu}_h$ and $\omega^{\overline{\rho},0}_h$, introduced in \ref{sec: condensazione quantistica} and we compose them with the quantization map $Q_h$. The semi-classical parameter is connected to the critical density \eqref{eq: critical density} of the state as \begin{equation}\label{eq: connection to the critical density}
		h/h_{0} = (\rho_{c}(\beta h_{0})/\rho_{c}(\beta h))^{2/3}, 
	\end{equation} where $h_0$ is a given reference scale. Thus, sending $h\to0^+$ corresponds to increasing the critical density $\rho_c(\beta h) \to +\infty$. The limiting states all satisfy the weak KMS condition \ref{def: classical KMS condition} and yield  commutative, infinite-volume representations of the CCR without requiring the finiteness of the local density or the existence of a local number operator $N_{\Lambda_L}$. A remarkable property is that with this procedure we are also able to construct a classical state $\omega^{\alpha,0}_0$ describing a macroscopic occupation of the ground-state sector of the theory.}

	\paragraph{No macroscopic occupation of the ground state: } We start with the finite volume Gibbs states, $(\omega^{L,\mu}_{h})_{h\in(0,+\infty)}$ \eqref{eq: Gibbs state} for $\mu < E_0(L)$.
    Thanks to the explicit expressions of states and $Q_h$, we can compute the limit for $h\to 0^+$ of the state $\omega_h^{L,\mu}\circ Q_h $. Given $f\in E_L$ we have
	\begin{equation}
		\omega^{L,\mu}_{h}\circ Q_{h}(W^{0}(f)) = e^{-\frac{h}{4}\norm{f}^{2}}\exp{-\frac{1}{4}\sum_{\underline{n}\in \mathbb{N}^{\nu}} \langle f, P_{\underline{n}}(L) f\rangle h \frac{1+z_{h}e^{-\beta h E_{\underline{n}}(L)}}{1-z_{h}e^{-\beta h E_{\underline{n}}(L)} }}.
	\end{equation} Now, for small $h$ and $\mu< E_{0}(L)$, we have  \begin{equation}\label{eq: stima nel discreto}
	\langle f, P_{\underline{n}}(L) f\rangle h \frac{1+z_{h}e^{-\beta h E_{\underline{n}}(L)}}{1-z_{h}e^{-\beta h E_{\underline{n}}(L)} } \leq \langle f, P_{\underline{n}}(L) f\rangle \frac{2}{E_{\underline{n} }(L) -\mu },
	\end{equation}and the term on the right is summable. So, by dominated convergence\begin{equation}
	\lim_{h\rightarrow 0^{+}} \omega^{L,\mu}_{h}\circ Q_{h}(W^{0}(f)) = \exp{-\frac{1}{2\beta}\langle f, (H_{L}- \mu I)^{-1} f\rangle} = \omega^{L,\mu}_{0}(W^{0}(f)),
	\end{equation} where $\omega_0^{L,\mu}$ is the state introduced in Eq. \eqref{eq: ansatz classico per stati KMS}. By continuity and linearity arguments, we deduce that the state $\omega^{L,\mu}_h\circ Q_h$ converges in the weak$^*$-topology, i.e. \begin{equation}
	\omega^{L,\mu}_{h}\circ Q_{h}(c)\rightarrow \omega^{L,\mu}_{0}(c),\quad \text{for all } c\in \mathcal{W}(E_{L},0).
	\end{equation} 
    By Corollary \ref{corollary: unicità dei KMS su volume finito}, we know that $\omega_0^{L,\mu}$ is the unique state satisfying the weak KMS condition at temperature $\beta$ for the dynamics implemented by $H_L-\mu I$.
    
    The argument remains unaltered for the infinite volume states $\{\omega^{\mu}_{h}\}_{h\in (0,+\infty)}$ \eqref{eq: quantum Gibbs con potenziale chimico fissato} with $\mu<0$, the only difference being that the summation on eigenvalues becomes an integral
	\begin{equation}
		\omega^{\mu}_{h}\circ Q_{h}(W^{0}(f))  = e^{-\frac{h}{4}\norm{f}^{2}}\exp{-\frac{1}{4}\int_{\mathbb{R}^{\nu}}\frac{\dd^{\nu}p}{(2\pi)^{\nu}}h\frac{1+ z_{h}e^{-\beta h p^{2}/2}}{1-z_{h}e^{-\beta h p^{2}/2}}|\hat{f}(p)|^{2}},\quad f\in E.
	\end{equation}
  With the continuum analogous of estimate \eqref{eq: stima nel discreto}
	\begin{equation}
	h\frac{1+ z_{h}e^{-\beta h p^{2}/2}}{1-z_{h}e^{-\beta h p^{2}/2}}|\hat{f}(p)|^{2}	\leq |\hat{f}(p)|^{2}	\frac{1}{p^{2}/2 -\mu},
	\end{equation} we can compute by dominated convergence the limit \begin{equation}
	\lim_{h\rightarrow 0^{+}}\omega^{\mu}_{h}\circ Q_{h}(c) =\exp{-\frac{1}{2\beta}\langle f, (H-\mu I)^{-1} f\rangle}= :\omega^{\mu}_{0}(c),\quad \forall\; c\in \mathcal{W}(E,0).
	\end{equation} The same algebraic computation of Prop. \ref{prop: verifica KMS weak} shows that $\omega^\mu_0$ satisfies the weak KMS condition at inverse temperature $\beta$ for \begin{equation}
	    \delta_0^\mu(W^0(f)) = i\Phi_0(i(H-\mu I)f)W^0(f),\quad f\in D(H).
	\end{equation}

    \paragraph{Macroscopic occupation of the ground-state:} In what follows we will assume $\nu\geq3$ and work directly with the $C^*$-algebras $(\mathcal{W}(\mathcal{S}(\mathbb{R}^{\nu}), h\sigma))_{h\in[0,+\infty)}$.
    The quantization map $Q_{h}$ on $\mathcal{W}(\mathcal{S}(\mathbb{R}^{\nu}),0)$ is just the restriction of the one for $\mathcal{W}(E,0)$. 
	
	 {We are interested in the weak$^*$-limit points of the states  $\omega_{h}^{\overline{\rho},0}\circ Q_{h}$, where $\omega^{\overline{\rho},0}_{h}$ is the condensate state defined in Eq. \eqref{eq: limit state se densita sopra quella critica}. These are KMS states for the dynamics \begin{equation}
	     \tau^0_{h,t}(W^h(f)) = W^h(e^{itH}f),
	 \end{equation} thus, the relevant classical dynamics is given by the weak derivations \begin{equation}
	     \delta^0_0(W^0(f)):= i\Phi_0(iH f)W^0(f),\quad  f\in \mathcal{S}(\mathbb{R}^\nu),
	 \end{equation} which is linear, continuous and pointwise closed.}
     
     As already noted in Eq. \eqref{eq: connection to the critical density} \begin{equation}
		\rho_{c}(\beta h) = \frac{1}{h^{\nu/2}}\int_{\mathbb{R}^{\nu}}\frac{\dd^{\nu }p }{(2\pi)^{\nu}} e^{-\beta p^{2}/2}(1- e^{-\beta  p^{2}/2}) \to +\infty,\quad \text{as } h\to 0^+.
	\end{equation}
  Hence, there exists an $h_0$ such that for $h< h_0$,  $\overline{\rho}<\rho_c(\beta h)$ and we cannot perform a semi-classical limit of  $\omega_{h}^{\overline{\rho},0}\circ Q_{h}$ simply because we cannot define a condensate state $\omega_{h}^{\overline{\rho},0}$ for $h<h_0$. {Physically, if we increase the critical density of the gas, the condensate fraction decreases down to $0$. To obtain a macroscopic occupation of the ground state, we modify the net by allowing the density $\overline{\rho}$ to depend on $h$ i.e. $\overline{\rho}\rightarrow \overline{\rho}(h)$ in such a way that \begin{equation}\overline{\rho}(h)> \rho_{c}(\beta h),\quad \text{for all }h\in(0,+\infty). \end{equation} The asymptotic behavior of the construction is then governed by the limit (which we assume exists, possibly after passing to a subsequence) \begin{equation}\label{eq: insospettabile densità}
	\lim_{h\rightarrow 0}h(\overline{\rho}(h)-\rho_{c}(\beta h)) =: \alpha \in [0,+\infty].
	\end{equation}
    \paragraph{Case $\alpha \in [0,+\infty)$:} 
    In this limit, the condensate term converges to\begin{equation}
	\lim_{h\to 0^+}\exp{-\frac{h}{4}2^{\nu+1}(\overline{\rho}(h) - \rho_{c}(\beta h))|\int_{\mathbb{R}^{\nu}}\dd^{\nu} xf(x)|^{2}} = \exp{-\frac{1}{2}2^{\nu}\alpha|\int_{\mathbb{R}^{\nu}}\dd^{\nu}x f(x)|^{2}}.
	\end{equation} This exponential term corresponds to a quasi-free state, where the quadratic form at the exponent is given by projection $\left| 1 \rangle \langle 1 \right|$ onto the infinite volume ground state\footnote{{By $|1\rangle$ we mean the effective wave-function given by the constant $1$. Note that $\left| 1 \rangle \langle 1 \right|$ is well defined as a quadratic form on $\mathcal{S}(\mathbb{R}^\nu)$, but the vector $|1\rangle$ is not in $\mathbb{L}^2(\mathbb{R}^\nu)$ and the form is not closable.}}, i.e. : 
    \begin{equation}
	    \exp{-\frac{1}{2}2^{\nu}\alpha|\int_{\mathbb{R}^{\nu}}\dd^{\nu}x f(x)|^{2}} = \exp{-\frac{1}{2}2^{\nu}\alpha\langle f,\left| 1 \rangle \langle 1 \right| f\rangle }
        \end{equation}}
For the thermal background term, since $\nu\geq 3$ and $\mathcal{F}(\mathcal{S}(\mathbb{R}^{\nu}))= \mathcal{S}(\mathbb{R}^{\nu})$, the integral\begin{equation}
	h\langle f, (I +e^{-\beta h H})(I - e^{-\beta h H})^{-1} f\rangle =h \int_{\mathbb{R}^{\nu}}\frac{\dd^{\nu} p}{(2\pi)^{\nu}}|\hat{f}(p)|^{2}\frac{1+e^{-\beta h p^{2}/2}}{1- e^{-\beta h p^{2}/2}},\quad f\in\mathcal{S}(\mathbb{R}^{\nu}),
	\end{equation} is finite. Moreover, the integrand can be bounded from above, as \begin{equation}\label{eq: dominated convergence per transizioni di fase}
	h|\hat{f}(p)|^{2}\frac{1+e^{-\beta h p^{2}/2}}{1-e^{-\beta h p^{2}/2}}\leq \frac{M}{(1+p^{2})^{n}(1-e^{-\beta p^{2}/2})},
	\end{equation} for some constant $M>0$ and $n\in\mathbb{N}$ arbitrarily large. The function to the right-hand side of \eqref{eq: dominated convergence per transizioni di fase} is integrable and we can exploit dominated convergence to obtain \begin{equation}
	\lim_{h\rightarrow 0^{+}}h\langle f, (I +e^{-\beta h H})(I - e^{-\beta h H})^{-1} f\rangle  = 2\langle f, (\beta H)^{-1} f\rangle = \frac{2}{\beta}\int_{\mathbb{R}^{\nu}}\frac{\dd^{\nu} p}{(2\pi)^{\nu}}\frac{|\hat{f}(p)|^{2}}{p^{2}}.
 	\end{equation}
In conclusion, we have
	\begin{equation}\label{eq: classical KMS post transition}
	\omega^{\overline{\rho}(h),0}_{h}\circ Q_{h}(W^{0}(f))
	\xrightarrow{h\rightarrow 0}
	\exp{-\frac{1}{2}(\langle f, (\beta H)^{-1} f\rangle + 2^{\nu} \alpha\Big|\int_{\mathbb{R}^{\nu}}\dd^{\nu} x f(x)\Big|^{2})}
	=:\omega^{\alpha,0}_{0}(W^{0}(f)).
	\end{equation}
	 By linearity, density and continuity we obtain convergence for every element of $\mathcal{W}(\mathcal{S}(\mathbb{R}^\nu),0)$. We proceed to verify the $(\delta^{0}_{0}, \beta)$-weak KMS condition on $\weylSO$. As in Section \ref{sec: sistemi a volume finito}, we are interested in the value of $\omega^{\alpha}_{0}(\Phi_{0}(h)W^{0}(k))$, for generic $h,k \in \mathcal{S}(\mathbb{R}^{\nu})$. This can be computed as in equation \eqref{eq: state on fields}
		\begin{align}
		\omega^{\alpha}_{0}(\Phi_{0}(k)W^{0}(h)) &= -i\frac{\dd}{\dd t}\omega^{\alpha}_{0}(W^{0}(h+tk))|_{t=0}\nonumber\\
		&=\frac{i}{\beta}\Re{\langle k, H^{-1} h\rangle}\omega^{\alpha}_{0}(W^{0}(k)) \nonumber \\&+ \frac{i}{\beta}\Re{2^{\nu}\alpha\int\dd^{\nu}x \overline{k}(x)\int\dd^{\nu}x h(x)}\omega^{\alpha}_{0}(W^{0}(k)).
		\end{align} By making the substitutions $k \rightarrow iHf$ and $h\rightarrow f+g$ we obtain\begin{multline}\label{eq: derivazione e stato transizione}
		\omega^{\alpha}_{0}(W^{0}(g)\delta_{0}^{0}(W^{0}(f)))\\ = \frac{1}{\beta}[ \sigma(g,f) +\Im{\int\dd^{\nu}x (\overline{f(x)}+\overline{g(x)})\int\dd^{\nu} x (Hf)(x)}]\omega^{\alpha}_{0} (W^{0}(f+g)).
		\end{multline} Now, as $Hf\in \mathcal{S}(\mathbb{R}^{\nu})$, the second integral in \ref{eq: derivazione e stato transizione} is equal to $\mathcal{F}(Hf)(0)$, but as \begin{equation}
		(Hf)(x) = \int\frac{\dd^{\nu} p}{(2\pi)^{\nu/2}} e^{-ipx}\frac{p^{2}}{2}\hat{f}(p),
		\end{equation} it follows that $\mathcal{F}(Hf)(0) = 0$. This result proves the equality \begin{equation}
		\omega^{\alpha}_{0}(\{W^{0}(f),W^{0}(g)\}) = \beta\omega^{\alpha}_{0}(W^{0}(g)\delta^{0}_{0}(W^{0}(f))).
		\end{equation} for every $f,g\in \mathcal{S}(\mathbb{R}^{\nu})$, independently from $\alpha$.
    In short, we have proved the following theorem
	
	 \begin{theorem}\label{th: Classical limit of the phase transition states} Let the dynamics of the Weyl algebra $\weylSH$ be implemented by the $^*$-automorphism $\tau^{0}_{ h}(W^{ h}(f)) = W^{ h}(e^{iHt}f)$ for all $f\in\mathcal{S}(\mathbb{R}^{\nu})$, where $H = -\Delta/2$. Then, 
		if $\overline{\rho}( h)> \rho_{c}(\beta h)$  for all $ h > 0$ and $ h(\overline{\rho}( h) -\rho_{c}(\beta, h))\xrightarrow{h\rightarrow 0^+} \alpha \geq 0$,  the net of quantum $(\tau^{ h}_{0},\beta h)$-KMS states $\omega^{\overline{\rho}(h),0}_{ h}$ with \begin{equation}
			\omega^{\overline{\rho}(h),0}_{ h}(W^{ h}(f)) = \exp{-\frac{h}{4}\bigg{(}\langle f, (I+e^{-\beta h H})(I- e^{-\beta hH})^{-1} f\rangle + 2^{\nu} (\overline{\rho}( h)-\rho_{c}(\beta h))|\int_{\mathbb{R}^{\nu}}\dd^{\nu} x f(x)|^{2}\bigg{)}} 
		\end{equation} converges {for $ h\rightarrow 0^{+}$ to the classical state $\omega^{\alpha,0}_{0}$:\begin{equation}
		\omega_{0}^{\alpha}(W^{0}(f)) = \exp{-\frac{1}{2}\bigg{(}\langle f (\beta H )^{-1} f\rangle + 2^{\nu}\alpha|\int_{\mathbb{R}^{\nu}}\dd^{\nu}x f(x)|^{2}\bigg{)}}
		\end{equation} in the sense that \begin{equation}
			\omega^{\overline{\rho}(h),0}_{ h}\circ Q_{ h}(c)\rightarrow \omega^{\alpha,0}_{0}(c),\quad  \text{for all }c \in \weylSO.
		\end{equation} }The weak$^*$-limit points $\omega^{\alpha,0}_{0}$ are labeled by the  parameter $\alpha\geq 0$ and they all satisfy the $(\delta_{0}^{0},\beta)$-weak KMS condition.
	\end{theorem}

     By an explicit computation, it is easy to verify that the value of quantum KMS states on field operators converge to the value of classical KMS states on  field functions. For example, the states $\{\omega^{\overline{\rho}(h),0}_{h}\}_{h\in(0,+\infty)}$ satisfy \begin{equation}
		\lim_{h\rightarrow 0^{+}}\omega^{\overline{\rho}(h),0}_{h}(\Phi_{h}(f_{1})\dots \Phi_{h}(f_{n})) = \omega^{\alpha,0}_{0}(\Phi_{0}(f_{1})\dots \Phi(f_{n})),
	\end{equation} for all $\{f_{1},\dots, f_{n}\}\subset \mathcal{S}(\mathbb{R}^{\nu})$, $n\in\mathbb{N}$.

{\paragraph{Infinite density limit:}
	 The most striking difference between the classical Gibbs states \eqref{eq: ansatz classico per stati KMS} and the quantum ones \ref{eq: Gibbs state} is that the latter are normal with respect to the Fock-Cook representation, allowing for a good definition of the number operator in the GNS representation while the former do not share this property for $\nu > 1$. Indeed, \begin{equation}
		\sum_{\underline{n}\in\mathbb{N}^{\nu}}\omega_{0}^{\mu, L}(a^{*}_{0}(\psi_{\underline{n}})a_{0}(\psi_{\underline{n}})) = \frac{1}{\beta}\sum_{n\in\mathbb{N}^{\nu}} \frac{1}{E_{\underline{n}(L)} - \mu},
	\end{equation} and since $E_{\underline{n}}(L)\propto \norm{\underline{n}}^{2}_{\mathbb{R}^{\nu}}$, the summation is finite if and only if $\nu = 1$. In particular, for high dimensions, the local number operator and the density are infinite. This consideration extends to  $\omega^{\mu}_0$ and $\omega^{\alpha,0}_0$. These states correspond to the infinite-density limits of the Araki-Woods quantum states, defining as such infinite-volume and infinite-density representation of the CCR.
  However, as the semi-classical parameter $h$ is dimensionless, $\alpha$ carries the dimensions of an inverse volume and can be interpreted as a \textit{renormalized density}. Indeed, $\alpha$  measures the fraction of excitations of a given type $f$ in the ground state. More precisely, for the \textit{single excitation} number operator $N_0(f):=a_0(f)^*a_0(f)$ one finds
\begin{equation}
\omega_0^{\alpha,0}(N_0(f))
= \langle f, (\beta H)^{-1} f\rangle + 2^\nu \alpha \langle f, |1\rangle\langle 1| f\rangle,
\end{equation} which decomposes as the sum of two quadratic forms: the first accounts for the contribution of the classical background, while the second quantifies the condensate fraction.}

 {What does the classical condensate represent? We have already discussed in Rmk. \ref{rmk: classical nature of the condensate} the classical nature of the quantum condensate and this same feature has been exploited in interacting systems via the successful Bogoliubov $c$-number substitution \cite{Bog_47}. In the previous limits, the condensate term is largely unaffected by the semi-classical scaling. Thus, the classical condensate still appears to describe the quantum condensate in the infinite-density limit.}

    \begin{remark}
	The operators $(H-\mu I)^{-1}$ and $H^{-1}$ are not compact, so we cannot apply the measure theoretic formalism of Sec. \ref{sec: sistemi a volume finito}. In particular, $\omega^{\mu}_{0}$ and $\omega_{0}^{\alpha,0}$ are not Gibbs states. However, focusing on $\omega^{\alpha,0}_{0}$, we can introduce the function $\theta^{\alpha}\colon \mathcal{S}(\mathbb{R})^{\nu}\ni f\rightarrow \omega^{\alpha}_{0}(W^{0}(f))\in \mathbb{C}$. $\theta^{\alpha}$ is continuous with respect to the locally convex topology of $\mathcal{S}(\mathbb{R}^{\nu})$, it satisfies $\sum_{j,k=1}^{n}\overline{z}_{j}z_{k}\theta^{\alpha}(f_{k}-f_{j}) \geq 0$ for arbitrary $\{f_{j}\}_{j\in\{1,\dots,n\}}\subset\mathcal{S}(\mathbb{R}^{\nu}),\;\{z_{j}\}_{j\in\{1,\dots,n\}}\subset \mathbb{C}, \;n\in\mathbb{N} $, and it is normalized as $\theta^{\alpha}(0) = 1$. So, by Minlos theorem \cite[Th. 4.7]{Huang_Yan_2000}, it is the Fourier transform of a Borel probability measure $\mu_{\alpha}\in \mathfrak{P}(\mathcal{S}(\mathbb{R}^{\nu})')$. Now, analyticity of the state implies that $\hat{\mu}_{\alpha}:= \theta^{\alpha}$ is differentiable on finite dimensional subspaces of $\mathcal{S}(\mathbb{R}^{\nu})$ and so, the function $\langle f, X(\cdot)\rangle$ is $\mu_{\alpha}$-integrable for all $f\in \mathcal{S}(\mathbb{R}^{\nu})$, where $\langle f, X(u)\rangle = \langle -iH f, u \rangle$. Then, by using the measure counterpart \cite[Th. 2.7]{Ammari_Sohinger_2023} of proposition  \ref{prop: caratterizzazioni equivalenti della condizione KMS cilindrica}, and exploiting the weak KMS condition satisfied by the states, we easily obtain that $\mu_{\alpha}$ is a KMS Borel probability measure as in definition \ref{def: KMS per misure}. 
\end{remark}

    {\paragraph{Case $\alpha = +\infty$ :} In this limit,  the statistic of excitations $f$ with non-zero projection onto the condensate state $|1\rangle\langle 1 | f\neq 0$ is trivial, i.e. there is always an infinite contribution to the number particles $N_0(f)$ coming from the condensate and as a result $\omega_0^{+\infty,0}(W^0(f)) = 0.$ This condition implies that the resulting state is not $C^k$ for any $k\geq 0$: formally, the expectation value of the number operator $N_0(f)$ above is infinite. Instead, excitations $g$ for which $|1\rangle\langle 1 | g = 0$ have  finite expectation values, characterized uniquely by the classical thermal background.}

    \subsection{Thermodynamic limit of the Classical System}\label{sec: limite termodinamico classico} {In the previous sections we have found states which are good candidates for describing thermodynamic equilibrium on classical Weyl algebras $\mathcal{W}(E,0)$. To complete this description, we recover the infinite-volume states as thermodynamic limits of the finite volume Gibbs ones $\omega^{L,\mu_L}_0$, in the spirit of the classical works by Cannon \cite{Cannon_73} and Lewin and Pulè \cite{LP_73}.  The main result of this section is Theorem \ref{th: Classical limit of the phase transition states}, which establishes that distinct thermodynamic limits may give rise to different infinite-volume equilibrium states for the same dynamics. The key distinction from the corresponding quantum result lies in the absence of the density as an order parameter: in the classical limit, the density diverges and is therefore no longer available to distinguish phases. A notable consequence is that, when $\mu = 0$, one can obtain more equilibrium states for every value of the temperature $1/\beta$. }

    {We take again $H_L$ to be the Dirichlet self-adjoint extension of $-\Delta/2|_{C_c^\infty(\Lambda_L)}$, $\Lambda_L:= [-L, L]^\nu$, on $\mathbb{L}^2(\mathbb{R}^\nu)$. Then , we have the following standard result which we state as a lemma for future convenience \begin{lemma}\label{lemma: strongly convergence lemma}
        Let $H_L\colon D(H_L)\subset \mathbb{L}^2(\mathbb{R}^\nu) \to \mathbb{L}^2(\mathbb{R}^\nu) $ be as above. Then, for every continuous, bounded function $h\in C_b(\mathbb{R}^\nu)$, \begin{equation}\lim_{L\to+\infty}h(H_L)g = h(H)g,\quad g \in \mathbb{L}^2(\mathbb{R}^\nu), \end{equation} that is, $h(H_L)$ strongly converges to $h(H)$ where $H$ is the unique self-adjoint extension of the restriction of $-\Delta/2$ to $C_c^\infty(\mathbb{R}^\nu)$.
    \end{lemma}\begin{proof}
       It follows from \cite[Ex. 3.1.29]{Bratteli_Robinson_87} that $e^{it H_{L}}$ strongly converges to $e^{itH}$, uniformly for $t$ in compact sets. Thus, by \cite[Lemma 5.2.25]{Bratteli_Robinson_97}, for any bounded and continuous function $h\in C_b(\mathbb{R}^\nu)$, $h(H_L)$ strongly converges to $h(H)$.
    \end{proof}}

	We focus on two possible limiting procedures.

	\paragraph{\textbf{Chemical potential $\mu< 0$ fixed:}} For $\mu <0$, consider  \begin{equation}h_\mu(x):= 1/(|x| -\mu)\in\mathbb{R},\quad h_\mu \in C_b(\mathbb{R}).\end{equation}Since $H_L\geq 0$, we have $h_\mu(H_L) = (H_L -\mu)^{-1}$ and $h_\mu(H_L) \to h_{\mu}(H)$ strongly by Lemma \ref{lemma: strongly convergence lemma}. Then, for any compactly supported $f\in \mathbb{L}^2(\mathbb{R}^\nu)$ we have the thermodynamic limit \begin{equation}\label{eq: thermodynamical limit before phase transition}
		\lim_{L\to +\infty} \omega^{L,\mu}_{0}(W^{0}(f)) = \lim_{L\to +\infty}\exp{-\frac{1}{2\beta}\langle f, h_\mu(H_L) f\rangle} =  \exp{-\frac{1}{2\beta}\langle f, h_\mu(H) f\rangle} = \omega^{\mu}_{0}(W^{0}(f)).
	\end{equation} By linearity and continuity, the limit in \eqref{eq: thermodynamical limit before phase transition} holds for the generic element of $\cup_{L>0}\mathcal{W}(E_{L},0)$. 
	\paragraph{Variable chemical potential $\mu_{L}$:} We take $\nu\geq 3$. If we consider a sequence of chemical potentials ${\mu_{L} < E_{0}(L)}$ converging in the Thermodynamic limit to some $\overline{\mu} <0$, we are in a similar situation to the one described in the previous paragraph. By the same reasoning, we obtain the infinite-volume state $\lim_{L\to+\infty}\omega_0^{L,\mu_L} = \omega_{0}^{\overline{\mu}}$ as a weak$^*$-limit.
	
{Crucially, the situation changes if we send the chemical potential to $0$ in the thermodynamic limit.  In the quantum regime, this is done by fixing the density $\overline{\rho}\geq \rho_{c}(\beta h)$  as in Eq. \eqref{eq: condition on the chemical pot}, while enlarging the volume $\Lambda_{L}$. In this classical framework, we can only require the chemical potential to be smaller than the ground-state energy $\mu_L<E_0(L)$ for defining Gibbs state for every value of $L$. As $E_0(L) = O(1/L^2)$ This condition ensures that the following quantity \begin{equation}\label{eq: classical condensate densityy}
     \frac{1}{|\Lambda_L|\beta(E_0(L)-\mu_L)},
\end{equation} can only have a positive limit (if it exists). $\alpha\geq 0$ as $L\to+\infty$. As we will see in Th. \ref{th: classical thermodynamical limit after phase transition},  \eqref{eq: classical condensate densityy} is related to the condensate fraction $\alpha$ in Eq. \eqref{eq: insospettabile densità}, and indeed, Eq. \eqref{eq: classical condensate densityy} corresponds to the classical limit of the renormalized ground state density \begin{equation}
    \lim_{h\to 0^+} \frac{1}{|\Lambda_L|}h\frac{e^{-\beta h (E_0(L)-\mu_L)}}{1- e^{-\beta h(E_0(L)-\mu_L)}} = \frac{1}{|\Lambda_L|\beta(E_0(L)-\mu_L)}.
\end{equation}}

To proceed, we need the following technical Lemma, which is a slight variation of \cite[Prop. 5.2.31]{Bratteli_Robinson_97} 
	
\begin{lemma}\label{lemma: convergenza dell'hamiltoniana inversa}
 {Consider the Hamiltonian operator $H_{L}\colon D(H_L)\to \mathbb{L}^2(\mathbb{R}^\nu)$, obtained as the self-adjoint extension with Dirichlet boundary conditions of $-\Delta/2|_{C_c^\infty(\Lambda_L)}$. Then, for every compactly supported $f\in \mathbb{L}^{2}(\mathbb{R}^\nu)$ we have \begin{equation}\label{eq: inverse laplacian action}
	\lim_{L\to+\infty}\langle f, H_{L}^{-1} f \rangle = \langle f, H^{-1} f\rangle := \int \frac{\dd p^{\nu}}{(2\pi)^{\nu}} \frac{2}{p^{2}}|\hat{f}(p)|^{2},
\end{equation} where $H$ is the unique self adjoint extension of $-\Delta/2$ on $C_c^\infty(\mathbb{R}^\nu)$.}
\end{lemma}
\begin{proof}
Notice that the integral on the right of \eqref{eq: inverse laplacian action} is well-defined because $\nu\geq 3$ and $\hat{f}\in\mathbb{L}^2(\mathbb{R}^\nu)$ is analytic because by assumption $f$  is compactly supported.
        Secondly, we know from Lemma \ref{lemma: strongly convergence lemma} that for every bounded continuous function $h\in C_{b}(\mathbb{R})$, $h(H_{L})\to h(H)$ strongly .
        Finally, {note that our Hamiltonian trivially satisfies the assumptions of \cite[Th. 6.3.12]{Bratteli_Robinson_97}, since $H_L$ is only given by the kinetic term, closed with Dirichlet boundary conditions, with null potential}. Thus, \cite[Cor. 6.3.13]{Bratteli_Robinson_97} establishes that the function $\Lambda_{L} \rightarrow \langle f, e^{-\alpha H_{L}} f  \rangle \in \mathbb{R}$ is monotone increasing in $\Lambda_{L}$ for every $\alpha>0$.
By functional calculus we have
	\begin{equation}
		\langle f, H_{L}^{-1} f \rangle
		=\lim_{ h \rightarrow 0^{+}} h \langle f, (I - e^{- h  H_{L}})^{-1} f \rangle
		= \inf_{ h \in (0,1)}  h  \langle f, (I -e^{- h  H_{L}})^{-1}f\rangle.
	\end{equation}
	Hence, we can compute
	\begin{align}
		\lim_{L\to+\infty } \langle f, H_{L}^{-1} f\rangle
		&= \lim_{\Lambda_{L}\uparrow \mathbb{R}^{\nu}}\inf_{ h \in (0,1)}  h  \langle f, (I- e^{- h  H_{L}})^{-1} f\rangle\nonumber
        \\
		&= \sup_{\Lambda_{L}}\inf_{ h \in (0,1)}  h  \sum_{n=0}^{\infty} \langle f, e^{-n h  H_{L}} f\rangle\nonumber
        \\
		&= \sup_{\Lambda_{L}}\inf_{ h \in (0,1)} \sup_{z\in (0,1)} h \sum_{n=0}^{\infty}\langle f, z^{n}e^{-n h  H_{L}}f\rangle \nonumber \\
        &\label{eq: monotonicity in Lambda}=\inf_{ h \in (0,1)} \sup_{z\in (0,1)} h \sum_{n=0}^{\infty}\langle f, z^{n}e^{-n h  H}f\rangle\\
		&\label{eq: insertion of z}=\inf_{ h \in  (0,1)} \sup_{z\in (0,1)}  h  \langle f, (I- ze^{- h  H} )^{-1} f\rangle
        \\
        &\label{eq: final result}= \langle f, H^{-1} f\rangle,
        \end{align}
\end{proof} {where we have inserted a $z$ to be able to compute the series explicitly in Eq. \eqref{eq: insertion of z} for $\norm{ze^{-hH}}<1$, and in the fourth equality \eqref{eq: monotonicity in Lambda} we have exploited the monotonicity in $\Lambda_L$ to move the supremum inside the summation.}

Now, we can state the main theorem of this section 
	\begin{theorem}(\textbf{Phase transitions for classical states})\label{th: classical thermodynamical limit after phase transition}
{ Fix any $\beta >0$ and consider an increasing net of boxes $\Lambda_L:=[-L,L]^\nu$.  For any $\alpha\in [0,+\infty)$ and net of chemical potentials $(\mu_L)_{L>0}$ satisfying
  \begin{equation}\label{eq: chemical potential net}
    \mu_L < E_0(L),\quad \lim_{L\to+\infty} \frac{1}{|\Lambda_L|}\frac{1}{\beta(E_0(L)-\mu_L  )} = \alpha,
 \end{equation} 
 let $\omega_0^{L,\mu_{L}}$ be the $(\delta^{L,\mu_{L}}_{0},\beta)$-weak KMS state corresponding to the dynamics $\delta_{0}^{L,\mu_{L}}(W^{0}(f)) = i\Phi_{0}(i(H_{L} -\mu_{L})f)W^{0}(f)$. It follows that the limit \begin{equation}
		\overline{\omega^{\alpha,0}_{0}}(c) = \lim_{\Lambda_{L}\uparrow \mathbb{R}^{\nu}}\omega^{L,\mu_{L}}_{0}(c)
	\end{equation}  exists for all $c\in \bigcup_{\Lambda_{L}\subset\mathbb{R}^{\nu}}\mathcal{W}(E_{L},0)$ and defines an analytic state such that
 \begin{equation}
	\overline{\omega_{0}^{\alpha,0}}(W^{0}(f)) = \exp{-\frac{1}{2\beta}\bigg{(} \langle f, H^{-1} f\rangle + 2^{\nu}\alpha |\int_{\mathbb{R}^{\nu}} \dd^{\nu} x f(x)|^{2}\bigg{)}},\quad f\in \bigcup_{\Lambda_{L}\subset \mathbb{R}^{\nu}}E_{L}.
	\end{equation} This state satisfy the $(\delta^{0}_{0},\beta)$-weak KMS condition, \textit{regardless of the value of $\alpha$}, where $\delta_{0}^{0}(W^{0}(f)) = i\Phi_{0}(iHf)W^{0}(f)$, for $ f\in D(H)$. Moreover, the state $\overline{\omega^{\alpha,0}_{0}}$ can be extended uniquely and continuously over $\weylSO$, to a $(\delta_{0}^{0},\beta)$-weak KMS state $\omega^{\alpha,0}_{0}$.}
	\end{theorem}
	\begin{proof}
     {Extract from the net $(\mu_L)_{L>0}$ an arbitrary sequence $(\mu_{L_{n}})_{n\geq 0}$, with $L_n \xrightarrow{n\to+\infty}+\infty$ and suppose that we can find an infinite subsequence $(L_{n_k})_{k\geq 0}$ satisfying  $\mu_{L_{n_k}} \leq 0$ for every $k$. Then, clearly \begin{equation}
        0\leq \alpha = \lim_{k\to+\infty}\frac{1}{|\Lambda_{L_{n_k}}|\beta(E_0(L_{n_k})-\mu_{L_{n_k}})} \leq\lim_{k\to+\infty}\frac{1}{|\Lambda_{L_{n_k}}|\beta(E_0(L_{n_k}))} =0,
    \end{equation} i.e. $\alpha = 0$. For these subsequences, the same argument of Lemma \ref{lemma: convergenza dell'hamiltoniana inversa} proves the convergence \begin{equation}
       \lim_{k\to +\infty }\exp{\frac{1}{2\beta}\langle f, (H_{L_{n_k}}-\mu_{L_{n_k}})^{-1} f\rangle} =  \exp{\frac{1}{2}\langle f, (\beta H)^{-1} f\rangle }. 
    \end{equation} Viceversa, if $\alpha >0$, these infinite subsequences do not exist. For sequences with an infinite number of positive chemical potentials $\mu_{L_n}$, we can always extract a positive subsequence $\mu_{L_{n_m}}\geq 0$. In the following, we prove that  for such subsequences, the states $\omega_0^{L_m, \mu_{L_m}}$ all converge to the same limit. Thus, we can suppose without loss of generality that $\mu_L>0$ for every $L>0$.}

		We introduce a partial ordering for multi-indexes $\underline{n}\in\mathbb{N}^{\nu}$ as follows: $\underline{n}> \underline{m}$ if $n_{i}\geq m_{i}$ for all $i\in \{1,\dots,\nu\}$ and $n_{j}>m_{j}$ for at least one index $j$. $(1,\dots, 1)\in\mathbb{N}^\nu$ will be labelled by $0$, to conform with the standard notation where $E_{0}(L) := E_{\{1,\dots,1\}}(L)$ is the ground state energy of the system.
		
		We expand the value of the generic state $\omega^{L,\mu_{L}}_{0}$ applied on Weyl elements $W^{0}(f)\in\mathcal{W}(E_{L},0)$ as \begin{equation}\label{eq: expansion of states}
			\omega^{L,\mu_{L}}_{0}(W^{0}(f)) = \exp{-\frac{1}{2\beta}\langle f, (H_{L} -\mu_{L}I)^{-1}f\rangle} = \exp{-\frac{1}{2\beta}\sum_{\underline{n}\in \mathbb{N}^{\nu}}\langle f, P_{\underline{n}}(L) f\rangle\frac{1}{E_{\underline{n}(L) } -\mu_{L}} },
		\end{equation} where we have introduced the projectors $P_{\underline{n}}(L)$ on the $\underline{n}$-element of the orthonormal basis of eigenvectors of $H_{L}$. We examine the terms appearing in the summation at the exponent in \eqref{eq: expansion of states}. For the ground state we have \begin{align}
			&\frac{1}{|\Lambda_{L}|}\frac{1}{\beta(E_{0}(L)-\mu_{L})}  \xrightarrow{L\to+\infty} \alpha,\\
			& |\Lambda_{L}| \langle f, P_{0}(L) f\rangle = 2^{\nu}|\int_{\Lambda_{L}}\dd^{\nu}x f(x) \prod_{i=1}^{\nu}\sin{(\frac{\pi}{2L} (x-L_{i}))}|^{2}\xrightarrow{L\to+\infty} 2^{\nu}|\int_{\mathbb{R}^{\nu}}\dd^{\nu}x f(x)|^{2},
		\end{align} so that, \begin{equation}\lim_{L\to+\infty}\langle f, P_{0}(L) f\rangle\frac{1}{\beta(E_{0}(L)-\mu_{L})} = 2^{\nu}\alpha |\int_{\mathbb{R}^{\nu}}\dd^{\nu} x f(x)|^{2}.
		\end{equation} Instead,  for $\underline{n}>(1,\ldots,1)$ the product between the volume and the projection term is bounded
  \begin{equation}
			|\Lambda_{L}|\langle f P_{\underline{n}}(L) f\rangle= 2^{\nu}|\int_{\Lambda_{L}}\dd^{\nu}x f(x) \prod_{i=1}^{\nu}\sin{(\frac{n_{i}\pi}{2L} (x-L))}|^{2} \leq 2^{\nu} (\int_{\mathbb{R}^{\nu}} |f(x)|)^{2} < \infty,
		\end{equation} {while the energy term converges to 0 when multiplied by $|\Lambda_{L}|^{-1}$ \begin{equation}
			\frac{1}{|\Lambda_{L}|\beta(E_{\underline{n}}(L)- \mu_{L})}= \frac{1}{|\Lambda_L|\beta(E_{\underline{n}}(L)- E_0(L) + E_0(L)-\mu_L)} \leq\frac{1}{|\Lambda_L|\beta(E_{\underline{n}}(L)- E_0(L))} \xrightarrow{\Lambda_L\uparrow \mathbb{R}^\nu} 0 .
		\end{equation}}
   Thus, we have for an arbitrary $\underline{m}> \{1,\dots, 1\}$  \begin{equation}\label{eq: indipendenza dal primo pezzo}
			\lim_{L\to+\infty} \sum_{\underline{n}> \{1,\dots,1\}}^{\underline{m}} \langle f, P_{\underline{n}}(L) f\rangle\frac{1}{E_{\underline{n}}(L) -\mu_{L}} = 0.
		\end{equation} At this point, we can fix an arbitrary $\underline{m}>\{1,\dots,1\}$ and consider the tail of the series \begin{equation}
			\alpha^{\underline{m}}(\mu_{L}):=\sum_{\underline{n}>\underline{m}}\langle f, P_{\underline{n}}(L) f\rangle\frac{1}{E_{\underline{n}}(L) -\mu_{L}}.
		\end{equation} For any such $\underline{m},\;\underline{m'}$ we will have
  \begin{equation}\limsup_{L\to+\infty}\alpha^{\underline{m}}(\mu_{L}) = \limsup_{L\to+\infty}\alpha^{\underline{m'}}(\mu_{L})\;,\qquad\liminf_{L\to+\infty}\alpha^{\underline{m}}(\mu_{L}) = \liminf_{L\to+\infty}\alpha^{\underline{m'}}(\mu_{L})\,.\end{equation}
		Now, $f(x) = 1/(1-x)$ is strictly increasing and convex for $x\in [0,1)$, hence,  $\alpha^{\underline{m}}(\mu_{L})$ is in turn  strictly increasing and convex as a function of $0<\mu_{L} < E_{0}(L)< E_{\underline{m}}(L)$. { Monotonicity entails
		\begin{align}\label{eq: limite con monotonia}\liminf_{L\to+\infty}\alpha^{\underline{m}}(\mu_{L}) \geq\liminf_{L\to+\infty}\alpha^{\underline{m}}(0)
			= & \liminf_{L\to+\infty}\left(\langle f, H_{L}^{-1} f\rangle
			- \sum_{\substack{\underline{n}\in \mathbb{N}^{\nu} \\ \underline{n}\geq\{1,\ldots,1\}}}^{\underline{m}}\langle f, P_{\underline{n}}(L) f\rangle
			\frac{1}{E_{\underline{n}}(L)}\right)\notag \\ \geq & \liminf_{L\to+\infty}\langle f, H_L^{-1} f\rangle -\limsup_{L\to+\infty}\sum^{\underline{m}}_{\substack{\underline{n}\in \mathbb{N}^{\nu} \\ \underline{n}\geq\{1,\ldots,1\}}}|\Lambda_L|\langle f, P_{\underline{n}}(L) f\rangle
			\frac{1}{|\Lambda_L| E_{\underline{n}}(L)} = \langle f, H^{-1} f \rangle,
		\end{align}
	  where we have used lemma \ref{lemma: convergenza dell'hamiltoniana inversa} for the convergence of $\langle f, H_{L}^{-1}f\rangle$ and $E_{\underline{n}}(L) |\Lambda_L| \geq  L^{\nu-2} \xrightarrow{L\to+\infty}  +\infty$.} We exploit convexity to obtain the other relevant estimate. First,  $\alpha^{\underline{m}}(\cdot)$ is differentiable, with derivative given by \begin{equation}\label{eq: first important inequality}
			(\frac{\dd}{\dd \mu}\alpha^{\underline{m}})|_{\mu = \mu_{L}} = \sum_{\underline{n}>\underline{m}}\langle f, P_{\underline{n}}(L) f\rangle\frac{1}{(E_{\underline{n}}(L) -\mu_{L})^{2}} \leq \alpha^{\underline{m}}(\mu_{L})\frac{1}{E_{\underline{m}}(\Lambda_{L})-\mu_{L}}.
		\end{equation}  The convexity property can be expressed as \begin{equation}\label{eq: second important inequality}
			\frac{\alpha^{\underline{m}}(\mu_L)-\alpha^{\underline{m}}(0) }{\mu_L}\leq(\frac{\dd}{\dd \mu}\alpha^{\underline{m}})|_{\mu= \mu_{L}} \leq \alpha^{\underline{m}}(\mu_{L})\frac{1}{E_{\underline{m}}(\Lambda_{L})-\mu_{L}},
		\end{equation}
		Now, for $\underline{m}>\{1,\dots,1\}$, if we call $m$ the greatest component of $\underline{m}$, we have
		{\begin{equation}
			\frac{E_{0}(L)}{E_{\underline{m}}(\Lambda_{L})}
			\leq \frac{\nu}{m^{2}}\,.
		\end{equation}}
		Choosing $m$ large enough so that $E_{\underline{m}}(\Lambda_{L})> 2E_{0}(L)$, inequality in \eqref{eq: second important inequality} can be solved for $\alpha^{\underline{m}}(\mu_{L})$, obtaining \begin{equation}
			\alpha^{\underline{m}}(\mu_{L})\leq \alpha^{\underline{m}}(0)\frac{E_{\underline{m}}(\Lambda_{L}) - \mu_{L} }{E_{\underline{m}}(\Lambda_{L}) - 2\mu_{L}}\leq \langle f, H^{-1} f\rangle\frac{1+(\alpha|\Lambda_{L}| E_{\underline{m}}(\Lambda_{L}))^{-1}}{1-2\nu/m^{2}}.
		\end{equation} Thus, for a fixed $\underline{n}>\{1,\dots,1\}$ we have\begin{equation}
			\limsup_{\Lambda_{L}\uparrow \mathbb{R}^{\nu}}\alpha^{\underline{n}}(\mu_{L})= \limsup_{\Lambda_{L}\uparrow \mathbb{R}^{\nu}}\alpha^{\underline{m}}(\mu_{L}) \leq \langle f, H^{-1}f\rangle\frac{1}{1-2\nu/m^{2}}.
		\end{equation} Since $\underline{m}$ is arbitrary, we obtain the chain of inequalities \begin{equation}
			\langle f, H^{-1} f\rangle 	\leq \liminf_{\Lambda_{L}\uparrow \mathbb{R}^{\nu}}\alpha^{\underline{m}}(\mu_{L}) \leq \limsup_{\Lambda_{L}\uparrow \mathbb{R}^{\nu}}\alpha^{\underline{m}}(\mu_{L})\leq\langle f, H^{-1} f\rangle,
		\end{equation} which implies that\begin{equation}
			\lim_{L\to+\infty} \sum_{\underline{n}> \{1,\dots,1\}} \langle f, P_{\underline{n}}(L) f\rangle\frac{1}{E_{\underline{n}}(L) -\mu_{L}} = \langle f, H^{-1} f\rangle.
		\end{equation} Then, for an arbitrary $f\in \bigcup_{\Lambda_{L}\subset \mathbb{R}^{\nu}} E_{L}$ we have \begin{multline}
			\lim_{\Lambda_L\uparrow\mathbb{R}^{\nu}}\omega^{\Lambda_{L},\mu_{L}}_{0}(W^{0}(f))  = \lim_{L\to+\infty}\exp{-\frac{1}{2\beta}\langle f, (H_{L}- \mu_{L} I)^{-1}f\rangle} \\=  \exp{-\frac{1}{2}\bigg{(} \langle f, (\beta H)^{-1} f\rangle + 2^{\nu}\alpha |\int_{\mathbb{R}^{\nu}} \dd^{\nu} x f(x)|^{2}\bigg{)}}= \overline{\omega^{\alpha,0}_{0}}(W^{0}(f))\,.
		\end{multline}
        By continuity and linearity we can verify the convergence for all elements of $\bigcup_{\Lambda_{L}\subset\mathbb{R}^{\nu}}\mathcal{W}(E_{L},0)$. Since every element of $\mathcal{S}(\mathbb{R}^{\nu})$ can be approximated by functions in $ \bigcup_{\Lambda_{L}\subset \mathbb{R}^{\nu}}\mathbb{L}^{2}(\Lambda_{L},\dd^{\nu}x)$ with the norm \begin{equation}
			\norm{f}_{\omega}:= \sqrt{|\hat{f}(0)|^{2} + \int_{\mathbb{R}^{\nu}}\frac{\dd^{\nu}p}{(2\pi)^{\nu}} |\hat{f}(p)|^{2}\frac{1}{p^{2}}},
		\end{equation} and since if $f_{n}\xrightarrow{\norm{\cdot}_{\omega}}f$ it follows that $W^{0}(f_{n})\rightarrow W^{0}(f)$ pointwise, we can extend by continuity the state to $\omega^{\alpha,0}_{0}\in \weylSO$, whose value on $\{W^{0}(f), \;f\in\mathcal{S}(\mathbb{R}^{\nu})\}$ is given by\begin{equation}
		\omega^{\alpha,0}_{0}(W^{0}(f)) = \exp{-\frac{1}{2}\bigg{(} \langle f, (\beta H)^{-1} f\rangle + 2^{\nu}\alpha |\int_{\mathbb{R}^{\nu}} \dd^{\nu} x f(x)|^{2}\bigg{)}}.
		\end{equation} 
        
        That the states $\omega_0^{\alpha,0}$ satisfy the $(\delta^0_0,\beta)$-weak KMS condition independently from $\alpha$ has been already verified in Th. \ref{th: Classical limit of the phase transition states}. This concludes the proof.
	\end{proof}

    \appendix
 
	\section{Weak derivations}\label{sec: appendice A}
\begin{proof}[Proof of proposition \ref{prop: caratterizzazione derivazioni continue I}:]\label{proof: caratterizzazione delle derivazioni continue e lineari}
	From the definition \ref{def: weak derivation} we know that if $\lambda,\mu\in\mathbb{C}$ and  $f,g\in E(\delta_{0})$, then $\lambda f + \mu g \in E(\delta_{0})$. Moreover, $\delta_{0}(W^{0}(f)) =: \Psi(f)\in C(E'_{\norm{\cdot}})$ with $\mathbb{C}\ni\lambda\rightarrow \Psi(\lambda f)[g]$ continuous for every $g\in E'_{\norm{\cdot}}$.
    Thanks to $W^{0}(f)$ being unitary and pointwise continuous in $f$, we can always write $\Psi(f) = \varphi(f)W^{0}(f)$ with $\varphi(f)\in C(E'_{\norm{\cdot}})$, for all $f\in E(\delta_{0})$.
    Now, by exploiting property (b) in \ref{def: weak derivation}  we obtain additivity:\begin{equation}
		\delta_{0}(W^{0}(f+g)) = \varphi(f+g)W^{0}(f+g) = (\varphi(f)+\varphi(g))W^{0}(f+g),\quad f,g\in E(\delta_{0}),
	\end{equation}The latter relation implies $\varphi(f+g) = \varphi(f)+\varphi(g)$. Secondly, it is clear that $0\in E(\delta_{0})$ and $\delta_{0}(W^{0}(0)) = 0$. Lastly, we derive the $\mathbb{R}$-linearity by a standard approximation argument. We start with \begin{equation}
		\delta(W^{0}(2f)) = \varphi(2f)W^{0}(2f) = 2\varphi(f)W^{0}(2f),\quad f\in E(\delta_{0}).
	\end{equation} By induction, it is easy to verify that $\forall n\in \mathbb{N}$ $\varphi(n f) = n\varphi(f)$. By simply taking $f = \Tilde{f}1/n$, we obtain that for every $f\in E(\delta_{0})$, $n\in \mathbb{N},\; n \neq 0$, $\varphi(f/n) = \varphi(f)/n$. So, for every $q\in \mathbb{Q}$, we have $\varphi(q f) = q\varphi(f)$. Now, for fixed $\lambda\in \mathbb{R}$, for all $\epsilon > 0,\; \text{there exists}\; q\in \mathbb{Q}$ such that $|q-\lambda| < \epsilon$. Moreover, for fixed $g\in E'_{\norm{\cdot}}$ there exists $\delta > 0$ such that $|\lambda -\lambda'|< \delta$ implies $|\varphi(\lambda f)(g)- \varphi(\lambda' f) (g)|< \epsilon/2$. By choosing $q\in \mathbb{Q}$ such that $|q-\lambda|< \min\{\delta,\epsilon/(2|\varphi(f)|,\epsilon/2\}$\footnote{if $\varphi(f) = 0$ it is sufficient to take the minimum between $\delta$ and $\epsilon/2$.} we obtain \begin{align}
		|\varphi(\lambda f)[g] - \lambda\varphi(f)[g]| \leq |\varphi(f)[g]| |\lambda - q| + |\varphi(\lambda f)[g] - \varphi(q f)[g]|< \epsilon.
	\end{align} Since $\epsilon$ is arbitrary, we obtain that for every $g\in E$, $\varphi(\lambda f)[g] = \lambda \varphi(f)[g]$, and so, we conclude that for every $f\in E(\delta_{0}),\; \lambda \in \mathbb{R}$ we have $\varphi(\lambda f) = \lambda \varphi(f)$. 
	
	Now, thanks to the linearity property, we know that $W^{0}(-f)\delta_{0}(W^{0}(f)) = \varphi(f)$ is a $\mathbb{R}$-linear, $\sigma(E_{\norm{\cdot}}',E)$-continuous functional. Moreover, from \begin{equation}
		\varphi(f)^{*} = (W^{0}(-f)\delta_{0}(W^{0}(f)))^{*}  = -\varphi(f),
	\end{equation} follows that $\varphi(f)[g]\in i\mathbb{R}$ for all $g\in E$. Then, we must have \begin{equation}
		\varphi(f)[g] = i\Re{ \langle g, \Tilde{f}\rangle},\quad \forall g\in E,
	\end{equation} for some fixed $\Tilde{f}\in E$. Now, we define an operator $L_{0}: E(\delta_{0})\rightarrow E$ by setting \begin{equation}
		L_{0}f := \Tilde{f}.
	\end{equation} This is a well-defined operator by non degeneracy of the scalar product. Linearity of the operator follows from the linearity of $\varphi(f)$ on $f$.
\end{proof}

	\begin{proof}[Proof of proposition \ref{prop: chiusura pointwise I}]
			$(\Rightarrow)$ If $\delta_{0}$ is pointwise closed, let us consider a sequence $\{f_{n}\}_{n\in\mathbb{N}}\subset E(\delta_{0})$ such that $f_{n}\rightarrow f$ and $L_{0}f_{n}\rightarrow \Tilde{f}$ in the $\sigma(E,E'_{\norm{\cdot}})$ topology. We have \begin{equation}
				\delta_{0}(W^{0}(f_{n})) = i\Phi_{0}(L_{0}f_{n})W^{0}(f_{n}) \rightarrow i\Phi_{0}(\Tilde{f})W^{0}(f)
			\end{equation} in the pointwise convergence topology.
			But as $\delta_{0}$ is pointwise closed, this implies that $f\in E(\delta_{0})$ and $i\Phi_{0}(\Tilde{f})W^{0}(f) = i\Phi_{0}(L_{0}f)W^{0}(f)$, which implies $f\in D(L_{0})$ and $\Tilde{f} = L_{0}f$.
			
			($\Leftarrow$) Now we assume that $L_{0}$ is closed. We consider $\{f_{n}\}_{n\in\mathbb{N}}\subset E(\delta_{0})$ such that $f_{n}\rightarrow f$ in $\sigma(E,E'_{\norm{\cdot}})$ and $\delta_{0}(W^{0}(f_{n}))\rightarrow \Psi \in C(E'_{\norm{\cdot}})$ pointwise. Then, if we call $\varphi:= W^{0}(-f)\Psi$, we have \begin{align}
				\varphi[\lambda g] &= \lim_{n\rightarrow \infty} W^{0}(-f_{n})\delta_{0}(W^{0}(f_{n}))[\lambda g] = \lambda \varphi[g],\nonumber \\
				\varphi[g+h] &= \lim_{n\rightarrow \infty} W^{0}(-f_{n})\delta_{0}(W^{0}(f_{n}))[ g+h] = \varphi[g] +\varphi[h],
			\end{align} so, $\varphi$ is $\mathbb{R}$-linear and continuous in the $\sigma(E'_{\norm{\cdot}},E)$-topology.
			Moreover, since for all $g\in E'_{\norm{\cdot}}$ it holds $W^{0}(-f_{n})[g]\delta_{0}(W^{0}(f_{n}))[g]\in i\mathbb{R}$, we conclude that $\varphi^{*} = -\varphi$.
			These conditions imply that {there exists a unique $\Tilde{f}\in E$ such that $\varphi[g] = i\Phi_{0}(\Tilde{f})[g]$ for all $g\in E'_{\norm{\cdot}}$}. This is equivalent to  $\sigma(g, L_{0}f_{n}) \rightarrow \sigma(g,\Tilde{f})$.
			Since $L_{0}$ is closed, in particular it is also $\sigma(E,E'_{\norm{\cdot}})-\sigma(E,E'_{\norm{\cdot}})$ closed and this implies that $f\in D(L_{0})$ and $L_{0}f = \Tilde{f}$, or in other words, {by definition of $D(L_0)$}, that $W^{0}(f)\in D(\delta_{0})$ and $\delta_{0}(W^{0}(f)) = \Psi$.
	\end{proof}
	
	\section{Measures and cylindrical measures on Hilbert Spaces}\label{sec: appendice B}

	\begin{proof}[Proof of lemma \ref{lemma: estensione a funzioni limitate} ]
 Let us take $F,G\in C^{\infty}_{b,cyl}(\mathcal{H}^{-s})$. Then, there exists a common $n\in \mathbb{N}$ sufficiently large such that both function can be written as $F(u) = \varphi(\pi_{n}(u)),\; G(u) = \psi(\pi_{n}(u))$, for some $\varphi, \psi \in C^{\infty}_{b}(\mathbb{R}^{2n})$. {Now, for every function $\varphi,\psi$ in $C^{\infty}_{b}(\mathbb{R}^{2n})$, we can find sequences  $\{\psi_m\}_{m\in\mathbb{N}},\{\varphi_{m}\}_{m\in\mathbb{N}}\subset C^{\infty}_{c}(\mathbb{R}^{2n})$ approximating $\psi,\varphi$ and their derivatives pointwise. We will call $F_m(u) := \varphi_m(\pi_n(u))$ and $G_m(u):= \psi_m(\pi_n(u))$. Then, $F_m,G_m \in C_{c,cyl}^\infty(\Phi)$.
		Since the measures which define the cylindrical measure $\mu_{*}$, $\{\mu_{L_{n}}\}_{n\in\mathbb{N}}$ are all normalized to one, we  conclude by the dominated convergence that \begin{align}
		    &\int_{\mathcal{H}^{-s}}\{F,G\}(u)\dd\mu(u) = \lim_{m\to +\infty} \int_{L_n}\{F_m,G_m\}(u) \dd \mu_{L_n}(u) \notag \\ =& \lim_{m\to+\infty}\beta\int_{L_n}\langle\nabla F_m(u),X(u)\rangle G_m(u)\dd\mu_{L_n}(u) = \beta\int_{\mathcal{H}^{-s}}\langle\nabla F(u), X(u)\rangle G(u)\dd\mu(u).
		\end{align}}
		\end{proof}
		\begin{proof}[Proof of proposition \ref{prop: caratterizzazioni equivalenti della condizione KMS cilindrica}.]
			(i)$ \Rightarrow$(ii) From the definition of cylindrical KMS-state follows that $ \mathcal{H}^{-s}\ni u\rightarrow\langle\varphi, X(u)\rangle\in\mathbb{R}$ is $\mu_{L_{n}}$-integrable for all $\varphi\in L_{n}$, for all $n\in\mathbb{N}$. We choose $\varphi_{1},\varphi_{2}\in \mathbb{R}$-span$L_{n_{0}}$ for arbitrary $n_{0}\in\mathbb{N}$. By taking $f,g \in \mathcal{H}^{s}$  such that $f = \varphi_{1}$, $\varphi_{2} = f+g$ we have the equality\begin{align}
		\Re{\langle i \varphi_{1}, \varphi_{2}\rangle}e^{i\langle \varphi_{2}, u\rangle}& = \Re{\langle if, g \rangle} e^{i\langle f+g, u\rangle} = -\{e^{i\langle f, u\rangle}, e^{i\langle g, u\rangle}\}.
		\end{align} Integrating this equality with the  measures $\{\mu_{L_{n}}\}_{n\geq n_{0}}$, we get \begin{multline}
		\Re{\langle i\varphi_{1}, \varphi_{2}\rangle }\int_{L_{n}}e^{i\langle\varphi_{2}, u\rangle}\dd\mu_{L_{n}}(u) = -\int_{L_{n}}\{e^{i\langle f,u\rangle}, e^{i\langle g,u\rangle}\}\dd\mu_{L_{n}}(u)	\\ = -i\beta\int_{L_{n}}\langle f, X(u)\rangle e^{i\langle  f+g, u\rangle}\dd\mu_{L_{n}}(u) = -i\beta\int_{L_{n}}\langle\varphi_{1}, X(u) e^{i\langle \varphi_{2}, u\rangle}\dd\mu_{L_{n}}(u)
		\end{multline}
    where in the second equality we have used the cylindrical KMS-condition for $F(\cdot) = e^{i\langle f, \cdot\rangle}, \; G(\cdot) = e^{i\langle g,\cdot \rangle}\in C^{\infty}_{b,cyl}(\mathcal{H}^{-s})$.
		
		(i)$ \Leftarrow $(ii) We take arbitrary $F, G\in C^{\infty}_{c,cyl}(\mathcal{H}^{-s})$, which can be written in terms of some $\varphi\in C^{\infty}_{c}(\mathbb{R}^{2n})$ and $\psi\in C^{\infty}_{c}(\mathbb{R}^{2m})$ for some $n,m\in \mathbb{N}$. We can express the Poisson bracket $\{F, G\}(u)$ in terms of the Fourier transforms of $\varphi$ and $\psi$ \cite[Lem. 2.2]{Ammari_Sohinger_2023}
		\begin{multline}
			\{F, G\}(u) = -(2\pi)^{\frac{n+m}{2}}\int_{\mathbb{R}^{n}\cross\mathbb{R}^{m}} \dd\mathcal{L}_{2n}(\lambda,\lambda')\dd\mathcal{L}_{2m}(\mu,\mu')\hat{\varphi}(\lambda_{1},\lambda'_{1},\dots,\lambda_{n},\lambda'_{n})\hat{\psi}(\mu_{1},\mu_{1}',\dots, \mu_{m},\mu'_{m})\\ \sum_{k,j=1}^{n}\Re{\langle i( \lambda_{j}e_{j} +\lambda'_{j}f_{j}), \mu_{k}e_{k} + \mu'_{k}f_{k}\rangle}e^{i\sum_{j=1}^{n}\lambda_{j}\langle e_{j}, u\rangle+\lambda_{j}'\langle f_{j}, u\rangle}e^{i\sum_{k=1}^{m}\mu_{k}\langle e_{k}, u\rangle+\mu_{k}'\langle f_{k}, u\rangle}.
		\end{multline} Now, we take the identity \begin{equation}
		\Re{\langle if, g\rangle}\int e^{i\langle f+g, u\rangle}\dd\mu_{*}(u) = -i\beta\int\langle f, X(u)\rangle e^{i\langle f+g, u\rangle}\dd\mu_{*}(u),
		\end{equation} with $f = \sum_{j=1}^{n}\lambda_{j}e_{j}+\lambda_{j}'f_{j}$ and $g = \sum_{k=1}^{m}\mu_{k}e_{k} +\mu_{k}'f_{k}$, we multiply both sides by $\hat{\varphi}\hat{\psi}$ and we integrate with the measure $\dd\mathcal{L}_{2n}(\lambda,\lambda')\dd\mathcal{L}_{2m}(\mu,\mu')$. Proceeding this way we obtain (assuming without loss of generality that $m\geq n$) \begin{equation}
		\int_{L_{m}}\{F, G\}(u)\dd\mu_{L_{m}}(u) = \beta\int_{L_{m}}\langle \nabla F (u), X(u)\rangle G(u)\dd\mu_{L_{m}}(u),
		\end{equation} which is exactly the cylindrical KMS condition \ref{def: KMS cilindrici}.
				\end{proof}

\section{Some remarks on the classical KMS condition}\label{sec: appendice C}
A major difficulty when dealing with thermal equilibrium on Weyl algebras is related to the fact that the $^*$-automorphisms of interest are rarely strongly continuous.
This spoils the possibility to use the usual quantum KMS conditions formulated in terms of $C^*$-dynamical systems.
A similar problem in the classical setting  led us to formulate the classical weak KMS condition \ref{def: classical KMS condition} which can be verified directly on $\mathcal{W}(E,0)$.
In this appendix we formulate an \textit{analytic} classical KMS condition for von Neumann algebras, and we prove that the weak KMS-states are also KMS states in their GNS representations.

 We consider a separable, complex Hilbert space $\mathcal{H}$.
Let us take take $H:D(H)\subset \mathcal{H} \rightarrow \mathcal{H}$ a self-adjoint, positive operator, with $0$ outside of $H$'s point spectrum.   C
{\begin{assumption}\label{assumtpions on D}Consider a subspace $\mathcal{D}$ of $D(H)$ with the following properties:
\begin{itemize}
	\item [(a)] $\mathcal{D}$ is norm dense in $\mathcal{H}$;
	\item [(b)] $\mathcal{D}$ is composed of analytic vectors for $H$ and $e^{iHt}\mathcal{D}\subset\mathcal{D}$;
	\item [(c)] $\mathcal{D}\subset D(H^{-1})$.
\end{itemize} 
\end{assumption}}As a concrete example we consider \begin{equation}
	\mathcal{D} = \{ P^{(H)}((a,b])\psi,\;\psi \in \mathcal{H},\; 0<a<b<\infty\},
\end{equation} where $P^{(H)}(\cdot)$ is the PVM of the self-adjoint operator $H$.

We work with the classical Weyl $C^*$-algebra $\mathcal{W}(\mathcal{D},0)$.
For every $f\in\mathcal{D}$, we define the linear, continuous weak derivation $\delta_{0}(W^{0}(f)):= i\Phi_{0}(iH f)W^{0}(f)$. Then, we always have a $(\delta_{0},\beta)$-weak KMS state\begin{equation}\label{eq: generic gaussian KMS}
	\omega(W^{0}(f)) = e^{-\frac{1}{2\beta}\langle f, H^{-1} f\rangle}.
\end{equation}  
Moving to the GNS representation $(\pi_{\omega},\mathcal{H}_{\omega},\Omega_{\omega})$ associated with $\omega$, the estimate \begin{equation}
	\norm{(\pi_{\omega}(W^{0}(f)- W^{0}(g))\pi_{\omega}(W^{0}(h)))\Omega_{\omega}}_{\omega}\leq c_{h}(\norm{f}+\norm{g})\norm{f-g}_{H},\quad f,g,h\in \mathcal{D}
\end{equation} with $\norm{f}_{H}:= \langle f, H^{-1}f\rangle^{1/2}$, follows by a simple extension of \cite[Prop. 5.2.29, (2)]{Bratteli_Robinson_97}. In particular, it is clear that the group of $^*$-automorphism defined on $\pi_{\omega}(\mathcal{W}(\mathcal{D},0))$ as $\tau^{\omega}_{t}(\pi_{\omega}(A)):= \pi_{\omega}(\tau_{t}(A))$ is $\sigma$-weakly continuous since{\begin{align}
& \lim_{t\to 0}\norm{(\tau^\omega_t(\pi_\omega(W^0(f))-\pi_\omega(W^0(f)))\pi_\omega(W^0(h)))\Omega_\omega}_\omega\notag \\ 
	=&\lim_{t\to 0}\norm{(\pi_\omega(W^0(e^{iHt}f)-W^0(f))\pi_\omega(W^0(h)))\Omega_\omega}_\omega\notag \\ \leq &\lim_{t\rightarrow 0}2c_h \norm{f}\norm{e^{iHt}f -f}_H = 0,\quad f\in\mathcal{D}.
\end{align}} Thus, $\tau^{\omega}$ extends $\sigma$-weak continuously to a $^*$-automorphism on the von-Neumann algebra $\pi_{\omega}(\mathcal{W}(\mathcal{D},0))''$. 

{Now, we wish to find a $\sigma$-weakly dense subspace of $\pi_\omega(\mathcal{W}(\mathcal{D},0))''$ where it is possible to define a \textit{Poisson bracket} and such that $\pi_\omega(\mathcal{W}(\mathcal{D},0))''\subset D(\delta_\omega)$, where $\delta_\omega$ is the \textit{true} derivation associated with continuous automorphism $\tau_\omega$. Let us denote by $\mathcal{F}$ the Fourier transform operator. We consider the set of analytic functions \begin{equation}
	\mathcal{A}_n:=\{\chi: \mathbb{R}^{n}\rightarrow \mathbb{R},\;\chi(t) = \mathcal{F}(\psi)(t), \text{ for }\psi\in C_c^{\infty}(\mathbb{R}^n)\},\quad n\in\mathbb{N},\; n>0. 
\end{equation} Note that $\mathcal{A}_n$ is closed under products, differentiation, and if $\chi_1\in \mathcal{A}_n,\;\chi_2\in \mathcal{A}_m$, then \begin{equation}(t_1,\dots,t_{n+m})\to\chi_1(t_1,\dots, t_n)\chi_2(t_{n+1,\dots ,t_{n+m}}) \in \mathcal{A}_{n+m}.\end{equation} Moreover, by the \textit{Paley-Wiener theorem} \cite[Th IX.11]{Reed_SimonII}, the analytic extension  $\chi\in \mathcal{A}_n$ satisfies \begin{equation}\label{eq: Paley-Wiener}
    |\chi(\zeta)| \leq \frac{C_N e^{R |\Im{\zeta}|}}{(1+|\zeta|)^N}, \quad \text{for all }  N\in\mathbb{N},  
\end{equation} where $R$ is some fixed constant (Eq. \eqref{eq: Paley-Wiener} is also sufficient for having $\chi\in \mathcal{A}_n$).  With this set, we can define a $^*$-subalgebra of $\pi_\omega(\mathcal{W}(\mathcal{D},0))''$ by taking \begin{equation}\label{eq: analytic set of weyl}
	\mathcal{W}(\mathcal{A}):=\text{LH}\{ \int_{\mathbb{R}^n}\pi_{\omega}(W^{0}(\sum_{i=1}^{n}e^{iHt_{i}}f_{i}))\chi(t)\dd^{n}t,\; \{f_{k}\}_{k\in\{1,\dots, n\}}\subset\mathcal{D},\;\chi\in \mathcal{A}_n, \;n\in\mathbb{N}\},
\end{equation} {where the integrals in \eqref{eq: analytic set of weyl} are defined  in the $\sigma$-weak topology. The elements in $\mathcal{W}(\mathcal{A})$ are all analytic for $\tau^{\omega}$ and form a $\sigma$-weakly dense set in $\pi_{\omega}(\mathcal{W}(\mathcal{D},0))''$ \cite[Prop. 2.5.22]{Bratteli_Robinson_87}}. A Poisson bracket can be defined by (the definition on all the elements of $\mathcal{W}(\mathcal{A})$ follows naturally)\begin{multline}
	\{\int^{\infty}_{0}\tau^{\omega}_{t}(\pi_{\omega}(W^{0}(f)))\chi_{1}(t)\dd t,\int^{\infty}_{0}\tau^{\omega}_{s}(\pi_{\omega}(W^{0}(g)))\chi_{2}(s)\dd s\}\\:= \int \pi_{\omega}(W^{0}(e^{iHt}f+e^{iHs}g))\chi_{1}(t)\chi_{2}(s)\sigma(e^{iHs}g,e^{iHt}f)\dd t\dd s,
\end{multline} Note that the function $\mathbb{C}^2\ni(z,w)\to \sigma(e^{iHz}g,e^{iHt}w)\in\mathbb{R}$ is analytic thanks the regularity properties of $f,g \in\mathcal{D}$. Then, $f = P^{(H)}((a,b])f,\; g = P^{(H)}(a,b])g$ for some common $0<a<b$ and\begin{equation}
\sup_{\substack{z,w\in\mathbb{C}^2\\ |\Im{z}|,|\Im{w}|\leq \alpha}}|\partial_{w}^{n}\partial_{z}^{m}\sigma(e^{iHz}g, e^{iHw}f)|\leq b^{n+m}e^{2\alpha b} < +\infty. \end{equation}
 Hence, $\mathbb{R}^2\ni (t,s)\to \chi_1(t)\chi_2(s)\sigma (e^{iHs)}g,e^{iHt}g)\in \mathcal{A}_2$ \cite[Th IX.11]{Reed_SimonII} and the Poisson brackets leave $\mathcal{W}(\mathcal{A})$ invariant.}

{
 The derivation $\delta_\omega$ is well defined on $\mathcal{W}(\mathcal{A})$ as the infinitesimal time limit, in the $\sigma$-weak topology, of $\tau^{\omega}$:\begin{align}
	&\delta^{\omega}_{0}(\int^{\infty}_{0}\pi_{\omega}(W^{0}(\sum_{i=1}^ne^{iHt_i}f_i))\chi(t)\dd^n t) = \lim_{s\to 0}\int^{\infty}_{0}\frac{1}{s}(\tau^\omega_s- I)(\pi_{\omega}(W^{0}(\sum_{i=1}^ne^{iHt_i}f_i)))\chi(t)\dd^n t\notag \\=&-   \lim_{s\to 0}\int^{\infty}_{0}\pi_{\omega}(W^{0}(\sum_{i=1}^ne^{iHt_i}f_i))\frac{\chi(t_1+s,\dots,t_n+s)-\chi(t_1,\dots,t_n)}{s}\dd^n t \notag \\ =& -\int^{\infty}_{0}\pi_{\omega}(W^{0}(\sum_{i=1}^ne^{iHt_i}f_i))\chi'(t)\dd^n t \in \pi_{\omega}(\mathcal{W}(\mathcal{D},0))'',
\end{align} where for $\chi \in \mathcal{A}_n$ we have defined \begin{equation}\chi'(t):= \frac{\partial}{\partial_s}\chi(t_1+s,\dots, t_n+s)|_{s = 0}.\end{equation}}

{Now, we formulate a definition of classical KMS states, closer to the usual non-commutative one for von Neumann algebras. Note that, if $\mathfrak{M}$ is a commutative von Neumann algebra, then it can be regarded as a function space, so that it does make sense to speak about Poisson $ ^*$-subalgebras and Poisson brackets. In particular, if $\mathfrak{A}$ is a $C^*$-algebra with a given Poisson structure $(\dot{\mathfrak{A}},\{\;,\;\})$ \cite[Sec. 1]{Drago_Pettinari_VandeVen_2024}, then, for every representation $(\pi,\mathcal{H}_\pi)$ , $\pi(\mathfrak{A})''$ is a von Neumann algebra with Poisson structure given by $(\pi(\dot{\mathfrak{A}}),\{\;,\;\}_\pi)$, \begin{equation}
    \{\pi(a), \pi(b)\}_\pi := \pi(\{a,b\}),\quad a,b\in \dot{\mathfrak{A}}.
\end{equation}} \begin{definition}(\textbf{KMS state for von Neumann algebras})\label{def: classical strong KMS state}
	Let $(\mathfrak{M},\tau)$ be a commutative $W^*$-dynamical system endowed with a Poisson structure $(\dot{\mathfrak{M}},\{\;,\;\})$, where $\dot{\mathfrak{M}}$ is a $\sigma$-weakly dense, $\tau$-invariant Poisson $^*$-subalgebra  of $\mathfrak{M}_\tau$, the $\tau$-analytic elements of $\mathfrak{M}$. Let  $\omega$ be a normal state, then, $\omega$ is a $(\delta,\beta)$-classical KMS state if the following identity is satisfied\begin{equation}
		\beta\omega(a\delta(b)) = \omega(\{b,a\}),\quad a,b\in\dot{\mathfrak{M}},
	\end{equation} where $\delta :\mathfrak{M}_\tau\rightarrow \mathfrak{M}$ is the derivation obtained from $\tau$ as\begin{equation}\label{eq: strong derivation}
		\delta(a) = \sigma\text{-weak-}\lim_{t\rightarrow 0}\frac{\tau_{t}(a)-a}{t},\quad a\in\mathfrak{M}_\tau.
	\end{equation}
\end{definition}
{The quantum KMS condition for von Neumann algebras requires that $\tau$ be a $\sigma$-weakly continuous one-parameter group of $^*$-automorphisms, and that the identity  
\begin{equation}
    \omega(A \, \tau_{i\beta}(B)) = \omega(BA)
\end{equation}
hold for all $A,B$ in the $\sigma$-weakly dense domain of $\tau$-analytic elements $\mathfrak{M}_\tau$ \cite[Def.~5.3.1]{Bratteli_Robinson_97}. Since distinct equilibrium states on a $C^*$-algebra correspond to inequivalent GNS representations, it would be preferable to formulate the equilibrium condition directly at the level of the $C^*$-algebra $\mathcal{W}(E,h\sigma)$. However, an analytic $C^*$-KMS condition is not meaningful in this setting, as automorphism groups are rarely strongly continuous \cite{Fannes_Verbeure_(74)}. An analogous obstruction arises for classical Weyl algebras: one cannot define derivations there (see Sec.~\ref{sec: classical KMS condition}), whereas in Def.~\ref{def: classical strong KMS state} the operator $\delta$ is a genuine derivation, defined on a $\sigma$-weakly dense domain of $\pi_\omega(W(E,0))''$. Nevertheless, the next proposition shows that weak KMS states automatically satisfy the requirements of Def.\ref{def: classical KMS condition}.}

\begin{proposition}
	Let $H\colon D(H)\subset \mathcal{H}\rightarrow \mathcal{H}$ be a self-adjoint, positive operator with $0$ outside of its point spectrum\footnote{{These conditions are not used in the proof, but they are necessary for finding a subspace $\mathcal{D}$ as above. Positivity can always be obtained for an Hamiltonian bounded from below by subtracting a negative chemical potential.}}  and $\mathcal{D}$ a subspace of  satisfying Assumptions \ref{assumtpions on D}. Then, consider the weak derivation $\delta_{0}\colon \mathcal{W}(\mathcal{D},0)\rightarrow C(\mathcal{D}'_{\norm{\cdot}})$, $\delta_{0}(W^{0}(f)) = i\Phi_{0}(iHf)W^{0}(f)$ and the derivation $\delta^{\omega}_{0} \colon \mathcal{W}(\mathcal{A})\rightarrow \pi_{\omega}(\mathcal{W}(\mathcal{D},0))''$ defined by following the prescription in equation \eqref{eq: strong derivation}. If $\omega$ is a $(\delta_{0},\beta)$-weak KMS state on $\mathcal{W}(\mathcal{D},0)$, then it is a $(\delta^{\omega}_{0},\beta)$-KMS state.
\end{proposition}
\begin{proof}
This is just a simple computations based on the properties we have listed in the previous paragraphs: \begin{align}
	& \omega\bigg{(}\{ \int_{0}^{\infty}\pi_{\omega}(W^{0}(\sum_{i=1}^{n}e^{iHt_{i}}f_{i}))\chi_{1}(t)\dd^{n}t, \int_{0}^{\infty}\pi_{\omega}(W^{0}(\sum_{j=1}^{n}e^{iHs_{j}}g_{j}))\chi_{2}(s)\dd^{n}s\}\bigg{)} \nonumber \\ =&\int_{0}^{\infty}\omega(W^{0}(\sum_{i=1}^{n}e^{iHt_{i}}f_{i}+e^{iHs_{i}}g_{i}))\sum_{i,j=1}^{n}\sigma(e^{iH t_{i}}f_{i},e^{iHs_{j}}g_{j})\chi_{1}(t)\chi_{2}(s)\dd^{n}t\dd^{n}s \nonumber\\
	=& \beta\int_{0}^{\infty}\omega(i\Phi_{0}(\sum_{k=1}^{n}iHe^{iHt_{k}}f_{k})W^{0}(\sum_{i=1}^{n}e^{iHt_{i}}f_{i}+e^{iHs_{i}}g_{i}))\chi_{1}(t)\chi_{2}(s)\dd^{n}t\dd^{n}s\nonumber\\
	=&-\beta\omega \bigg{(}\int^{\infty}_{0} \pi_{\omega}(W^{0}(\sum^{n}_{i=1}e^{iHt_{i}}f_{i}+e^{iH_{i}s_{i}}g_{i} ))\frac{\dd}{\dd\tau}\chi_{1}(t_{1}+\tau,\dots,t_{n}+\tau)|_{\tau = 0}\chi_{2}(s)\dd^{n}t\dd^{m}s\bigg{)}\nonumber\\
	=&\beta \omega\bigg{(}\int_{0}^{\infty}\pi_{\omega}(W^{0}(\sum_{j=1}^{n}e^{iHs_{j}}g_{j}))\chi_{2}(s)\dd^{n}s\delta_{0}^{\omega}( \int_{0}^{\infty}\pi_{\omega}(W^{0}(\sum_{i=1}^{n}e^{iHt_{i}}f_{i}))\chi_{1}(t)\dd^{n}t)\bigg{)}.
\end{align}
\end{proof}

\end{document}